\documentclass[notitlepage,a4paper,twoside,leqno,10pt,nofootinbib,showkeys]{revtex4-1}

\usepackage[paperwidth=170mm,paperheight=241mm,centering,hmargin=2cm,vmargin=2.5cm]{geometry}
\usepackage{hyperref}
\hypersetup{colorlinks=true,allcolors=[rgb]{0,0,0.6}}
\usepackage[utf8]{inputenc}
\usepackage[usenames,dvipsnames]{xcolor}
\usepackage{calc}
\usepackage[english]{babel}
\usepackage{amssymb,amsmath,amsfonts,amsthm}
\usepackage{mathrsfs}
\usepackage{dsfont}
\usepackage{enumitem}
\usepackage[shadow,textwidth=12mm,textsize=tiny,obeyDraft]{todonotes}
\usepackage[varg]{txfonts}
\usepackage{microtype}
\usepackage{makecell}
\usepackage[normalem]{ulem}
\DeclareMathOperator{\Tr}{Tr}

\DeclareMathOperator{\ran}{Ran}

\renewcommand{\Im}{\mathrm{Im}}
\renewcommand{\Re}{\mathrm{Re}}

\theoremstyle{plain}
\newtheorem{thm}{Theorem}[section]
\newtheorem*{thm*}{Theorem}
\newtheorem{proposition}[thm]{Proposition}
\newtheorem{lemma}[thm]{Lemma}
\newtheorem{corollary}[thm]{Corollary}
\theoremstyle{definition}
\newtheorem{definition}[thm]{Definition}
\newtheorem*{definition*}{Definition}

\newtheorem{innerassumption}{Assumption}
\newenvironment{assumption}[1]{\renewcommand\theinnerassumption{#1}\innerassumption}{\endinnerassumption}

\theoremstyle{remark}

\newtheorem{remark}[thm]{Remark}
\newtheorem*{remark*}{Remark}
\renewcommand\thesection{\Roman{section}}

\makeatletter
\renewcommand\p@subsection{\thesection .}
\makeatother

\renewcommand{\qedsymbol}{$\blacksquare$}
\renewcommand{\emptyset}{\O}
\definecolor{myblue}{rgb}{0,0,0.6}
\makeatletter
\renewcommand\frontmatter@abstractwidth{\dimexpr 12.5cm \relax}
\makeatother
\begin{document}
\bibliographystyle{abbrvnat}
\title{Bohr’s correspondence principle in quantum field theory and
  classical renormalization scheme: the Nelson model.}
\date{\today} \author{Zied Ammari} \email{zied.ammari@univ-rennes1.fr}
\affiliation{IRMAR and Centre Henri Lebesgue; Université de Rennes I\\
  Campus de Beaulieu, 263 avenue du Général Leclerc\\ CS 74205, 35042
  RENNES Cedex} \author{Marco Falconi}\email{marco.falconi@mathematik.uni-stuttgart.de}
\thanks{the author has been partially supported by the Centre Henri
  Lebesgue (programme ``Investissements d'avenir'' ---
  ANR-11-LABX-0020-01).} \affiliation{Institut für Analysis, Dynamik
  und Modellierung\\ Universität Stuttgart; Pfaffenwaldring 57 D-70569
  Stuttgart}

\begin{abstract}
  \begin{description}
  \item[Abstract] In the mid Sixties Edward Nelson proved the
    existence of a consistent quantum field theory that describes the
    Yukawa-like interaction of a non-relativistic nucleon field with a
    relativistic meson field. Since then it is thought, despite the
    renormalization procedure involved in the construction, that the
    quantum dynamics should be governed in the classical limit by a
    Schr\"odinger-Klein-Gordon system with Yukawa coupling.  In the
    present paper we prove this fact in the form of a Bohr
    correspondence principle. Besides, our result enlighten the nature
    of the renormalization method employed in this model which we
    interpret as a strategy that allows to put the related classical
    Hamiltonian PDE in a normal form suitable for a canonical
    quantization.
  \end{description}
\end{abstract}
\keywords{Nelson model, Renormalization, Schrödinger-Klein-Gordon system, Yukawa interaction, Classical limit.\\
  \href{http://www.ams.org/msc/msc2010.html}{MSC2010}:
  \href{http://www.ams.org/msc/msc2010.html?t=&s=35Q40&btn=Search&ls=s}{35Q40},
  \href{http://www.ams.org/msc/msc2010.html?t=&s=81T16&btn=Search&ls=s}{81T16},
  \href{http://www.ams.org/msc/msc2010.html?t=&s=81Q20&btn=Search&ls=s}{81Q20}}
\maketitle
\tableofcontents
\hyphenation{S-KG}

\section{Introduction.}
\label{sec:introduction}

Modern theoretical physics explains how  matter interacts with radiation and proposes
phenomenological models of quantum field theory that in principle describe such fundamental interaction. Giving a firm mathematical ground to these models is known to be a difficult task related to renormalization theory \citep{MR2414874,MR1380265,Fol2,MR947959,1969LNP.....2.....H,GJ-book}. Since the fifties there were spectacular advances in these problems  culminating with the perturbative renormalization of quantum electrodynamics, the birth of the renormalization group method and the renormalizability of gauge field theories.
Nevertheless, conceptual mathematical difficulties remain as well as outstanding open problems,
see \cite{MR1773042,MR1768636}. The purpose of the present article is twofold: first to prove the quantum-classical correspondence for the renormalized Nelson model and second to point out
a conceptually different point of view on renormalization  through the analysis of an
elementary example of quantum field theory and the study of the relationship
with its classical  formulation.

The so-called Nelson model is a system of Quantum Field Theory that
has been widely studied from a mathematical standpoint
\citep[see e.g.][]{MR3201223,Ar01,GHPS,MR1943190,Mol,GGM1,AH12,BFS,BFS1,BFS2,BFLMS,DG1,GGM1,FGS3,Pi,Sp3}.
It consists of non-relativistic spin zero particles interacting
with a scalar boson field, and can be used to model various systems of
physical interest, such as nucleons interacting with a meson
field. In the mid sixties Edward Nelson rigorously constructed a quantum
dynamic for this model free of ultraviolet (high energy) cutoffs in
the particle-field coupling, see \citep{nelson:1190}.
This is done by means of a renormalization procedure: roughly speaking, we
need to subtract a divergent quantity from the Hamiltonian, so the
latter can be defined as a self-adjoint operator in the limit of the ultraviolet cutoff.
The quantum dynamics is rather singular in this case (renormalization is
necessary); and the resulting generator has no explicit form as an
operator though it is unitarily equivalent to an explicit one.
Since the work of Gross \citep{EugeneP1962219} and Nelson \citep{nelson:1190}
it is believed, but never proved, that the renormalized dynamics is generated by
a canonical quantization of the Schrödinger-Klein-Gordon (S-KG) system with Yukawa coupling.
In other words, the quantum fluctuations of the particle-field system are centered around
the classical trajectories of the Schrödinger-Klein-Gordon system at certain scale
and the renormalization procedure preserves the suitable quantum-classical correspondence
as well as being necessary to define the quantum dynamics. We give a mathematical formulation
of such result in Theorem \ref{thm:2} in the form of a Bohr
correspondence principle.
The proof of the above theorem is rather technical and constitutes a large part of this
article.

Recently, the authors of this paper have studied the classical limit
of the Nelson model, in its regularized
version~\citep{Ammari:2014aa,falconi:012303}. We have proved that the
quantum dynamic converges, when an effective semiclassical parameter
$\varepsilon \to 0$, towards a non-linear Hamiltonian flow on a
classical phase space. This flow is governed by a
Schrödinger-Klein-Gordon system, with a regularized Yukawa-type
coupling. To extend the classical-quantum correspondence to the system
without ultraviolet cutoff, we rely on the recent techniques
elaborated in the mean-field approximation of many-body Schrödinger
dynamics in
\citep{ammari:nier:2008,MR2513969,MR2802894,2011arXiv1111.5918A} as
well as the result with cutoff \citep{Ammari:2014aa}.  As a matter of
fact the renormalization procedure, implemented by a dressing
transform, generates a many-body Schrödinger dynamics in a mean-field
scaling.  So it was convenient that the mean-field approximation was
derived with the same general techniques that allow to prove equally
the classical approximation for QFT models.  The result is further
discussed in Subsection~\ref{sec:class-limit-renorm}, and all the
details and proofs are provided in
Section~\ref{sec:class-limit-renorm-1}.

As mentioned before, the study of the quantum-classical correspondence
led us to a reinterpretation of renormalization, for the Nelson model
and maybe in general, as a procedure that allows to put the related
classical Hamiltonian PDE in a normal form suitable for a canonical
quantization.  This normal form is implemented by a near identity
change of coordinates close in spirit to the work of Jalal Shatah on
nonlinear Klein-Gordon equation with quadratic terms, see
\cite{MR803256}. Usually renormalization in quantum electro-dynamics
and QFT is addressed through the renormalization group method, see
e.g. \cite{MR2414874,MR579493}.  In the eighties and nineties two deep
ideas have found fruitful applications beyond their birthplace: the
normal form techniques are extended to PDEs and the renormalization
group method is adapted to ODEs and PDEs after the influential works
of Shatah \cite{MR803256}, Goldenfeld and al.~\cite{MR1289625}
respectively (see also Bricmont and al.~\cite{MR1280993}).  These
ideas resulted in a strong research activity, see e.g.
\cite{MR1862435,MR1755506,MR2737489,MR1127050,MR1222944,MR1341412}. The connection between the normal form and the renormalization group
has been highlighted for ODEs in \cite{MR2450492}. So this suggests
that \emph{renormalization can be carried out through normal forms
  prior to quantization in a Hamiltonian framework} rather than using
a renormalization group method and dealing with Lagrangians,
actions and Feynman amplitudes. Actually such point of
view may be linked in some sense to long standing ideas of Paul Dirac and Hendrik
Kramers on "classical" renormalization of quantum electrodynamics, see
\cite{MR895358}.  Moreover, both normal form and group renormalization
methods seem to be related techniques that deal with the problem of
reparametrizing  the dependence of a physical theory with respect to
its high/low energy scales (or fast/slow motion).

The proposed approach (if successful) is in our opinion more
mathematically founded due to the long history of normal form theory
for ODEs and PDEs. It relies less in physical considerations and
intuitions that have surrounded the subject of renormalization with an
impenetrable air of mystery. Although our suggestion is based on a
very simple model (no mass or charge renormalization is needed) we
hope that it could be extended in the future to more realistic quantum
field theories and it will be of interest for other colleagues.  A
further discussion of this point is continued in
Subsection~\ref{sec:an-altern-renorm}, and more details are provided
in Section~\ref{sec:class-renorm-nels} for the Nelson model.

For the sake of presentation, we collected the
notations and basic definitions---used throughout the paper---in the
Subsection~\ref{sec:notat-defin} below.  In subsection \ref{sec:class-limit-renorm}
 we  present our main result on the classical-quantum correspondence principle and in Subsection \ref{sec:an-altern-renorm}  we discuss our alternative point of view on renormalization.
The rest of the paper is organized as follows: in Section~\ref{sec:nelson-hamiltonian} we review the basic properties of the quantum system and the usual procedure of
renormalization; in Section~\ref{sec:classical-system} we analyze the
classical S-KG dynamics and the classical dressing transformation; in
Section~\ref{sec:class-renorm-nels} we apply the renormalization via
classical dressing to the Nelson model; in
Section~\ref{sec:class-limit-renorm-1} we study in detail the
classical limit of the renormalized Nelson model, and prove our main
Theorem~\ref{thm:2}.

\subsection{Notations and general definitions.}
\label{sec:notat-defin}

\begin{itemize}[label=\color{myblue}*]
\item We fix once and for all
  $0<\bar{\varepsilon},m_0,M$. We also define the
  function $\omega(k)=\sqrt{k^2+m^2_0}$.
\item The effective (semiclassical) parameter will be denoted by $\varepsilon\in(0,\bar \varepsilon)$.
\item Let $\mathcal{Z }$ be a Hilbert space; then we denote by
  $\Gamma_s(\mathcal{Z })$ the symmetric Fock space over
  $\mathcal{Z}$. We have that
  \begin{equation*}
    \Gamma_s(\mathcal{Z})=\bigoplus_{n=0}^{\infty}\mathcal{Z}^{\otimes_s n}\; \text{ with }\; \mathcal{Z}^{\otimes_s 0}=\mathds{C}\; .
  \end{equation*}
\item Let $X$ be an operator on a Hilbert space $\mathcal{Z }$. We
  will usually denote by $D(X)\subset \mathcal{Z }$ its domain of
  definition, and by $Q(X)\subset \mathcal{Z }$ the domain of
  definition of the corresponding quadratic form.
\item Let $S:\mathcal{Z }\supseteq D(S)\to \mathcal{Z }$ be a densely
  defined self-adjoint operator on $\mathcal{Z }$. Its second
  quantization $d\Gamma(S)$ is the self-adjoint operator on
  $\Gamma_s(\mathcal{Z })$ defined by
  \begin{equation*}
    d\Gamma(S)\rvert_{D(S)^{\otimes_s n}}=\varepsilon\sum_{k=1}^n 1\otimes\dotsm\otimes \underbrace{S}_{k}\otimes \dotsm\otimes 1\; .
  \end{equation*}
   The operator $d\Gamma (1)$ is
  usually named the number operator and denoted by $N$.
  \item We denote by $\mathcal{C }^{\infty}_0(N)$ the subspace of finite
  particle vectors:
  \begin{equation*}
    \mathcal{C}^{\infty}_0(N)=\{\psi\in \Gamma_s(\mathcal{Z})\; ;\; \exists \bar{n}\in\mathds{N}, \psi\bigr\rvert_{\mathcal{Z}^{\otimes _s n}}=0 \;\forall n> \bar{n}\}\; .
  \end{equation*}
  \item Let $U$ be a unitary operator on $\mathcal{Z }$. We define
  $\Gamma(U)$ to be the unitary operator on $\Gamma_s(\mathcal{Z })$
  given by
  \begin{equation*}
    \Gamma(U)\rvert_{\mathcal{Z}^{\otimes_s n}}=\bigotimes_{k=1}^n U\; .
  \end{equation*}
  If $U=e^{it S}$ is a one parameter group of unitary operators on $\mathcal{Z }$,
  $\Gamma(e^{it S})=e^{i\frac{t}{\varepsilon} d\Gamma(S)}$.
\item On $\Gamma_s(\mathcal{Z })$, we define the annihilation/creation
  operators $a^{\#}(g)$, $g\in \mathcal{Z }$, by their action on
  $f^{\otimes n}\in \mathcal{Z }^{\otimes _s n}$ (with $a(g)f_0=0$ for
  any $f_0\in \mathcal{Z }^{\otimes_s 0}=\mathds{C}$):
  \begin{align*}
    a(g)f^{\otimes n}&=\sqrt{\varepsilon n} \; \langle g , f\rangle_{\mathcal{Z}} \; f^{\otimes (n-1)}\; ;\\
    a^{*}(g)f^{\otimes n}&=\sqrt{\varepsilon (n+1)} \; g\otimes_s f^{\otimes n}\; .
  \end{align*}
   They satisfy the Canonical Commutation Relations (CCR),
  $[a(f),a^{*}(g)]=\varepsilon \langle f , g \rangle_{\mathcal{Z }} $.

  If $\mathcal{Z }=L^2 (\mathds{R}^d)$ it is useful to introduce the
  operator valued distributions $a^{\#}(x)$ defined by
  \begin{equation*}
    a(g)=\int_{\mathds{R}^d}^{}\bar{g}(x)a(x)  dx\quad ,\quad a^{*}(g)=\int_{\mathds{R}^d}^{}g(x)a^{*}(x)  dx\; .
  \end{equation*}
\item
  $\mathcal{H}=\Gamma_s(L^2 (\mathds{R}^3)\oplus L^2
  (\mathds{R}^3))\simeq\Gamma_s(L^2(\mathds{R}^3))\otimes \Gamma_s(L^2
  (\mathds{R}^3))$.
  We denote by $\psi^{\#}(x)$ and $N_1$ the annihilation/creation and
  number operators corresponding to the nucleons (conventionally taken
  to be the first Fock space), by $a^{\#}(k)$ and $N_2$ the
  annihilation/creation and number operators corresponding to the
  meson scalar field (second Fock space). In particular, we will always use  the following
  $\varepsilon$-dependent representation of the CCR if not specified
  otherwise:
  \begin{equation*}
    [\psi(x),\psi^{*}(x')]=\varepsilon\delta(x-x')\quad ,\quad  [a(k),a^{*}(k')]=\varepsilon\delta(k-k')\; .
  \end{equation*}
 \item  We will sometimes use  the  following decomposition:
  \begin{equation*}
    \mathcal{H}=\bigoplus_{n=0}^{\infty}\mathcal{H}_n\text{, with } \mathcal{H}_n=\bigl(L^2(\mathds{R}^3)\bigr)^{\otimes_s n}\otimes \Gamma_s(L^2 (\mathds{R}^3))\; .
  \end{equation*}
  We denote by $T^{(n)}:=T\bigr\rvert _{\mathcal{H }_n}$ the
  restriction to $\mathcal{H }_n$ of any operator $T$ on $\mathcal{H}$
\item On $\mathcal{H}$, the Segal quantization of
  $L^2 (\mathds{R}^3 )\oplus L^2 (\mathds{R}^3 ) \ni \xi =\xi_1\oplus
  \xi _2 $
  is given by
  $R(\xi )=\bigl(\psi^{*} (\xi _1)+\psi (\xi _1)+a^{*}(\xi _2)+a(\xi
  _2)\bigr)/\sqrt{2}$,
  and therefore the Weyl operator becomes
  $W(\xi )=e^{\frac{i}{\sqrt{2}} \bigl(\psi^{*} (\xi _1)+\psi (\xi
    _1)\bigr)}\, e^{\frac{i}{\sqrt{2}} \bigl(a^{*}(\xi _2)+a(\xi
    _2)\bigr) }$.
\item Given a Hilbert space $\mathcal{Z }$, we denote by
  $\mathcal{L}(\mathcal{Z})$ the $C^{*}$-algebra of bounded operators;
  by $\mathcal{K}(\mathcal{Z })\subset \mathcal{L}(\mathcal{Z})$ the
  $C^{*}$-algebra of compact operators; and by
  $\mathcal{L}^1(\mathcal{Z})\subset \mathcal{K}(\mathcal{Z })$ the
  trace-class ideal.
\item We denote classical Hamiltonian flows by boldface capital letters
  (e.g. $\mathbf{E}(\cdot )$); their corresponding energy functional by
  script capital letters (e.g. $\mathscr{E}$).
\item Let $f\in \mathscr{S}'(\mathds{R}^d)$. We denote by
  $\mathcal{F }(f)(k)$ its Fourier transform
  \begin{equation*}
    \mathcal{F}(f)(k)=\frac{1}{(2\pi)^{d/2}} \int_{\mathds{R}^d}^{}f(x)e^{-ik\cdot x}  dx\; .
  \end{equation*}
\item We denote by $\mathcal{C}^{\infty }_0(\mathds{R}^d )$ the
  infinitely differentiable functions of compact support. We denote by
  $H^s(\mathds{R}^d)$ the non-homogeneous Sobolev space:
  \begin{equation*}
    H^s(\mathds{R}^d)=\Bigl\{f\in \mathscr{S}'(\mathds{R}^d)\; ,\; \int_{\mathds{R}^d}^{}(1+\lvert k  \rvert_{}^2)^s\lvert \mathcal{F}(f)(k)  \rvert_{}^2  dk<+\infty \Bigr\}\; ;
  \end{equation*}
  and its ``Fourier transform''
  \begin{equation*}
    \mathcal{F}H^s(\mathds{R}^d)=\Bigl\{f\,,\, \mathcal{F}^{-1}f\in H^s(\mathds{R}^d) \Bigr\}\; .
  \end{equation*}
  \item Let $\mathcal{Z}$ be a Hilbert space. We denote by
  $\mathfrak{P}(\mathcal{Z})$ the set of Borel probability measures on
  $\mathcal{Z}$.
\end{itemize}

\subsection{The classical limit of the renormalized Nelson model.}
\label{sec:class-limit-renorm}

The Schrödinger-Klein-Gordon equations  with
Yukawa-like coupling is a widely studied system of non-linear PDEs in three dimension~\citep[see
e.g.][]{MR0390555
  ,MR515899
  ,MR550215
  ,BaCha
  ,MR778979
  ,MR1915302
  ,MR2403699
  ,MR2906550
}. This system can be written as:
\begin{equation*}
  \left\{
    \begin{aligned}
      &i\partial_t u=-\frac{\Delta }{2M}u+ Vu+ A u\\
      &(\square+m^2_0)A=- \lvert u \rvert^2
    \end{aligned}
  \right .\; ;
\end{equation*}
where $m_0,M>0$ are real parameters and $V$ is a non-negative  potential
that is confining or equal to zero. Using the complex field $\alpha $
as a dynamical variable instead of $A$ (see Equation~\eqref{eq:71} of
Section~\ref{sec:classical-system}), the aforementioned dynamics can
be seen as a Hamilton equation generated by the following
energy functional, densely defined on\footnote{Sometimes
   the shorthand notation $L^2\oplus L^2$ is used instead of
  $L^2 (\mathds{R}^3 )\oplus L^2 (\mathds{R}^3 )$, if no confusion
  arises.}  $L^2 \oplus L^2 $:
\begin{equation*}
  \mathscr{E}(u, \alpha ):= \Bigl\langle u  , \Bigl(-\tfrac{\Delta}{2M}+V\Bigr)u \Bigr\rangle_2+\langle \alpha   , \omega \alpha  \rangle_2+\tfrac{1}{(2\pi )^{3/2}}\int_{\mathds{R}^6}^{}\tfrac{1}{\sqrt{2\omega (k)}}\Bigl(\bar{\alpha}(k)e^{-ik\cdot x}+ \alpha (k)e^{ik\cdot x}\Bigr)\lvert u(x)  \rvert_{}^2  dxdk\; .
\end{equation*}
With suitable assumptions on the external potential $V$, one
proves the global existence of the associated flow $\mathbf{E}(t)$;
a detailed discussion can be found in
Subsection~\ref{sec:glob-exist-results} where the precise condition on $V$
is given by Assumption~\eqref{ass:2}. So there is a Hilbert
space\footnote{
  If $V=0$, $\mathcal{D}$ could be the whole space
  $L^2 \oplus L^2 $.}  $\mathcal{D}=Q(-\Delta +V)\oplus
  \mathcal{F}H^{\frac{1}{2}}(\mathds{R}^3)$ densely imbedded in $L^2 \oplus L^2 $ such
that there exists a classical flow
$\mathbf{E}:\mathds{R}\times \mathcal{D}\to \mathcal{D}$ that solves the
Schr\"odinger-Klein-Gordon equation~\eqref{eq:72} written using the complex field $\alpha$.

A question of significant interest, both mathematically and physically, is
whether it is possible to quantize the Schr\"odinger-Klein-Gordon dynamics with Yukawa coupling as a consistent  theory that describes quantum mechanically the particle-field interaction. As mentioned previously, E.~Nelson rigorously constructed a self-adjoint operator satisfying in some sense the above requirement. Afterward the model is proved to satisfy some of the main properties that are familiar in the axiomatic approach to quantum fields, see \cite{MR0293448}. Furthermore,
asymptotic completeness is established in  \cite{MR1809881}. The problem of quantization of such infinite dimensional nonlinear dynamics is related to constructive quantum field theory. The general framework is as follows.\\
Let $\mathcal{Z}$ be a complex Hilbert space with inner
product $\langle \, \cdot\, ,\, \cdot\, \rangle$. We define the
associated symplectic structure $\Sigma (\mathcal{Z})$ as the pair
$\{\mathcal{Y},B(\,\cdot\, ,\,\cdot\, )\}$ where $\mathcal{Y}$ is
$\mathcal{Z}$ considered as a real Hilbert space with inner product
$\langle \,\cdot\, ,\, \cdot\, \rangle_r=\Re\langle \,\cdot\, ,\,
\cdot\, \rangle_{}$,
and $B(\,\cdot\, ,\,\cdot\, )$ is the symplectic form defined by
$B(\,\cdot\, ,\,\cdot\, )=\Im\langle \,\cdot\, , \,\cdot\,
\rangle_{}$. Following
\citet{MR0112626
}, we define a (bosonic) quantization of the structure
$\Sigma (\mathcal{Z})$ any linear map $R(\cdot)$ from $\mathcal{Y}$ to
self-adjoint operators on a complex Hilbert space such that:
\begin{itemize}[label=\color{myblue}*]
\item The Weyl operator $W(z)=e^{iR(z)}$ is weakly continuous when
  restricted to any finite dimensional subspace of $\mathcal{Y}$;
\item $W(z_1)W(z_2)=e^{-\frac{i}{2} B(z_1,z_2)}W(z_1+z_2)$ for any
  $z_1,z_2\in \mathcal{Y}$ (Weyl's relations).
\end{itemize}
When the dimension of $\mathcal{Z}$ is not finite, there are
uncountably many irreducible unitarily inequivalent  Segal
quantizations of $\Sigma (\mathcal{Z})$ (or representation of Weyl's relations). A representation of
particular relevance in physics is the so-called Fock
representation~\citep{fock
  ,MR0044378
}  on the symmetric Fock space $\Gamma _s(\mathcal{Z})$. Once this representation
is considered there is a natural way to quantize  polynomial
functionals on $\mathcal{Z}$ into quadratic forms on
$\Gamma _s(\mathcal{Z})$ according to the Wick or normal order (we briefly
outline the essential features of Wick quantization on Section~\ref{sec:comm-hath_r-i}, the
reader may refer to
\citep{ammari:nier:2008,Ber,MR3060648
} for a more detailed presentation).

Following this rules, the formal quantization of the classical energy  $\mathscr{E}$
yields a quadratic form $h$ on the Fock space $\Gamma _s(L^2 \oplus L^2)$ which plays the role of a quantum energy. The difficulty now lies on the fact that the quadratic form $h$ do not define straightforwardly  a dynamical system (i.e.~$h$ is not related to a self-adjoint operator).
Nevertheless, according to the work of Nelson it is possible in our case to define for any $\sigma _0\in \mathds{R}_+$, a renormalized self-adjoint operator
$H_{ren}(\sigma _0)$ associated in some specific sense to $h$ (see
Sections~\ref{sec:nelson-hamiltonian} and~\ref{sec:class-renorm-nels} for details).
However, the relationship between the classical and the quantum theory at hand is obscured by the renormalization procedure and it is unclear even formally if the quantum dynamics
generated by $H_{ren}(\sigma _0)$ are still related to the original Schr\"odinger-Klein-Gordon  equation. Therefore, we believe that  it is mathematically  interesting to verify Bohr's correspondence principle for this model.
\begin{quote}\emph{Bohr's principle}:
  The quantum system should reproduce, in the limit of large quantum
  numbers, the classical behavior.
\end{quote}
This principle may be reformulated as follows. We
make the quantization procedure dependent on a effective parameter
$\varepsilon $, that would converge to zero in the limit.
The physical interpretation of $\varepsilon $ is of a
quantity of the same order of magnitude as the Planck's constant, that
becomes negligible when large energies and orbits are considered. In
the Fock representation, we introduce the $\varepsilon $-dependence in
the annihilation and creation operator valued distributions
$\psi ^{\#}(x)$ and $a^{\#}(k)$, whose commutation relations then
become $[\psi (x),\psi ^{*}(x')]=\varepsilon \delta (x-x')$ and
$[a(k),a^{*}(k')]=\varepsilon \delta (k-k')$. If in the limit
$\varepsilon \to 0$ the quantum unitary dynamics converges towards the
Hamiltonian flow generated by the Schr\"odinger-Klein-Gordon equation with Yukawa interaction, Bohr's principle is satisfied.

If the phase space $\mathcal{Z}$ is finite dimensional, the
quantum-classical correspondence has been proved in the context of semiclassical or microlocal
analysis, with the aid of pseudo-differential calculus, Wigner measures or coherent states
\citep[see
e.g.][]{MR818831
  ,Hor
  ,GMMP
  ,Ger
  ,HMR
  ,Mar
  ,Rob
  ,AFFGP
  ,LiPa
  ,MR0332046
,CRR}.
If $\mathcal{Z}$ is infinite dimensional, the situation is more
complicated, and there are fewer results for systems with unconserved number of particles
\citep{MR2205462
  ,MR3282640
  ,ammari:nier:2008
  ,Ammari:2014aa
  ,Frank:2014aa
  ,Frank:2015aa
}. The approach we adopt here makes use of the infinite-dimensional
Wigner measures introduced by
\citet{2011arXiv1111.5918A
  ,ammari:nier:2008
  ,MR2513969
  ,MR2802894
}.
Given a family of normal quantum states
$(\varrho_{\varepsilon } )_{\varepsilon \in (0,\bar{\varepsilon} )}$
on the Fock space, we say that a Borel probability measure $\mu $
on $\mathcal{Z}$ is a Wigner measure associated to it if there exists
a sequence
$(\varepsilon _k)_{k\in \mathds{N}}\subset (0,\bar{\varepsilon} )$
such that $\varepsilon _k\to 0$ and\footnote{$W(\xi )$ is the $\varepsilon_k$-dependent  Weyl
   operator explicitly defined by \eqref{eq:104}.}
\begin{equation}
  \label{eq:83}
  \lim_{k\to \infty }\Tr[\varrho _{\varepsilon _k}W(\xi )]=\int_{\mathcal{Z}}^{}e^{i \sqrt{2}\Re\langle \xi   , z \rangle_{\mathcal{Z}}}  d\mu (z) \; ,\; \forall \xi \in \mathcal{Z}\; .
\end{equation}
Wigner measures are related to phase-space analysis and are in general
an effective tool for the study of the classical limit. We denote by
$\mathcal{M}\bigl(\varrho _{\varepsilon }, \varepsilon \in
(0,\bar{\varepsilon} )\bigr)$
the set of Wigner measures associated to
$(\varrho_{\varepsilon } )_{\varepsilon \in (0,\bar{\varepsilon} )}$.
Let $e^{-i \frac{t}{\varepsilon } H_{ren}(\sigma _0)}$ be the quantum
dynamics on $\Gamma _s(\mathcal{Z})$, $\mathcal{Z}=L^2\oplus L^2$,
then the time-evolved quantum states can be written as
$(e^{- i\frac{t}{\varepsilon }H_{ren}(\sigma _0)}\varrho _{\varepsilon
}e^{i\frac{t}{\varepsilon } H_{ren}(\sigma _0)})_{\varepsilon \in
  (0,\bar{\varepsilon} )}$.
Bohr's principle is satisfied if Wigner measures of time evolved
quantum states are exactly the pushed forward, by the classical flow
$\mathbf{E}(t)$, of the initial Wigner measures at time $t=0$; i.e.
\begin{equation}
  \label{eq:87}
  \mathcal{M}\Bigl(e^{- i\frac{t}{\varepsilon }H_{ren}(\sigma _0)}\varrho _{\varepsilon}e^{i\frac{t}{\varepsilon }H_{ren}(\sigma _0)}, \varepsilon \in(0,\bar{\varepsilon} )\Bigr)=\Bigl\{\mathbf{E}(t)_{\#}\mu \, ,\, \mu \in \mathcal{M}(\varrho _{\varepsilon },\varepsilon \in (0,\bar{\varepsilon} ))\Bigr\}\; .
\end{equation}

To ensure that
$\mathcal{M}\bigl(\varrho _{\varepsilon }, \varepsilon \in
(0,\bar{\varepsilon} )\bigr)\neq \emptyset $,
it is sufficient to assume that there exist $\delta >0$ and $C>0$ such
that, for any $\varepsilon \in (0,\bar{\varepsilon} )$,
$\Tr[\varrho _{\varepsilon }N^{\delta }]<C$; where $N$ is the number
operator of the Fock space $\Gamma _s(\mathcal{Z})$ with $\mathcal{Z}=L^2\oplus L^2$. Actually, we make the following more restrictive assumptions: Let
$(\varrho_{\varepsilon } )_{\varepsilon \in (0,\bar{\varepsilon} )}$
be a family of normal states on
$\Gamma _s(L^2 (\mathds{R}^3 )\oplus L^2 (\mathds{R}^3 ))$, then
\begin{gather}
  \label{eq:84}\tag{$A_{0}$}
  \exists \mathfrak{C}>0\, ,\, \forall \varepsilon \in (0,\bar{\varepsilon} )\, ,\, \forall k\in \mathds{N}\, ,\, \Tr[\varrho _{\varepsilon }N_1^k]\leq \mathfrak{C}^k\; ;\\
  \label{eq:86}\tag{$A_\rho$}
  \exists C>0  \, ,\, \forall \varepsilon \in (0,\bar{\varepsilon} )\, ,\, \Tr[\varrho _{\varepsilon }(N+U_{\infty }^{*}H_0U_{\infty })]\leq C\; ;
\end{gather}
where $N_1$ is the nucleonic number operator, $N=N_1+N_2$ the total
number operator, $H_0$ is the free Hamiltonian defined by
Equation~\eqref{eq:1} and $U_{\infty }$ is the unitary quantum
dressing defined in Lemma~\ref{lemma:1}. As a matter of fact, it is possible in principle to
remove Assumption~\eqref{eq:84}, but it has an important role in
connection with the parameter $\sigma _0$ related to the
renormalization procedure.
This condition restricts the considered states $\varrho _{\varepsilon }$ to be at most with
$[\mathfrak{C}/\varepsilon ]$ nucleons.

We are now in a position to state precisely our result: \emph{the
  Bohr's correspondence principle holds between the renormalized
  quantum dynamics of the Nelson model generated by $H_{ren}(\sigma _0)$ and the Schr\"odinger-Klein-Gordon classical flow generated by $\mathscr{E}$}. The operator
  $H_{ren}(\sigma _0)$  is constructed in Subsection \ref{sec:extension-mathcalh} according to Definition \ref{def:8}. Recall that $\mathcal{D}=Q(-\Delta +V)\oplus
  \mathcal{F}H^{\frac{1}{2}}(\mathds{R}^3)$.
\begin{thm}
  \label{thm:2}
  Let $\mathbf{E}:\mathds{R}\times \mathcal{D}\to \mathcal{D}$ be the
 Schr\"odinger-Klein-Gordon flow provided by Theorem \ref{prop:8} and solving the equation \eqref{eq:72} with a potential $V$   satisfying Assumption~\eqref{ass:2}. Let
  $(\varrho _{\varepsilon })_{\varepsilon \in (0,\bar{\varepsilon} )}$
  be a family of normal states in
  $\Gamma_s\bigl(L^2 (\mathds{R}^3 )\oplus L^2 (\mathds{R}^3 )\bigr)$
  that satisfies Assumptions~\eqref{eq:84} and~\eqref{eq:86}. Then:
  \begin{itemize}
    \item[(i)] there
  exists a $\sigma _0\in \mathds{R}_+$ such that the dynamics
  $e^{- i\frac{t}{\varepsilon }H_{ren}(\sigma _0)}$ is non-trivial on
  the states $\varrho _{\varepsilon }$.
    \item[(ii)] $\mathcal{M}\bigl(\varrho _{\varepsilon }, \varepsilon \in
  (0,\bar{\varepsilon} )\bigr)\neq \emptyset $
    \item[(iii)] for any $t\in \mathds{R}$,
  \begin{equation}
    \label{eq:88}
    \mathcal{M}\Bigl(e^{- i\frac{t}{\varepsilon } H_{ren}(\sigma _0)}\varrho _{\varepsilon}e^{i\frac{t}{\varepsilon }H_{ren}(\sigma _0)}, \varepsilon \in(0,\bar{\varepsilon} )\Bigr)=\Bigl\{\mathbf{E}(t)_{\#}\mu \, ,\, \mu \in \mathcal{M}\bigl(\varrho _{\varepsilon },\varepsilon \in (0,\bar{\varepsilon} )\bigr)\Bigr\}\; .
  \end{equation}
  \end{itemize}
  Furthermore, let
  $(\varepsilon _k)_{k\in \mathds{N}}\subset (0,\bar{\varepsilon} )$
  be a sequence such that $\lim_{k\to \infty }\varepsilon _k= 0$ and
  $\mathcal{M}\bigl(\varrho _{\varepsilon_k }, k \in
  \mathds{N}\bigr)=\{\mu\}$, i.e.: for any $\xi \in L^2 \oplus L^2 $,
  $$\lim_{k\to \infty }\Tr[\varrho _{\varepsilon _k}W(\xi )]= \int_{L^2\oplus L^2} e^{i
    \sqrt{2}\Re\langle \xi , z \rangle} d\mu (z).$$
  Then for any $t\in \mathds{R}$,
  $\mathcal{M}\bigl(e^{- i\frac{t}{\varepsilon_k } H_{ren}(\sigma
    _0)}\varrho _{\varepsilon}e^{i\frac{t}{\varepsilon_k
    }H_{ren}(\sigma _0)}, k \in
  \mathds{N}\bigr)=\{\mathbf{E}(t)_{\#}\mu\}$, i.e.:
  \begin{equation}
    \label{eq:92a}
    \lim_{k\to \infty }\Tr\Bigl[e^{- i\frac{t}{\varepsilon_k } H_{ren}(\sigma _0)}\varrho _{\varepsilon_k}e^{i\frac{t}{\varepsilon_k } H_{ren}(\sigma _0)}W(\xi )\Bigr]=\int_{L^2\oplus L^2} e^{i\sqrt{2}\Re\langle \xi , z \rangle} d (\mathbf{E}(t)_{\#}\mu)(z)\; ,\; \forall \xi \in L^2 \oplus L^2 \; .
  \end{equation}
 \end{thm}
\begin{remark}
\begin{itemize}[label=\color{myblue}*]
\item  The choice of $\sigma_0$ is related to our Definition \ref{def:8} of the renormalized dynamics and the localization of states $(\varrho_\varepsilon)_{\varepsilon\in(0,\bar\varepsilon)}$ satisfying Assumption \eqref{eq:84} (see Lemma \ref{lemma:10}).  Actually, one can take any $\sigma_0\geq2 K (\mathfrak{C}+1+\bar\varepsilon)$ where $K>0$  is a constant given in Theorem \ref{thm:1}.
  \item We remark that every Wigner measure
  $\mu \in \mathcal{M}(\varrho _{\varepsilon },\varepsilon \in
  (0,\bar{\varepsilon} ))$,
  with
  $(\varrho _{\varepsilon })_{\varepsilon \in (0,\bar{\varepsilon} )}$
  satisfying Assumption~\eqref{eq:86} is a Borel probability measure on
  $\mathcal{D}$ equipped with its graph norm, hence the push-forward by means of the classical flow   $\mathbf{E}$ is well defined (see
  Section~\ref{sec:defin-wign-meas}).
  \item Adopting a shorthand notation, the last assertion of the above
  theorem can be written as:
  \begin{equation*}
    \varrho_{\varepsilon _k}\rightarrow \mu \Leftrightarrow\biggl(\forall t\in \mathds{R}\;,\; e^{- i\frac{t}{\varepsilon_k} H_{ren}(\sigma _0)}\varrho_{\varepsilon_k}e^{i\frac{t}{\varepsilon_k } H_{ren}(\sigma_0)}\rightarrow \mathbf{E}(t)_{\#}\mu \biggr)\; .
  \end{equation*}
\end{itemize}
\end{remark}

\subsection{An alternative renormalization procedure.}
\label{sec:an-altern-renorm}

The quantization of physically interesting classical Hamiltonian systems with infinite degrees of freedom often leads to nowhere defined energy operators, that at best make sense only as quadratic forms. In order to cure the ultraviolet and infinite volume divergencies typical of QFT, a renormalization procedure is needed \cite{MR0255186}. On the
non-perturbative level, this has been usually done by manipulating a regularization of the quadratic form which one obtains by a formal
quantization. The goal is to define another, related, form that is
closed and semi-bounded when the regularization is removed: hence one gets a unique self-adjoint operator  and eventually this will lead to a construction of an interacting QFT verifying some of the Wightman axioms \cite{MR0161603}. This
``Hamiltonian approach'' was considered in few situations
\citep[e.g.][]{MR635783,MR0234682,MR0215585,MR0218073,
glimm1968
  ,MR0266533
  ,nelson:1190
}. We briefly review it in the following paragraphs.

Starting  with a one particle Hilbert space $\mathcal{Z}$ that has an associated symplectic
structure $\Sigma(\mathcal{Z})$ and a Fock representation on
$\Gamma_s(\mathcal{Z})$, one easily shows that the free dynamic is unitarily
implemented. Actually, the free classical Hamiltonian is described by a  positive closed densely defined quadratic  form
$\mathscr{E}_0:D(\mathscr{E}_0)\subset \mathcal{Z}\to \mathds{R}_+$. In
addition, the Wick quantization procedure $\mathrm{Wick}(\cdot )$
maps $\mathscr{E}_0$ into a closed, positive, and densely defined quadratic form
$h_0=\mathrm{Wick}(\mathscr{E}_0)$ on $\Gamma_s(\mathcal{Z})$ which in turn is associated to a unique self-adjoint operator $H_0$.

The interacting classical system is described by a ``perturbation'' of
$\mathscr{E}_0$, namely
$\mathscr{E}=\mathscr{E}_0+\mathscr{E}_I:D(\mathscr{E})\subset
\mathcal{Z}\to \mathds{R}$,
where $D(\mathscr{E})\cap D(\mathscr{E}_0)$ is dense in $\mathcal{Z}$.
Usually, one follows the same rules and maps $\mathscr{E}$ by
$\mathrm{Wick}$ onto a quadratic form $h=\mathrm{Wick}(\mathscr{E})$
on $\Gamma_s(\mathcal{Z})$. The difficulty is that $h$ may not make
sense rigourously or at best is just symmetric but not bounded from
below. Therefore it is not possible, a priori, to associate a
self-adjoint operator to $h$, and hence to define a quantum dynamics.

The non-perturbative renormalization of the Hamiltonian operator is
usually done as follows. A family $(h_{\sigma})_{\sigma\geq 0}$ of
closed, densely defined, and semi-bounded quadratic forms is
introduced, such that for any $\sigma\geq 0$,
$Q(h)\subset Q(h_{\sigma})$, and for any $\Psi,\Phi\in Q(h)$,
$\lim_{\sigma\to \infty}h_{\sigma}(\Psi,\Phi)=h(\Psi,\Phi)$. The idea
is, roughly speaking, to manipulate the form $h_{\sigma}$, or the
associated self-adjoint operator $H_{\sigma}$, in a way such that in
the limit $\sigma\to \infty$ a self-adjoint operator is obtained. This
is done by means of a $\sigma$-dependent transformation $U_{\sigma}$
on $\Gamma_s(\mathcal{Z})$, that is called dressing. The purpose of
the dressing is two-fold: single out the counter term that is causing the
unboundedness from below of $h$, and yield the correct domain where
the renormalized (dressed) operator is densely defined and (hopefully)
self-adjoint. As explained in detail in
Section~\ref{sec:nelson-hamiltonian}, for the Nelson model the
dressing $U_{\sigma}$ is a unitary transformation, the divergent term
causing the unboundedness is a scalar function $E_{\sigma}$ (self-energy), and the
domain of definition of the renormalized dressed operator
$\hat{H}_{ren}$ is a subset of the free form domain $Q(h_0)$. In fact one shows that
that
$\hat{h}_{\sigma}(\,\cdot \,,\,\cdot
\,)=h_{\sigma}(U_{\sigma}\,\cdot\, ,U_{\sigma}\,\cdot\,
)-E_{\sigma}\langle \,\cdot\, , \,\cdot\, \rangle_{}$
is a densely defined quadratic form that is closed and bounded from
below uniformly with respect to $\sigma$, and therefore defines a
unique self-adjoint operator $\hat{H}_{ren}$ in the limit. More
precisely, $\hat{H}_{\sigma}$ converges to $\hat{H}_{ren}$ in the
norm resolvent sense, see \cite{MR0293448}.

There are more complicated situations in which the above procedure is
not sufficient, and a so-called wavefunction renormalization is
required. This yields a change of Hilbert space for the renormalized
dynamics. The idea is that we can define a new Hilbert space
$\mathcal{R}$ as the completion of the pre-Hilbert space defined by
$\langle \hat{\Psi} , \hat{\Phi} \rangle_{\mathcal{R}}=\lim_{\sigma\to
  \infty}\frac{\langle U_{\sigma}\Psi , U_{\sigma}\Phi
  \rangle_{\Gamma_s(\mathcal{Z})}}{\lVert U_{\sigma}\Omega
  \rVert_{\Gamma_s(\mathcal{Z})}^2}$;
where $\Psi,\Phi\in D$, the latter being dense in
$\Gamma_s(\mathcal{Z})$, and $\Omega\in D$ is the Fock vacuum
(obviously in this case $U_{\sigma}$ is not unitary). If
$U_{\sigma}[D]\subset D(H_{\sigma})$, the renormalized Hamiltonian
$\hat{H}_{ren}$ on $\mathcal{R}$ is defined by
$\langle \hat{\Psi} , \hat{H}_{ren}\hat{\Phi}
\rangle_{\mathcal{R}}=\lim_{\sigma\to
  \infty}\frac{h_{\sigma}(U_{\sigma}\Psi ,U_{\sigma}\Phi
  )-s_{\sigma}(U_{\sigma}\Psi ,U_{\sigma}\Phi )}{\lVert
  U_{\sigma}\Omega \rVert_{\Gamma_s(\mathcal{Z})}^2}$,
where $s_{\sigma}(U_{\sigma}\Psi ,U_{\sigma}\Phi )$ is the singular
part of $h_\sigma$. Now suppose that: to the singular part
$s_{\sigma}$ is associated an operator $S_{\sigma}$ defined on
$D(H_{\sigma})$; and for any $\Psi\in D$,
$\lVert (H_{\sigma}-S_{\sigma})U_{\sigma}\Psi
\rVert_{\Gamma_s(\mathcal{Z})}^{}/\lVert U_{\sigma}\Omega
\rVert_{\Gamma_s(\mathcal{Z})}^{}$
is uniformly bounded with respect to $\sigma$. Then it follows from
Riesz's representation theorem that $\hat{H}_{ren}$ is a densely
defined (symmetric) operator on $\mathcal{R}$. The existence of
eventual self-adjoint extensions of $\hat{H}_{ren}$ has to be then
proved by other means.

In this paper, we would like to introduce a different point of view on
the non-perturbative renormalization. It is founded not on
manipulations at the quantized level, \emph{but on  suitable
  transformations of the classical energy before quantization}. These transformations are related to \emph{normal form techniques for Hamiltonian
  PDEs} \citep{MR3203027,MR1857574}.  In particular for the Nelson model, we exploit a classical counterpart of the dressing transformation that puts the S-KG equation into a normal form suitable for quantization, see Subsection~\ref{sec:classical-dressing}. We remark that we obtain the
same renormalized Nelson Hamiltonian $H_{ren}$ as the one obtained in
Section~\ref{sec:nelson-hamiltonian} by standard techniques in \cite{nelson:1190}.
The  idea of using the properties of the classical field equations to construct
a quantum field theory dates to a series of works by I.~E.~Segal
\citep{MR0135093
  ,MR0177670
  ,MR0192369
  ,MR0394772
} and \citet{MR0220490
}. Some years later, ``Segal's program'' has been continued by Balaban
\citep{MR0386522
  ,MR0416335
}. Local relativistic interacting quantum fields can be constructed,
but not the interacting quantum dynamics. Our approach to
non-perturbative renormalization by means of normal forms of Hamiltonian PDEs
may be seen as a further continuation of Segal's program, though we
adopt a different point of view. While the cited works concentrated on
the construction of the relativistic self-interacting fields, we focus on
particle-field interacting Hamiltonians.

The procedure is described in detail in
Section~\ref{sec:class-renorm-nels}, but we will briefly outline the
strategy here, for the idea is very simple. Instead of quantizing
directly the energy functional $\mathscr{E}$, we make a ``change of
coordinates'' in the classical phase space $\mathcal{Z}$. Let
$\mathbf{D}:\mathcal{Z}\to \mathcal{Z}$ be a symplectomorphism, such
that there exists a dense subset $\mathcal{D}\subset D(\mathscr{E})$
of $\mathcal{Z}$ such that
$\mathbf{D}[\mathcal{D}]\subset \mathcal{D}$. Then we can calculate
the energy in the new coordinates, i.e.  $\mathscr{E} (\mathbf{D}(z))$,
for any $z\in \mathcal{D}$. It turns out that not only the
quantization $\mathrm{Wick}(\mathscr{E}\circ\mathbf{D})$ differs from
$\mathrm{Wick}(\mathscr{E})$, but it has also different properties as a
quadratic form. In particular, for a specific choice of $\mathbf{D}$,
\emph{$\mathrm{Wick}(\mathscr{E}\circ\mathbf{D})$ is closed and bounded
  from below} and the associated self-adjoint operator is the dressed
  renormalized Hamiltonian $\hat{H}_{ren}$ (while
$\mathrm{Wick}(\mathscr{E})$ is not even bounded from below). It follows
that, for the Nelson model, a suitable classical near identity change of coordinates
is sufficient to renormalize the quantum dynamics, without any
additional manipulation (only a self-energy renormalization is needed).

In addition, $\mathbf{D}$ gives informations on the relation between
quantum dressed dynamics $\hat{H}_{ren}$, and the undressed one
$H_{ren}$. The canonical map $\mathbf{D}$ can be associated to a group
of symplectomorphisms
$\bigl(\mathbf{D}(\theta)\bigr)_{\theta\in \mathds{R}}$ with generator
$\mathscr{D}$, by the relation $\mathbf{D}=\mathbf{D}(1)$. It turns
out that $\mathrm{Wick}(\mathscr{D})$ is a well-defined
self-adjoint operator $T_{\infty}$, and therefore it generates a strongly
continuous one-parameter unitary  group $e^{-i \frac{\theta}{\epsilon}T_{\infty}}$
that is the quantum analogous of $\mathbf{D}(\theta)$. If we now
denote $\hat{\mathscr{E}}=\mathscr{E}\circ \mathbf{D}$, we define the
corresponding classical evolution group by $\hat{\mathbf{E}}(t)$;
also, we denote by $\mathbf{E}(t)$ the classical evolution group
associated to $\mathscr{E}$. By definition, these groups are related
by
$\mathbf{E}(t)=\mathbf{D}(1)\circ \hat{\mathbf{E}}\circ
\mathbf{D}(-1)$.
At the quantum level, it is not possible to define directly the
evolution corresponding to $\mathbf{E}(t)$, for
$\mathrm{Wick}(\mathscr{E})$ is not associated to a self-adjoint
operator. Nevertheless, the right hand side of the previous relation
has a quantum correspondent in
$e^{-\frac{i}{\epsilon}T_{\infty}}e^{-i
  \frac{t}{\epsilon}\hat{H}_{ren}}e^{\frac{i}{\epsilon}T_{\infty}}$,
and that provides an indirect definition of the undressed renormalized
dynamics $H_{ren}$:
\begin{equation*}
  e^{-i \frac{t}{\epsilon}H_{ren}}:=  e^{-\frac{i}{\varepsilon}T_{\infty}}e^{-i\frac{t}{\varepsilon}\hat{H}_{ren}}e^{\frac{i}{\varepsilon}T_{\infty}}\; .
\end{equation*}

To sum up, we think that this new point of view has some advantages
over the usual Hamiltonian approach, since the manipulations are done
at the classical level where observables are commuting (in our case,
an explicit near identity change of coordinates is sufficient and there is no
need of introducing cut-offs, and  divergent quantities like
the quantum self-energy do not show up explicitly). More importantly, the approach draws the
link with the topic of normal forms for Hamiltonian PDEs where substantial advances
have been accomplished. Therefore, we believe that this point of view deserves further study and we hope that it will be of interest in clarifying some of the mathematical methods and techniques of "constructive"  QFT.

\section{The quantum system: Nelson Hamiltonian.}
\label{sec:nelson-hamiltonian}

In this section we define the quantum system of "nucleons" interacting
with a  meson field, and review the standard renormalization procedure due to
\citet{nelson:1190
}. Since we are interested in the classical limit and our
original and  dressed  Hamiltonians depend in an effective parameter
$\varepsilon\in(0,\bar \varepsilon)$, we need  to derive several  estimates
that are uniform with respect to $\varepsilon$. So this explains  why
we go through the technical details related to the renormalization procedure
of the Nelson model.

On
$\mathcal{H}=\Gamma_s(L^2(\mathds{R}^3))\otimes \Gamma_s(L^2
(\mathds{R}^3))$
we define the following free
Hamiltonian as a positive self-adjoint operator given by:
\begin{equation}
  \label{eq:1}
  \begin{split}
    H_0=\int_{\mathds{R}^3}^{}\psi^{*}(x)\Bigl(-\tfrac{\Delta}{2M}+V(x)\Bigr)\psi(x)  dx+\int_{\mathds{R}^3}^{}a^{*}(k)\omega(k)a(k)  dk=d\Gamma (-\tfrac{\Delta }{2M}+V)+d\Gamma (\omega )\; ,
  \end{split}
\end{equation}
where $V\in L^2_{loc} (\mathds{R}^3,\mathds{R}_+)$. We denote its
domain of self-adjointness by $D(H_0)$. We denote by $d\Gamma $ the
second quantization acting either on the first or second Fock space,
when no confusion arises.

Now let $\chi\in \mathcal{C}_0^{\infty}(\mathds{R}^3)$;
$0\leq \chi\leq 1$ and $\chi\equiv 1$ if $\lvert k\rvert_{}^{}\leq 1$,
$\chi\equiv 0$ if $\lvert k\rvert_{}^{}\geq 2$. Then, for all
$\sigma>0$ define $\chi_{\sigma}(k)=\chi(k/\sigma)$; it will play the
role of an ultraviolet cutoff in the interaction. The Nelson
Hamiltonian with cutoff has thus the form:
\begin{equation}
  \label{eq:2}
  H_{\sigma}=H_0+\tfrac{1}{(2\pi)^{3/2}}\int_{\mathds{R}^3}^{}\psi^{*}(x)\Bigl(a^{*}\bigl(\tfrac{e^{-ik\cdot x}}{\sqrt{2\omega}}\chi_{\sigma}\bigr)+a\bigl(\tfrac{e^{-ik\cdot x}}{\sqrt{2\omega}}\chi_{\sigma}\bigr)\Bigr)\psi(x)  dx\; .
\end{equation}
We will denote the interaction part by $H_I(\sigma)=H_{\sigma}-H_0$.
\begin{remark}
  \label{rem:3}
  There is no loss of generality in the choice of $\chi$ as a radial
  function \citep[see][Proposition
  3.9]{MR1809881}.
\end{remark}
The following proposition shows the self-adjointness of $H_\sigma$, see e.g.
\citep[][Proposition
2.5]{Ammari:2014aa
} or \citep{Falconi:2014aa
}.
\begin{proposition}
  \label{prop:1}
  For any $\sigma>0$, $H_{\sigma}$ is essentially self-adjoint on
  $D(H_0)\cap \mathcal{C}^{\infty}_0(N)$.
\end{proposition}
To obtain a meaningful limit when $\sigma\to \infty$, we use  a
 dressing transformation, introduced in the physics literature
by \citet{greenberg1958
} following  the work of van Hove \cite{MR0104463,MR0104464}. The dressing and the renormalization procedure are described in
Sections~\ref{sec:dressing-transform} and~\ref{sec:renorm-hamilt}
respectively. In Section~\ref{sec:extension-mathcalh} we discuss a
possible extension of the renormalized Hamiltonian on $\mathcal{H}_n$
to the whole Fock space $\mathcal{H}$. The extension we choose is not
the only possible one, however the choice is motivated by two facts: other
extensions should provide the same classical limit, and our choice
$\hat{H}_{ren}(\sigma _0)$ is, in our opinion, more consistent with the
quantization procedure of the classical energy functional.

\subsection{Dressing.}
\label{sec:dressing-transform}

The dressing transform was introduced as an alternative way of doing
renormalization in the Hamiltonian formalism, and has been utilized in
a rigorous fashion in various situations \citep[see
e.g.][]{MR0234682,1969LNP.....2.....H,nelson:1190
  ,MR0266533
}. For the Nelson Hamiltonian, it consists of a unitary transformation
that singles out the singular self-energy.

From now on, let $0<\sigma_0<\sigma$, with $\sigma_0$ fixed. Then
define:
\begin{gather}
  \label{eq:3}
  g_{\sigma}(k)=-\frac{i}{(2\pi)^{3/2}}\frac{1}{\sqrt{2\omega(k)}}\frac{\chi_{\sigma}(k)-\chi_{\sigma_0}(k)}{\frac{k^2}{2M}+\omega(k)}\; ;\\
  \label{eq:4}
  E_{\sigma}=\frac{1}{2(2\pi)^3}\int_{\mathds{R}^3}^{}\frac{1}{\omega(k)}\frac{(\chi_{\sigma}-\chi_{\sigma_0})^2(k)}{\frac{k^2}{2M}+\omega(k)}  dk-\frac{1}{(2\pi)^3}\int_{\mathds{R}^3}^{}\frac{\chi_{\sigma}(k)}{\omega(k)}\frac{(\chi_{\sigma}-\chi_{\sigma_0})(k)}{\frac{k^2}{2M}+\omega(k)}  dk\; .
\end{gather}

The dressing transformation is the unitary operator generated by (the
dependence on $\sigma_0$ will be usually omitted):
\begin{equation}
  \label{eq:5}
  T_{\sigma}=\int_{\mathds{R}^3}^{}\psi^{*}(x)  \Bigl(a^{*}(g_{\sigma}e^{-ik\cdot x})+a(g_{\sigma}e^{-ik\cdot x})\Bigr)\psi(x)  dx\; .
\end{equation}
The function $g_{\sigma}\in L^2 (\mathds{R}^3)$ for all
$\sigma\leq \infty$; therefore it is possible to prove the following
Lemma, e.g. utilizing the criterion of
\citep{Falconi:2014aa
}.
\begin{lemma}
  \label{lemma:1}
  For any $\sigma\leq\infty$, $T_{\sigma}$ is essentially self-adjoint
  on $\mathcal{C }^{\infty}_0(N)$. We denote by $U_{\sigma}(\theta )$
  the corresponding one-parameter unitary group
  $U_{\sigma}(\theta )=e^{-i\frac{\theta }{\varepsilon} T_{\sigma}}$.
\end{lemma}
For the sake of brevity, we will write
$U_{\sigma}:=U_{\sigma}(-1)$. We remark that $T_{\sigma}$ and  $
H_{\sigma}$ preserve the number of ``nucleons'', i.e.:
for any $\sigma\leq\infty$, $\sigma'<\infty$:
\begin{equation}
  \label{eq:6}
  [T_{\sigma},N_1]=0=[H_{\sigma'},N_1]\; .
\end{equation}
The above operators also commute in the resolvent sense.
We are now in  position to define the dressed Hamiltonian
\begin{equation}
  \label{eq:7}
  \hat{H}_{\sigma}:=U_{\sigma}\bigl(H_{\sigma}-\varepsilon N_1 E_{\sigma}\bigr)U^{*}_{\sigma}\; .
\end{equation}
The operator $\hat{H}_{\sigma}$ is self-adjoint for any
$\sigma <\infty$, since $H_{\sigma}$ and $N_1$ are commuting
self-adjoint operators and $U_{\sigma}$ is unitary. The purpose is to
show that the quadratic form associated with
$\hat{H}_{\sigma}\bigr\rvert _{\mathcal{H }_n}$ satisfies the
hypotheses of KLMN theorem, even when $\sigma=\infty$, so it is
possible to define uniquely a self-adjoint operator
$\hat{H}_{\infty}$. In order to do that, we need to study in detail
the form associated with $\hat{H}_{\sigma}^{(n)}$.

By Equation~\eqref{eq:7}, it follows immediately that
\begin{equation}
  \label{eq:8}
  \hat{H}_{\sigma}^{(n)}=\varepsilon U_{\sigma}^{(n)}\Biggl(\frac{H_{\sigma}^{(n)}}{\varepsilon}-(\varepsilon n) E_{\sigma}\Biggr)(U^{(n)}_{\sigma})^{*}\; .
\end{equation}
A suitable calculation
\citep{nelson:1190
  ,MR1809881
} yields:
\begin{equation}
  \label{eq:9}
  \begin{split}
    \hat{H}_{\sigma}^{(n)}=H_{\sigma_0}^{(n)} +\varepsilon^2\sum_{i<j}^{}V_{\sigma}(x_i-x_j)+\frac{\varepsilon}{2M} \sum_{j=1}^n\Biggl(\Bigl(a^{*}(r_{\sigma}e^{-ik\cdot x_j})^2+a(r_{\sigma}e^{-ik\cdot x_j})^2\Bigr)\\+2a^{*}(r_{\sigma}e^{-ik\cdot x_j})a(r_{\sigma}e^{-ik\cdot x_j})-2\Bigl(D_{x_j}a(r_{\sigma}e^{-ik\cdot x_j})+a^{*}(r_{\sigma}e^{-ik\cdot x_j})D_{x_j}\Bigr)  \Biggr)\; ;
  \end{split}
\end{equation}
where $D_{x_j}=-i\nabla_{x_j} $ and
$$r_{\sigma}(k)=-ik g_{\sigma}(k)\,,$$
\begin{equation}
  \label{eq:10}
  V_{\sigma}(x)=2\Re \int_{\mathds{R}^3}^{}\omega(k)\lvert g_{\sigma}(k)\rvert_{}^2e^{-ik\cdot x}  dk -4\Im\int_{\mathds{R}^3}^{}\frac{\bar{g}_{\sigma}(k)}{(2\pi)^{3/2}}\frac{\chi_{\sigma}(k)}{\sqrt{2\omega(k)}}e^{-ik\cdot x}  dk\; .
\end{equation}
It is also possible to write $\hat{H}_{\sigma}$ in its second
quantized form as:
\begin{gather}
  \label{eq:11}
  \hat{H}_{\sigma}=H_0+\hat{H}_I(\sigma)\; ;\\
  \label{eq:12}
  \begin{split}
    \hat{H}_I(\sigma)=H_I(\sigma_0) + \frac{1}{2}\int_{\mathds{R}^{6}}^{}\psi^{*}(x)\psi^{*}(y)V_{\sigma}(x-y)\psi(x)\psi(y)  dxdy\\+\tfrac{1}{2M}\int_{\mathds{R}^3}^{}\psi^{*}(x)\Biggl(\Bigl(a^{*}(r_{\sigma}e^{-ik\cdot x})^2+a(r_{\sigma}e^{-ik\cdot x})^2\Bigr)+2a^{*}(r_{\sigma}e^{-ik\cdot x})a(r_{\sigma}e^{-ik\cdot x})\\-2\Bigl(D_{x}a(r_{\sigma}e^{-ik\cdot x})+a^{*}(r_{\sigma}e^{-ik\cdot x})D_{x}\Bigr)\Biggr)\psi(x)  dx\; .
  \end{split}
\end{gather}
\begin{remark}
  \label{rem:8}
  The dressed interaction Hamiltonian $\hat{H}_I(\sigma)$ contains a
  first term analogous to the undressed interaction with cutoff, a
  second term of two-body interaction between nucleons, and a more
  singular term that can be only defined as a form when
  $\sigma=\infty$.
\end{remark}

\subsection{Renormalization.}
\label{sec:renorm-hamilt}

We will now define the renormalized self-adjoint operator
$\hat{H}^{(n)}_{\infty}$. A simple calculation shows that
$E_{\sigma }\to -\infty $ when $\sigma \to +\infty $; hence the
subtraction of the self-energy in the definition~\eqref{eq:7} of
$\hat{H}_{\sigma}$ is necessary. It is actually the only
renormalization necessary for this system. We prove that the
quadratic form associated with $\hat{H}_{\sigma}^{(n)}$ of
Equation~\eqref{eq:9} has meaning for any $\sigma\leq \infty$, and the
KLMN theorem \citep[see][Theorem
X.17]{MR0493420
} can be applied, with a suitable choice of $\sigma_0$.

We start with some preparatory lemmas:
\begin{lemma}
  \label{lemma:2}
  For any $0\leq\sigma\leq\infty$, the symmetric function $V_{\sigma}$
  satisfies:
  \begin{enumerate}[label=\color{myblue}(\roman*)]
  \item\label{item:1}
    $V_{\sigma}(1-\Delta)^{-1/2}\in \mathcal{L }(L^2 (\mathds{R}^3))$;
  \item\label{item:2}
    $(1-\Delta)^{-1/2} V_{\sigma}(1-\Delta)^{-1/2}\in \mathcal{K}(L^2
    (\mathds{R}^3))$.
  \end{enumerate}
  In particular, $V_{\sigma}\in L^s (\mathds{R}^3)\cap  L^{3,\infty} (\mathds{R}^{3})$, for any
  $s\in[2,+\infty[$.
\end{lemma}
\begin{proof}\renewcommand{\qedsymbol}{}
  It is sufficient to show \citep[][Corollary
  D.6]{2011arXiv1111.5918A
  } that $V_{\sigma}\in L^{3,\infty} (\mathds{R}^{3})$ (weak-$L^p$
  spaces). Write
  $V_{\sigma}=V_{\sigma}^{(1)}+V_{\sigma}^{(2)}$,
  \begin{align}
    \label{eq:13}
    V_{\sigma}^{(1)}(x)&=2\Re\int_{\mathds{R}^3}^{}\omega(k)\lvert g_{\sigma}(k) \rvert_{}^2e^{-ik\cdot x}  dk =2(2\pi)^{3/2}\Re\mathcal{F}\Bigl(\omega\lvert g_{\sigma}  \rvert_{}^2\Bigr)(x) \; ;\\
    \label{eq:28}
    V_{\sigma}^{(2)}(x)&=-2 \sqrt{2} \Im \int_{\mathds{R}^3}^{}\frac{\bar{g}_{\sigma}(k)}{(2\pi)^{3/2}}\frac{\chi_{\sigma}(k)}{\sqrt{\omega(k)}}e^{-ik\cdot x}  dk=-2 \sqrt{2} \Im \mathcal{F}\Bigl(\bar{g}_{\sigma}\frac{\chi_{\sigma}}{\sqrt{\omega}}\Bigr)(x)\; .
  \end{align}
  \begin{itemize}[label=\color{myblue}\textbullet]
  \item $\bigl[V_{\sigma }^{(1)}\bigr]$. For any
    $\sigma \leq \infty $,
    $\omega \lvert g_{\sigma } \rvert_{}^2\in L^{s'} (\mathds{R}^3 )$,
    $1\leq s'\leq 2$. Then $V_{\sigma }^{(1)}\in L^s (\mathds{R}^3 )$
    for any $s\in [2,+\infty ]$; furthermore
    $V_{\sigma }^{(1)}\in \mathcal{C }_0(\mathds{R}^3)$ (the space of
    continuous functions converging to zero at infinity). Hence
    $V_{\sigma }^{(1)}\in L^{3,\infty } (\mathds{R}^3 )$.
  \item $\bigl[V_{\sigma }^{(2)}\bigr]$. For any
    $\sigma \leq \infty $,
    $\bar{g}_{\sigma }\frac{\chi _{\sigma }}{\sqrt{\omega }}\in L^{s'}
    (\mathds{R}^3 )$,
    $1<s'\leq 2$. Therefore $V_{\sigma }^{(2)}\in L^s (\mathds{R}^3 )$
    for any $s\in [2,+\infty [$. It remains to show that
    $V_{\sigma }^{(2)}\in L^{3,\infty } (\mathds{R}^{3} )$. Define
    $f(k)\in L^2 (\mathds{R}^3 )$:
    \begin{equation}
      \label{eq:14}
      f(k):=\frac{\chi _{\sigma }(k)}{\omega (k)}\frac{\bigl(\chi _\sigma -\chi _{\sigma _0}\bigr)(k)}{\frac{k^2}{2M}+\omega (k)}\; .
    \end{equation}
    Then there is a constant $c>0$ such that
    $\lvert V_{\sigma }^{(2)}(x) \rvert_{}^{}\leq c\lvert \mathcal{F
    }(f)(x) \rvert_{}^{}$,
    where the Fourier transform is intended to be on
    $L^2 (\mathds{R}^3 )$. The function $f$ is radial, so we introduce
    the spherical coordinates
    $(r,\theta ,\phi )\equiv k\in \mathds{R}^3$, such that the
    $z$-axis coincides with the vector $x$.  We then obtain:
    \begin{equation*}
      \begin{split}
        \lim_{R\to +\infty }\int_{B(0,R)}^{}f(k)e^{-ik\cdot x}  dk=\lim_{R\to +\infty }\int_0^Rdr\int_0^{\pi }d\theta \int_0^{2\pi }d\phi\; r^2f(r) e^{-ir\lvert x  \rvert_{}^{}\cos \theta }\sin \theta\\=2\pi \lim_{R\to +\infty }\int_0^R  dr\int_{-1}^1  dy \;r^2f(r)e^{-ir\lvert x  \rvert_{}^{}y}=\frac{4\pi }{\lvert x  \rvert_{}^{}}\lim_{R\to +\infty }\int_0^Rf(r)r\sin (r\lvert x  \rvert_{}^{})  dr\; .
      \end{split}
    \end{equation*}
    Since for any $\sigma \leq +\infty $, $f(r)r\in L^1 (\mathds{R} )$
    we can take the limit $R\to +\infty $ and conclude:
    \begin{equation}
      \label{eq:15}
      \mathcal{F}(f)(x)=\frac{4\pi }{\lvert x  \rvert_{}^{}}\int_0^{+\infty }f(r)r\sin (r\lvert x  \rvert_{}^{})  dr\; .
    \end{equation}
    Therefore, for any $x\in \mathds{R}^3\setminus \{0\}$, there
    exists a
    $0<\tilde{c}\leq 4\pi c\lVert f(r)r \rVert_{L^1 (\mathds{R} )}^{}$
    such that:
    \begin{equation}
      \label{eq:16}
      \lvert V_{\sigma }^{(2)}(x)  \rvert_{}^{}\leq \frac{\tilde{c}}{\lvert x  \rvert_{}^{}}\; .
    \end{equation}
    Let $\lambda $ be the Lebesgue measure in $\mathds{R}^3$. Since
    $\bigl\{x:\lvert V_{\sigma }^{(2)}\rvert_{}^{}>t \bigr\}\subset
    \bigl\{x:\frac{\tilde{c}}{\lvert x \rvert_{}^{}}>t\bigr\}$,
    there is a positive $C$ such that:
    \begin{equation}
      \label{eq:17}
      \lambda \Bigl\{x:\lvert V_{\sigma }^{(2)}(x)  \rvert_{}^{}>t\Bigr\}\leq \lambda \Bigl\{x:\frac{\tilde{c}}{\lvert x  \rvert_{}^{}}>t\Bigr\}\leq \frac{C}{t^3}\; .
    \end{equation}
    Finally \eqref{eq:17} implies
    $V_{\sigma }^{(2)}\in L^{3,\infty } (\mathds{R}^3 )$.\hfill
    $\blacksquare$
  \end{itemize}
\end{proof}

\begin{lemma}
  \label{lemma:3}
  There exists $c>0$ such that for any
  $\varepsilon \in (0,\bar{\varepsilon}) $, $\sigma \leq +\infty $:
  \begin{align}
    \label{eq:18}
    \Bigl\lVert \bigl[(H_0+1)^{-1/2}D_{x_j} a(r_{\sigma }e^{-ik\cdot
      x_j})(H_0+1)^{-1/2}\bigr]^{(n)}
    \Bigr\rVert_{\mathcal{L}(\mathcal{H}_n)}^{}&\leq
    \frac{c}{\sqrt{n\varepsilon}} \lVert \omega ^{-1/2}r_{\sigma }  \rVert_2^{}\; ;\\
    \label{eq:19}
    \Bigl\lVert\bigl[(H_0+1)^{-1/2} a^{*}(r_{\sigma }e^{-ik\cdot x_j})
    D_{x_j} (H_0+1)^{-1/2}\bigr]^{(n)}
    \Bigr\rVert_{\mathcal{L}(\mathcal{H}_n)}^{}&\leq
    \frac{c}{\sqrt{n\varepsilon}} \lVert \omega ^{-1/2}r_{\sigma }  \rVert_2^{}\; .
  \end{align}
  Moreover, \eqref{eq:18} holds if we replace the left $H_0$ by
  $d\Gamma (-\tfrac{\Delta }{2M}+V)$ and the right $H_0$ by
  $d\Gamma (\omega )$ and similarly \eqref{eq:19} holds if we replace the
  left $H_0$ by $d\Gamma (\omega )$ and the right $H_0$ by
  $d\Gamma (-\tfrac{\Delta }{2M}+V)$.
\end{lemma}
\begin{proof}
  Let $S_n\equiv S_n\otimes 1$ be the symmetrizer on $\mathcal{H }_n$
  (acting only on the $\{x_1,\dotsc,x_{n}\}$ variables) and
  $\Psi _n\in \mathcal{H }_n$ with $n>0$. Then:
  \begin{equation*}
    \begin{split}
      \langle \Psi _n  , d\Gamma (-\Delta )\Psi _n \rangle_{}=\langle \Psi _n  , (n\varepsilon) S_n(D_{x_1})^2\otimes 1^{n-1} \Psi _n \rangle_{}=(n\varepsilon )\langle \Psi _n  , (D_{x_j})^2\Psi _n \rangle_{}\; .
    \end{split}
  \end{equation*}
  Hence
  $(n\varepsilon) \lVert D_{x_j}\Psi _n \rVert_{}^2\leq \bigl\lVert
  \bigl(d\Gamma (-\Delta )+1\bigr)^{1/2}\Psi _n \bigr\rVert_{}^2$.
  It follows that
  \begin{equation}
    \label{eq:20}
    \Bigl\lVert \bigl[D_{x_j}\bigl(d\Gamma (-\Delta )+1\bigr)^{-1/2}\bigr]^{(n)}  \Bigr\rVert_{\mathcal{L}(\mathcal{H}_n)}^{}\leq \frac{1}{\sqrt{n\varepsilon }}\; ;\; \Bigl\lVert\bigl[\bigl(d\Gamma (-\Delta )+1\bigr)^{-1/2}D_{x_j}\bigr]^{(n)}  \Bigr\rVert_{\mathcal{L}(\mathcal{H}_n)}^{}\leq \frac{1}{\sqrt{n\varepsilon }}\; .
  \end{equation}
  Using \eqref{eq:20} we obtain for any $\Psi _n\in \mathcal{H }_n$,
  with $\lVert \Psi _n \rVert_{}^{}=1$:
  \begin{eqnarray*}
      \Bigl\lVert (H_0+1)^{-1/2}D_{x_j} a(r_{\sigma }e^{-ik\cdot x_j})(H_0+1)^{-1/2}\Psi _n  \Bigr\rVert_{}^{}&\leq& \frac{c}{\sqrt{n\varepsilon }}\Bigl\lVert a(r_{\sigma }e^{-ik\cdot x_j})\bigl(d\Gamma (\omega )+1\bigr)^{-1/2}\Psi _n  \Bigr\rVert_{}^{}\\&\leq& \frac{c}{\sqrt{n\varepsilon }}\lVert \omega ^{-1/2}r_{\sigma }  \rVert_2^{}\; ;
  \end{eqnarray*}
  where the last inequality follows from standard estimates on the
  Fock space \citep[see][Lemma
  2.1]{Ammari:2014aa
  }. The bound~\eqref{eq:19} is obtained by adjunction.
\end{proof}
\begin{lemma}
  \label{lemma:4}
  There exists $c>0$ such that for any
  $\varepsilon \in (0,\bar{\varepsilon} )$, $\sigma \leq +\infty $:
  \begin{align}
    \label{eq:21}
    \Bigl\lVert \bigl[(H_0+1)^{-1/2}a^{*}(r_{\sigma }e^{-ik\cdot x_j}) a(r_{\sigma }e^{-ik\cdot x_j}) (H_0+1)^{-1/2}\bigr]^{(n)} \Bigr\rVert_{\mathcal{L}(\mathcal{H}_n)}^{}&\leq c\lVert \omega ^{-1/2}r_{\sigma }  \rVert_2^2\; ;\\
    \label{eq:22}
    \Bigl\lVert \bigl[(H_0+1)^{-1/2}\bigl(a^{*}(r_{\sigma }e^{-ik\cdot x_j})\bigr)^2 (H_0+1)^{-1/2}\bigr]^{(n)} \Bigr\rVert_{\mathcal{L}(\mathcal{H}_n)}^{}&\leq c\lVert \omega ^{-1/4}r_{\sigma }  \rVert_2^2\; ;\\
    \label{eq:23}
    \Bigl\lVert \bigl[(H_0+1)^{-1/2}\bigl(a(r_{\sigma }e^{-ik\cdot x_j})\bigr)^2 (H_0+1)^{-1/2}\bigr]^{(n)} \Bigr\rVert_{\mathcal{L}(\mathcal{H}_n)}^{}&\leq c\lVert \omega ^{-1/4}r_{\sigma }  \rVert_2^2\; .
  \end{align}
  The same bounds hold if $H_0$ is replaced by $d\Gamma (\omega )$.
\end{lemma}
\begin{proof}
  First of all observe that, since $m_0>0$, there exists $c>0$ such
  that, uniformly in $\varepsilon \in (0,\bar{\varepsilon} )$:
  \begin{align*}
    \Bigl\lVert (H_0+1)^{-1/2}\bigl(d\Gamma (\omega )+1\bigr)^{1/2}  \Bigr\rVert_{\mathcal{L}(\mathcal{H})}^{}\leq c\; ;\;
    \Bigl\lVert (H_0+1)^{-1/2}(N_2+1)^{1/2}  \Bigr\rVert_{\mathcal{L}(\mathcal{H})}^{}\leq c\; .
  \end{align*}
  Equation~\eqref{eq:21} is easy to prove:
  \begin{equation*}
    \begin{split}
      \Bigl\lVert \bigl[(H_0+1)^{-1/2}a^{*}(r_{\sigma }e^{-ik\cdot x_j}) a(r_{\sigma }e^{-ik\cdot x_j}) (H_0+1)^{-1/2}\bigr]^{(n)} \Bigr\rVert_{\mathcal{L}(\mathcal{H}_n)}^{}\leq c\Bigl\lVert \bigl[\bigl(d\Gamma (\omega )+1\bigr)^{-1/2} \\a^{*}(r_{\sigma }e^{-ik\cdot x_j})\bigr]^{(n)} \Bigr\rVert_{\mathcal{L}(\mathcal{H}_n)}^{}\,\cdot \,\Bigl\lVert \bigl[a(r_{\sigma }e^{-ik\cdot x_j}) \bigl(d\Gamma (\omega )+1\bigr)^{-1/2} \bigr]^{(n)} \Bigr\rVert_{\mathcal{L}(\mathcal{H}_n)}^{}\leq c\lVert \omega ^{-1/2}r_{\sigma }  \rVert_2^2\; .
    \end{split}
  \end{equation*}
  For the proof of \eqref{eq:22} the reader may refer to
  \citep[][Lemma 3.3
  (iv)]{MR1809881
  }. Finally \eqref{eq:23} follows from \eqref{eq:22} by
  adjunction.
\end{proof}
\begin{lemma}
  \label{lemma:5}
  There exists $c(\sigma _0)>0$ such that for any
  $\varepsilon \in (0,\bar{\varepsilon} )$ and $\lambda \geq 1$:
  \begin{align}
    \label{eq:24}
    \Bigl\lVert \bigl[ \bigl(H_0+\lambda \bigr)^{-1/2}H_I(\sigma _0)\bigl(H_0+\lambda \bigr)^{-1/2}  \bigr]^{(n)}  \Bigr\rVert_{\mathcal{L}(\mathcal{H}_n)}^{}&\leq c(\sigma _0)\lambda ^{-1/2}(n\varepsilon )\; ;\\
    \label{eq:25}
    \Bigl\lVert \bigl[ \bigl(H_0+\lambda \bigr)^{-1/2}\varepsilon
    ^2\sum_{i<j}^{}V_{\sigma }(x_i-x_j) \bigl(H_0+\lambda
    \bigr)^{-1/2}  \bigr]^{(n)}
    \Bigr\rVert_{\mathcal{L}(\mathcal{H}_n)}^{}&\leq c(\sigma _0)
    \lambda ^{-1/2}\sqrt{
n\varepsilon (1+n\varepsilon) }\; .
  \end{align}
\end{lemma}
\begin{proof}
  The inequality \eqref{eq:24} can be proved by a standard argument on
  the Fock space \citep[see e.g.][Proposition
  IV.1]{falconi:012303
  }.

  To prove \eqref{eq:25} we proceed as follows. First of all, by means
  of \ref{item:1}, Lemma~\ref{lemma:2} we can write:
  \begin{equation*}
    \begin{split}
      \Bigl\lVert \bigl( -\Delta _{x_i}+\lambda   \bigr)^{-1/2}V_{\sigma }(x_i-x_j)\bigl( -\Delta _{x_i}+\lambda   \bigr)^{-1/2}  \Bigr\rVert_{\mathcal{L}(\mathcal{H}_n)}^{}\leq \lambda ^{-1/2}\Bigl\lVert V_{\sigma }(x_i)\bigl( -\Delta _{x_i}+\lambda   \bigr)^{-1/2}  \Bigr\rVert_{\mathcal{L}(\mathcal{H}_n)}^{}\\\leq c(\sigma _0)\lambda ^{-1/2}\; .
    \end{split}
  \end{equation*}
  Therefore
  $V_{\sigma }(x_i-x_j)\leq c(\sigma _0)\lambda ^{-1/2}\bigl(-\Delta
  _{x_i}+\lambda \bigr)$.
  Let $\Psi _n\in \mathcal{H }_n$; using its symmetry, and some
  algebraic manipulations we can write:
  \begin{equation*}
    \begin{split}
      \bigl\langle \Psi _n  , \varepsilon ^2\sum_{i<j}^{}V_{\sigma }(x_i-x_j)\Psi _n \bigr\rangle_{}\leq c(\sigma _0)(n\varepsilon )^2\bigl\langle \Psi _n  , \bigl(\lambda ^{-1/2}(D_{x_1})^2+\lambda ^{1/2}\bigr)\Psi _n \bigr\rangle_{}\\=c(\sigma _0)\bigl\langle \Psi _n  , N_1\bigl(\lambda ^{-1/2}d\Gamma (D_x^2)+\lambda ^{1/2}N_1\bigr)\Psi _n \bigr\rangle_{}\\\leq c(\sigma _0) \lambda ^{-1/2}\Bigl[\Bigl\lVert N_1^{1/2}\bigl(d\Gamma (D_x^2)+\lambda \bigr)^{1/2} \Psi _n \Bigr\rVert_{}^2+\Bigl\lVert N_1\bigl(d\Gamma (D_x^2)+\lambda \bigr)^{1/2}\Psi _n  \Bigr\rVert_{}^2\Bigr]\\\leq c(\sigma _0)\lambda^{-1/2} \bigl\langle \Psi _n  , \bigl(N_1+N_1^2\bigr)\bigl(d\Gamma (D_x^2)+\lambda \bigr)\Psi _n \bigr\rangle_{}\; ;
    \end{split}
  \end{equation*}
  where the constant $c(\sigma _0)$ is redefined in each
  inequality. The result follows since $N_1$ commutes with
  $d\Gamma (D_x^2)$.
\end{proof}

Combining Lemmas~\ref{lemma:3},~\ref{lemma:4} and~\ref{lemma:5}
together, we can prove easily the following proposition.
\begin{proposition}
  \label{prop:2}
  There exist $c>0$ and $c(\sigma _0)>0$ such that for any
  $\varepsilon \in (0,\bar{\varepsilon} )$, $\lambda \geq 1$,
  $\sigma _0<\sigma \leq +\infty $ and for any $\Psi \in D(N_1)$:
  \begin{equation}
    \label{eq:26}
    \begin{split}
      \Bigl\lVert \bigl(H_0+\lambda \bigr)^{-1/2}\hat{H}_I(\sigma ) \bigl(H_0+\lambda \bigr)^{-1/2}\Psi   \Bigr\rVert_{}^{}\leq \Bigl[c\bigl(\lVert \omega ^{-1/2}r_{\sigma }  \rVert_2^2+\lVert \omega ^{-1/4}r_{\sigma }  \rVert_2^2+\lVert \omega ^{-1/2}r_{\sigma }  \rVert_2^{}\bigr)\\+c(\sigma _0)\lambda^{-1/2} \Bigr]\,\cdot \,\Bigl\lVert (N_1+1)\Psi   \Bigr\rVert_{}^{}\; .
    \end{split}
  \end{equation}
\end{proposition}

Consider now $\hat{H}_I(\sigma )^{(n)}$. It follows easily from
Equation~\eqref{eq:26} above that for any
$\sigma _0<\sigma \leq +\infty $, and
$\Psi _n\in D(H_0^{1/2})\cap \mathcal{H }_n$:
\begin{equation}
  \label{eq:27}
  \begin{split}
    \Bigl\lvert \bigl\langle \Psi _n  , \hat{H}_I(\sigma )^{(n)}\Psi _n \bigr\rangle_{}  \Bigr\rvert_{}^{}\leq \Bigl[c(n\varepsilon +1)\bigl(\lVert \omega ^{-1/2}r_{\sigma }  \rVert_2^2+\lVert \omega ^{-1/4}r_{\sigma }  \rVert_2^2+\lVert \omega ^{-1/2}r_{\sigma }  \rVert_2^{}\bigr)\\+c(\sigma _0)(n\varepsilon +1)\lambda^{-1/2} \Bigr]\bigl\langle \Psi _n  , H_0^{(n)}\Psi _n \bigr\rangle_{}\\+\lambda \Bigl[c(n\varepsilon +1)\bigl(\lVert \omega ^{-1/2}r_{\sigma }  \rVert_2^2+\lVert \omega ^{-1/4}r_{\sigma }  \rVert_2^2+\lVert \omega ^{-1/2}r_{\sigma }  \rVert_2^{}\bigr)\\+c(\sigma _0)(n\varepsilon +1)\lambda^{-1/2} \Bigr]\bigl\langle \Psi _n  ,  \Psi _n\rangle_{}\; .
  \end{split}
\end{equation}
Consider now the term
$\bigl(\lVert \omega ^{-1/2}r_{\sigma } \rVert_2^2+\lVert \omega
^{-1/4}r_{\sigma } \rVert_2^2+\lVert \omega ^{-1/2}r_{\sigma }
\rVert_2^{}\bigr)$;
by definition of $r_{\sigma }$, there exists $c>0$ such that,
uniformly in $\sigma \leq +\infty $:
\begin{equation}
  \label{eq:29}
  \lVert \omega ^{-1/2}r_{\sigma }\rVert_2^2+\lVert \omega ^{-1/4}r_{\sigma } \rVert_2^2+\lVert \omega^{-1/2}r_{\sigma } \rVert_2^{}\leq c\bigl( \sigma _0^{-2}+\sigma _0^{-1}\bigr)\; .
\end{equation}
Hence for any $\sigma _0\geq 1$ there exist $K >0$ ($K=2c$), $c(\sigma _0)>0$
and $C(n,\varepsilon ,\lambda ,\sigma _0)>0$ such that
\eqref{eq:27} becomes:
\begin{equation}
  \label{eq:30}
  \begin{split}
    \Bigl\lvert \bigl\langle \Psi _n  , \hat{H}_I(\sigma )^{(n)}\Psi _n \bigr\rangle_{}  \Bigr\rvert_{}^{}\leq \Bigl[\tfrac{K(n\varepsilon +1)}{\sigma _0}+c(\sigma _0)(n\varepsilon +1)\lambda^{-1/2} \Bigr]\bigl\langle \Psi _n  , H_0^{(n)}\Psi _n \bigr\rangle_{}+C(n,\varepsilon ,\lambda ,\sigma _0)\bigl\langle \Psi _n  ,  \Psi _n\rangle_{}\; .
  \end{split}
\end{equation}
Therefore choosing
\begin{equation}
  \label{eq:32}
  \sigma _0>2K(n\varepsilon+1)
\end{equation}
and then $\lambda > \bigl(2c(\sigma _0)(n\varepsilon +1)\bigr)^2$, we
obtain the following bound for any
$\Psi _n\in D(H_0^{1/2})\cap \mathcal{H }_n$, with $a<1$, $b>0$ and
uniformly in $\sigma _0< \sigma \leq +\infty $:
\begin{equation}
  \label{eq:31}
  \Bigl\lvert \bigl\langle \Psi _n  , \hat{H}_I(\sigma )^{(n)}\Psi _n \bigr\rangle_{}  \Bigr\rvert_{}^{}\leq a\bigl\langle \Psi _n  , H_0^{(n)}\Psi _n \bigr\rangle_{}+b \bigl\langle \Psi _n  ,  \Psi _n\rangle_{}\; .
\end{equation}
Applying KLMN theorem, \eqref{eq:31} proves the following result
\citep[see e.g.][for additional
details]{nelson:1190
  ,MR1809881
}.
\begin{thm}
  \label{thm:1}
  There exists $K>0$ such that, for any $n\in \mathds{N}$, and
  $\varepsilon \in (0,\bar{\varepsilon} )$ the following statements
  hold:
  \begin{enumerate}[label=\color{myblue}(\roman*)]
  \item\label{item:3} For any
    $\bigl(2K(n\varepsilon +1)\bigr)<\sigma_0<\sigma \leq +\infty $,
    there exists a unique self-adjoint operator
    $\hat{H}_{\sigma }^{(n)}$ with domain
    $\hat{D}_{\sigma }^{(n)}\subset
    D\bigl((H_0^{(n)})^{1/2}\bigr)\subset \mathcal{H }_n$
    associated to the symmetric form
    $\hat{h}_{\sigma }^{(n)}(\cdot ,\cdot )$, defined for any
    $\Psi ,\Phi \in D\bigl((H_0^{(n)})^{1/2}\bigr)$ as:
    \begin{equation}
      \label{eq:33}
      \hat{h}_{\sigma }^{(n)}(\Psi ,\Phi)=\bigl\langle \Psi   ,  H_{0}^{(n)}\Phi  \bigl\rangle_{}+\bigl\langle \Psi   ,  \hat{H}_{I }(\sigma )^{(n)}\Phi  \bigl\rangle_{}\; .
    \end{equation}
    The operator $\hat{H}_{\sigma }^{(n)}$ is bounded from below, with
    bound $-b_{\sigma _0}(\sigma)$ (where
    $\lvert b_{\sigma _0}(\sigma) \rvert_{}^{}$ is a bounded increasing
    function of $\sigma $).
  \item\label{item:4} The following convergence holds in the norm topology  of
    $\mathcal{L }(\mathcal{H }_n)$:
    \begin{equation}
      \label{eq:34}
      \lim_{\sigma \to +\infty }(z-\hat{H}_{\sigma
      }^{(n)})^{{-1}}=(z-\hat{H}_{\infty }^{(n)})^{-1}\;,\quad
      \text{ for all } \quad z\in \mathds{C}\setminus\mathds{R}.
    \end{equation}
  \item\label{item:5} For any $t\in \mathds{R}$, the following
    convergence holds in the strong topology of
    $\mathcal{L }(\mathcal{H }_n)$:
    \begin{equation}
      \label{eq:35}
      s-\lim_{\sigma \to +\infty }e^{-i \frac{t}{\varepsilon }\hat{H}_{\sigma }^{(n)}}=e^{-i \frac{t}{\varepsilon }\hat{H}_{\infty }^{(n)}}\; .
    \end{equation}
  \end{enumerate}
\end{thm}
\begin{remark}
  \label{rem:1}
  The operator $\hat{H}_{\infty }^{(n)}$ can be decomposed only in the
  sense of forms, i.e.
  \begin{equation}
    \label{eq:36}
    \hat{H}_{\infty }^{(n)}=H_0^{(n)}\;\dotplus\;\hat{H}_I^{(n)}(\infty )\; ;
  \end{equation}
  where $\dotplus$ has to be intended as the form sum.
\end{remark}

\subsection{Extension of $\hat{H}_{\infty }^{(n)}$ to $\mathcal{H }$.}
\label{sec:extension-mathcalh}

We have defined the self-adjoint operator $\hat{H}_{\infty }^{(n)}$
 which depends in $\sigma_0$ for each $n\in \mathds{N}$. Now we are interested in extending it to
the whole space $\mathcal{H }$. This can be done in at least two different
ways, however we choose the one that is more suitable to interpret
$\hat{H}_{\infty }$ as the Wick quantization of a classical symbol.

Let $K$ be defined by Theorem~\ref{thm:1}. Then define
$\mathfrak{N}(\varepsilon,\sigma _0 )\in \mathds{N}$ by:
\begin{equation}
  \label{eq:55}
  \mathfrak{N}(\varepsilon,\sigma _0 )=\Bigl[\frac{\sigma _0-2K}{2K\varepsilon } -1 \Bigr]\; ;
\end{equation}
where the square brackets mean that we take the integer part if the
number within is positive, zero otherwise.
\begin{definition}[$\hat{H}_{ren}(\sigma _0)$]
  \label{def:7}
  Let $0\leq \sigma _0<+\infty $ be fixed. Then we define
  $\hat{H}_{ren}(\sigma _0)$ on $\mathcal{H }$ by:
  \begin{equation}
    \label{eq:59}
    \hat{H}_{ren }(\sigma _0)\bigr\rvert_{\mathcal{H}_n}=\left\{
      \begin{aligned}
        &\hat{H}_{\infty }^{(n)}&\text{ if }n\leq \mathfrak{N}(\varepsilon,\sigma _0 )\\
        &0&\text{ if }n>\mathfrak{N}(\varepsilon,\sigma _0 )
      \end{aligned}
    \right.
  \end{equation}
  where $\mathfrak{N}(\varepsilon,\sigma _0 )$ is defined by
  \eqref{eq:55}. We may also write
  $\hat{H}_{ren}(\sigma _0)=H_0\dotplus \hat{H}_{ren, I}(\sigma _0)$
  as a sum of quadratic forms.
\end{definition}
The operator $\hat{H}_{ren}(\sigma _0)$ is self-adjoint on
$\mathcal{H }$, with domain of self adjointness:
\begin{equation}
  \label{eq:60}
  \hat{D}_{ren }(\sigma _0)=\Bigl\{ \Psi \in \mathcal{H}\;,\;\Psi\bigr\rvert_{\mathcal{H}_n}\in \hat{D}_{\infty }^{(n)}\text{ for any }n\leq \mathfrak{N}(\varepsilon,\sigma _0 )  \Bigr\}\; .
\end{equation}
Acting with the dressing operator $U_{\infty }$ defined in
Lemma~\ref{lemma:1} (with the same fixed $\sigma _0$ as for
$\hat{H}_{ren }(\sigma _0)$), we can also define the undressed
extension $H_{ren}(\sigma _0)$.
\begin{definition}[$H_{ren}(\sigma _0)$]
  \label{def:8}
  Let $0\leq\sigma _0<+\infty $ be fixed. Then we define the
  following operator on $\mathcal{H }$:
  \begin{equation}
    \label{eq:61}
    H_{ren }(\sigma _0)=U^{*}_{\infty }(\sigma _0)\hat{H}_{ren }(\sigma _0)U_{\infty }(\sigma _0)\; .
  \end{equation}
\end{definition}
The operator $H_{ren}(\sigma _0)$ is self-adjoint on $\mathcal{H }$,
with domain of self adjointness:
\begin{equation}
  \label{eq:62}
  D_{ren }(\sigma _0)=\Bigl\{ \Psi \in \mathcal{H}\;,\;\Psi\bigr\rvert_{\mathcal{H}_n}\in e^{-\frac{i}{\varepsilon }T_{\infty }^{(n)}}\hat{D}_{\infty }^{(n)}\text{ for any }n\leq \mathfrak{N}(\varepsilon,\sigma _0 )  \Bigr\}\; .
\end{equation}
\begin{remark}
  \label{rem:7}
  Let $\sigma _0\geq 0$ be fixed. Then the $\hat{H}_{\sigma }$ given by
  \eqref{eq:7} defines, in the limit $\sigma \to \infty $, a
  symmetric quadratic form $\hat{h}_{\infty }$ on
  $D(H_0^{1/2})\subset \mathcal{H }$. Also $\hat{H}_{ren}(\sigma _0)$
  defines a quadratic form $\hat{h}_{ren}$. We
  have\footnote{$\mathds{1}_{[0,\mathfrak{N}]}(N_1)$ is the orthogonal
    projector on $\bigoplus _{n=0}^{\mathfrak{N}}\mathcal{H }_n$.}:
  \begin{equation}
    \label{eq:64}
    \hat{h}_{\infty }(\mathds{1}_{[0,\mathfrak{N}]}(N_1)\,\cdot\, ,\,\cdot \,)=\hat{h}_{ren}(\mathds{1}_{[0,\mathfrak{N}]}(N_1)\,\cdot \,,\,\cdot\, )\; .
  \end{equation}
  However, we are not able to prove that there is a self-adjoint
  operator on $\mathcal{H}$ associated to $\hat{h}_{\infty }$, and it
  is possible that there is none.
\end{remark}

\section{The classical system: S-KG equations.}
\label{sec:classical-system}

In this section we define the S-KG system, with initial data in a
suitable dense subset of $L^2 (\mathds{R}^3 )\oplus L^2 (\mathds{R}^3 )$, that describes the
classical dynamics of a particle-field interaction. Then we introduce the classical
dressing transformation (viewed itself as a dynamical system), and
then study the transformation it induces on the Hamiltonian
functional. Finally, we discuss the global existence of  unique
solutions of the classical equations, both in their original and dressed
form.

\paragraph*{The Yukawa coupling:}

The S-KG[Y] system (Schrödinger-Klein-Gordon with Yukawa interaction),
or undressed classical equations, is defined by:
\begin{equation}
  \label{eq:70}\tag{S-KG[Y]}
  \left\{
    \begin{aligned}
      &i\partial_t u=-\frac{\Delta }{2M}u+ Vu+ A u\\
      &(\square+m^2_0)A=- \lvert u \rvert^2
    \end{aligned}
  \right .\; ;
\end{equation}
where $V:\mathds{R}^3\to \mathds{R}$ is an external potential. If we
introduce the complex field $\alpha $, defined by
\begin{align}
  \label{eq:71}
  A(x)&=\tfrac{1}{(2\pi )^{\frac{3}{2}}}\int_{\mathds{R}^3}^{}\tfrac{1}{\sqrt{2\omega (k)}}\bigl(\bar{\alpha}(k)e^{-ik\cdot x}+\alpha(k)e^{ik\cdot x}  \bigr)  dk\; ,\\
  \dot{A}(x)&=-\tfrac{i}{(2\pi )^{\frac{3}{2}}}\int_{\mathds{R}^3}^{}\sqrt{\tfrac{\omega(k)}{2}}\bigl(\alpha(k)e^{ik\cdot x}-\bar{\alpha}(k)e^{-ik\cdot x}  \bigr)  dk\; ,
\end{align}
we can rewrite \eqref{eq:70} as the equivalent system\footnote{The two
  systems are equivalent since
  $(1+\omega^{\varsigma } )\Re \alpha\in L^2 (\mathds{R}^3
  )\Leftrightarrow A\in H^{\varsigma +1/2}(\mathds{R}^3)$,
  $(1+\omega^{\varsigma } )\Im \alpha\in L^2 (\mathds{R}^3
  )\Leftrightarrow \partial _t A\in H^{\varsigma -1/2}(\mathds{R}^3)$.
  In \eqref{eq:72} the unknowns are $u$ and $\alpha $.}:
\begin{equation}
  \label{eq:72}\tag{S-KG$_{\alpha }$[Y]}
  \left\{
    \begin{aligned}
      &i\partial_t u=-\frac{\Delta }{2M}u+ Vu+ A u\\
      &i\partial _t\alpha =\omega \alpha +\frac{1}{\sqrt{2\omega }}\mathcal{F}\bigl(\lvert u  \rvert_{}^2\bigr)
    \end{aligned}
  \right .\; .
\end{equation}

\paragraph*{The ``dressed'' coupling:}

The system that arises from the dressed interaction is quite
complicated. We will denote it by S-KG[D], and it has the following
form\footnote{We denote by $\partial _{(i)}$ the derivative with
  respect to the $i$-th component of the variable $x\in \mathds{R}^3$.
  Analogously, we denote by $v^{(i)}$ the $i$-th component of a
  $3$-dimensional vector $v$.}:
\begin{equation}
  \label{eq:73}\tag{S-KG[D]}
  \left\{
    \begin{aligned}
      i\partial _tu&=-\frac{\Delta }{2M}u+Vu+(W*\lvert u  \rvert_{}^2)u+[(\varphi *A)+(\xi  *\partial _tA)]u+\sum_{i=1}^3[(\rho ^{(i)}*A)\partial _{(i)}+(\zeta ^{(i)}*A)^2]u\\
      (\square +&m^2_0)A=-\varphi *\lvert u  \rvert_{}^2+i\sum_{i=1}^3\rho ^{(i)}*[(u\partial _{(i)}u)-\sqrt{2M}(\zeta ^{(i)}*\partial _tA)]
    \end{aligned}
  \right.
\end{equation}
where: $V,W,\varphi :\mathds{R}^3\to \mathds{R}$ with $W,\varphi $
even; $\xi :\mathds{R}^3\to \mathds{C}$, even;
$\rho :(\mathds{R}^3)^3\to \mathds{C}$, odd; and
$\zeta :(\mathds{R}^3)^3\to \mathds{R}$, odd. Obviously also
\eqref{eq:73} can be written as an equivalent system
S-KG$_{\alpha }$[D], with unknowns $u$ and $\alpha $ (omitted
here). As discussed in detail in Section \ref{sec:glob-exist-results},
with a suitable choice of $W$, $\varphi$, $\xi$, $\rho$ and $\zeta $
the global well-posedness of \eqref{eq:73} follows directly from the
global well-posedness of \eqref{eq:70}.

\subsection{Dressing.}
\label{sec:classical-dressing}

We look for a classical correspondent of the dressing transformation
$U_{\infty }(\theta )$. Since $U_{\infty }(\theta )$ is a
one-parameter group of unitary transformations on $\mathcal{H }$, the
classical counterpart of its generator is expected to induce a
non-linear evolution on the phase-space
$L^2 (\mathds{R}^3 )\oplus L^2 (\mathds{R}^3 )$, using the quantum-classical
correspondence principle for systems with infinite degrees of freedom
 \citep[see
e.g.][]{MR530915,MR0332046,ammari:nier:2008
}. The resulting ``classical dressing'' $D_{g_{\infty }}(\theta )$
plays a crucial role in proving our results: on one hand it is
necessary to link the S-KG classical dynamics with the quantum dressed
one; on the other it is at the heart of the "classical" renormalization
procedure.

Let $g\in L^2 (\mathds{R}^3 )$; define the following functional
$\mathscr{D}_g:L^2 (\mathds{R}^3 )\oplus L^2 (\mathds{R}^3 )\to
\mathds{R}$,
\begin{equation}
  \label{eq:40}
  \mathscr{D}_g(u, \alpha ):= \int_{\mathds{R}^6}^{}\Bigl(g(k)\bar{\alpha}(k)e^{-ik\cdot x}+\bar{g}(k)\alpha (k)e^{ik\cdot x} \Bigr) \lvert u(x)  \rvert_{}^2 dxdk\;.
\end{equation}
The functional $\mathscr{D}_g$ induces the following Hamiltonian
equations of motion:
\begin{equation}
  \label{eq:41}
  \left\{
    \begin{aligned}
      i\partial _{\theta }u&=A_gu\\
      i\partial _{\theta }\alpha &=g F(\lvert u  \rvert_{}^2)
    \end{aligned}
  \right. \; ;
\end{equation}
where
\begin{align}
  \label{eq:42}
  A_g(x)&=\int_{\mathds{R}^3}^{}\Bigl(g(k)\bar{\alpha}(k)e^{-ik\cdot x}+\bar{g}(k)\alpha (k)e^{ik\cdot x} \Bigr)  dk\; ,\\
  \label{eq:43}
  F(\lvert u  \rvert_{}^2)(k)&=\int_{\mathds{R}^3}^{}e^{-ik\cdot x}\lvert u(x)  \rvert_{}^2  dx\; .
\end{align}
Observe that for any $g\in L^2 (\mathds{R}^3 )$ and
$x\in \mathds{R}^3$, $A_g(x)\in \mathds{R}$. This will lead to an
explicit form for the solutions of the Cauchy problem related to
\eqref{eq:41}. The latter can be rewritten in integral form, for any
$\theta \in \mathds{R}$:
\begin{equation}
  \label{eq:44}
  \left\{
    \begin{aligned}
      u_{\theta}(x)&=u_0(x)\exp\biggl\{-i\int_{0}^{\theta }(A_g)_{\tau}(x)\,  d\tau \biggr\}\\
      \alpha_{\theta }(k) &=\alpha_0(k)-ig(k)\int_{0}^{\theta }F(\lvert  u_{\tau } \rvert_{}^2)(k)\,  d\tau
    \end{aligned}
  \right. \; ;
\end{equation}
where $(A_g)_{\tau}$ is defined by \eqref{eq:42} with $\alpha $
replaced by $\alpha _{\tau}$; analogously we define $B_g$ by
\eqref{eq:42} with $\alpha $ replaced by $\beta $.
\begin{lemma}
  \label{lemma:6}
  Let $s\geq 0$, $s-\frac{1}{2}\leq \varsigma \leq s+\frac{1}{2}$;
  $(1+\omega^{\frac{1}{2}} )g\in L^2 (\mathds{R}^3)$. Also, let
  $u,v\in H^s(\mathds{R}^3)$ and
  $(1+\omega^{\varsigma } )\alpha ,(1+\omega^{\varsigma } )\beta \in
  L^2 (\mathds{R}^3 )$.
  Then there exist constants $C_s,C_{\varsigma }>0$ such that:
  \begin{gather}
    \label{eq:45}
    \lVert (A_g-B_g)u  \rVert_{H^s}^{}\leq C_s\max_{w\in \{u,v\}}\lVert w  \rVert_{H^s}^{}\lVert (1+\omega ^{\frac{1}{2}})g  \rVert_2^{}\lVert (1+\omega ^{\varsigma})(\alpha-\beta)  \rVert_2^{}\; ,\\
    \label{eq:46}
    \lVert A_g(u-v)  \rVert_{H^s}^{}\leq C_s\max_{\zeta \in \{\alpha ,\beta \}}\lVert (1+\omega ^{\varsigma})\zeta   \rVert_2^{}\lVert (1+\omega ^{\frac{1}{2}})g  \rVert_2^{}\lVert u-v  \rVert_{H^s}^{}\; ,\\
    \label{eq:47}
    \Bigl\lVert (1+\omega ^{\varsigma})g\int_{\mathds{R}^3}^{}e^{-ik\cdot x}\bigl((u-v)\bar{v}+(\bar{u}-\bar{v})u\bigr)  dx  \Bigr\rVert_2^{}\leq C_\varsigma \max_{w\in \{u,v\}}\lVert w  \rVert_{H^s}^{}\lVert (1+\omega ^{\frac{1}{2}})g  \rVert_2^{}\lVert u-v  \rVert_{H^s}^{}\; .
  \end{gather}
\end{lemma}
\begin{proof}
  If $s\in \mathds{N}$, the results follow by standard estimates,
  keeping in mind that
  $\lvert k \rvert_{}^{}\leq \omega(k) \leq \lvert k \rvert_{}^{}+m_0 $.
  The bounds for non-integer $s$ are then obtained by interpolation.
\end{proof}
\begin{proposition}
  \label{prop:5}
  Let $\theta \in \mathds{R}$, $(u_0, \alpha _0)\in L^2 \oplus L^2 $.
  If
  $(u_{\theta}, \alpha _{\theta})\in \mathcal{C}^0_{}(\mathds{R}, L^2
  \oplus L^2 )$
  is a solution of~\eqref{eq:44}, then it is unique,
  i.e. any
  $(v_{\theta}, \beta _{\theta})\in \mathcal{C}^0_{}(\mathds{R}, L^2
  \oplus L^2 )$
  that satisfies~\eqref{eq:44} is such that
  $ (v_{\theta}, \beta _{\theta})=(u_{\theta}, \alpha _{\theta})$.
\end{proposition}
\begin{proof}
  We have:
  \begin{equation*}
    \begin{split}
      \frac{i}{2}\partial _{\theta }\Bigl(\bigl\lVert u_{\theta}-v_{\theta}  \bigr\rVert_2^2+\bigl\lVert \alpha _{\theta}-\beta _{\theta}  \bigr\rVert_2^2\Bigr)=\Im \Bigl(\Bigl\langle u_{\theta}-v_{\theta}  , \bigl((A_g)_{\theta}-(B_g)_{\theta} \bigr)u_{\theta}+(B_g)_{\theta }\bigl(u_{\theta }-v_{\theta } \bigr) \Bigr\rangle_2 \\+\Bigl\langle \alpha _{\theta }-\beta _{\theta }  , g\int_{\mathds{R}^3}^{}e^{-ik\cdot x}\bigl((u_{\theta }-v_{\theta })\bar{v}_{\theta }+(\bar{u}_{\theta }-\bar{v}_{\theta })u_{\theta }\bigr)  dx  \Bigr\rangle_2 \Bigr)\; .
    \end{split}
  \end{equation*}
  The result hence is an application of the estimates of
  Lemma~\ref{lemma:6} with $s=0$ and Gronwall's Lemma.
\end{proof}

Now that we are assured that the solution of~\eqref{eq:44} is unique,
we can construct it explicitly. Since $A_g(x)$ is real, it follows
that for any $\theta \in \mathds{R}$:
$\lvert u_{\theta } \rvert_{}^{}=\lvert u_0 \rvert_{}^{}$.  Therefore
$F(\lvert u_{\theta } \rvert_{}^2)=F(\lvert u_0 \rvert_{}^2)$, and
\begin{equation*}
  \alpha _{\theta }(k)=\alpha _0(k)-i\theta g(k)F(\lvert u_0  \rvert_{}^2)(k)\; .
\end{equation*}
Substituting this explicit form in the expression for $u_{\theta}$, we
obtain the solution for any
$(u_0, \alpha _0)\equiv (u, \alpha) \in L^2 (\mathds{R}^3 )\oplus L^2
(\mathds{R}^3 )$:
\begin{equation}
  \label{eq:48}
  \left\{
    \begin{aligned}
      u_{\theta}(x)&=u(x) \; \exp\biggl\{-i\theta A_g(x)+i \theta^2\Im \int_{\mathds{R}^3}^{}F(\lvert u  \rvert_{}^2)(k)\lvert g(k)  \rvert_{}^2e^{ik\cdot x}  dk\biggr\}\\
      \alpha _{\theta}(k)&=\alpha(k) -i\theta g(k)F(\lvert u  \rvert_{}^2)(k)
    \end{aligned}
  \right.\; .
\end{equation}
This system of equations defines a non-linear symplectomorphism: the
``classical dressing map'' on $L^2 \oplus L^2 $.
\begin{definition}
  \label{def:4}
  Let $g\in L^2 (\mathds{R}^3 )$. Then
  $\mathbf{D}_g(\cdot ):\mathds{R}\times \bigl(L^2 \oplus L^2
  \bigr)\to L^2 \oplus L^2 $ is defined by \eqref{eq:48} as:
  \begin{equation*}
    \mathbf{D}_g(\theta ) (u, \alpha )=(u_{\theta },\alpha _{\theta})\; .
  \end{equation*}
  The map $\mathbf{D}_g(\cdot )$ is the Hamiltonian flow generated by
  $\mathscr{D}_g$.
\end{definition}
Using the explicit form~\eqref{eq:48} and Lemma~\ref{lemma:6}, it is
straightforward to prove some interesting properties of the classical
dressing map. The results are formulated in the following proposition,
after the definition of  useful classes  of subspaces of
$L^2 \oplus L^2 $.
\begin{definition}
  \label{def:5}
  Let $s\geq 0$, $s-\frac{1}{2}\leq \varsigma \leq s+\frac{1}{2}$. We
  define the spaces
  $H^{s}(\mathds{R}^3)\oplus\mathcal{F }H^{\varsigma
  }(\mathds{R}^3)\subseteq L^2 (\mathds{R}^3 )\oplus L^2 (\mathds{R}^3
  )$:
  \begin{equation*}
    H^{s}(\mathds{R}^3)\oplus\mathcal{F}H^{\varsigma }(\mathds{R}^3) =\Bigl\{(u, \alpha) \in L^2 (\mathds{R}^3 )\oplus L^2 (\mathds{R}^3  )\text{ , } u\in H^s(\mathds{R}^3)\text{ and } \mathcal{F}^{-1}(\alpha) \in H^{\varsigma } (\mathds{R}^3  )\Bigr\}\; .
  \end{equation*}
\end{definition}
\begin{proposition}
  \label{prop:6}
  Let $s\geq 0$, $s-\frac{1}{2}\leq \varsigma \leq s+\frac{1}{2}$; and
  $g\in \mathcal{F }H^{\frac{1}{2}} (\mathds{R}^3 )$. Then
  \begin{equation*}
    \mathbf{D}_g: \mathds{R}\times \bigl(H^{s}\oplus\mathcal{F}H^{\varsigma }\bigr)\to H^{s}\oplus\mathcal{F}H^{\varsigma }\; ;
  \end{equation*}
  i.e. the flow preserves the spaces
  $H^{s}\oplus\mathcal{F}H^{\varsigma }$.  Furthermore, it is a
  bijection with inverse
  $\bigl(\mathbf{D}_g(\theta )\bigr)^{-1}=\mathbf{D}_g(-\theta )$.
  Hence the classical dressing is an Hamiltonian flow on
  $H^{s}\oplus\mathcal{F }H^{\varsigma }$.
\end{proposition}
\begin{corollary}
  \label{cor:1}
  Let $s\geq 0$, $s-\frac{1}{2}\leq \varsigma \leq s+\frac{1}{2}$,
  $\theta \in \mathds{R}$, and
  $g\in \mathcal{F }H^{\frac{1}{2}} (\mathds{R}^3 )$. Then there
  exists a constant $C(g,\theta)>0$ and a
  $\lambda(s) \in \mathds{N}^{*}$ such that for any
  $(u, \alpha) \in H^{s}\oplus\mathcal{F }H^{\varsigma }$:
  \begin{equation}
    \label{eq:63}
    \lVert \mathbf{D}_g(\theta ) (u, \alpha )  \rVert_{H^s\oplus \mathcal{F}H^{\varsigma }}^{}\leq C(g,\theta) \lVert (u, \alpha)  \rVert_{H^s\oplus \mathcal{F}H^{\varsigma }}^{\lambda(s) }\; .
  \end{equation}
\end{corollary}
Using the positivity of both $-\Delta $ and $V$, and
Corollary~\ref{cor:1} one also obtains the following result.
\begin{corollary}
  \label{cor:2}
  Let $V\in L^2_{loc} (\mathds{R}^d, \mathds{R}_+ )$; and let
  $Q(-\Delta +V)\subset L^2 (\mathds{R}^3 )$ be the form domain of
  $-\Delta +V$. Then for any
  $\frac{1}{2}\leq \varsigma \leq \frac{3}{2}$, and
  $g\in \mathcal{F }H^{\frac{1}{2}} (\mathds{R}^3 )$:
  \begin{equation*}
    \mathbf{D}_g: \mathds{R}\times \bigl(Q(-\Delta +V)\oplus \mathcal{F}H^{\varsigma }\bigr)\to Q(-\Delta +V)\oplus \mathcal{F}H^{\varsigma }\; .
  \end{equation*}
\end{corollary}

\subsection{Classical Hamiltonians.}
\label{sec:class-hamilt}

In this section we define the classical Hamiltonian functionals that
generate the undressed and dressed dynamics on
$L^2 \oplus L^2 $. Then we show that
they are related by a suitable classical dressing: the quantum
procedure described in Section~\ref{sec:renorm-hamilt} is reproduced,
in simplified terms, on the classical level.

\begin{definition}[$\mathscr{E}$, $\hat{\mathscr{E}}$]\label{def:9}
  The undressed Hamiltonian (or energy) $\mathscr{E}$ is defined as
  the following real functional on
  $L^2 (\mathds{R}^3 )\oplus L^2 (\mathds{R}^3 )$:
  \begin{equation*}
    \label{eq:49}
    \mathscr{E}(u, \alpha ):= \Bigl\langle u  , \Bigl(-\tfrac{\Delta }{2M}+V\Bigr)u \Bigr\rangle_2+\langle \alpha   , \omega \alpha  \rangle_2+\tfrac{1}{(2\pi )^{3/2}}\int_{\mathds{R}^6}^{}\tfrac{1}{\sqrt{2\omega (k)}}\Bigl(\bar{\alpha}(k)e^{-ik\cdot x}+ \alpha (k)e^{ik\cdot x}\Bigr)\lvert u(x)  \rvert_{}^2  dxdk \; .
  \end{equation*}
  We denote by $\mathscr{E}_0$ the free part of the classical energy,
  namely
  \begin{equation*}\label{eq:110}
    \mathscr{E}_0(u, \alpha )= \Bigl\langle u  , \Bigl(-\tfrac{\Delta}{2M}+V\Bigr)u \Bigr\rangle_2+\langle \alpha   , \omega \alpha  \rangle_2\; .
  \end{equation*}
  Let
  $\chi _{\sigma _0}\in L^{\infty } (\mathds{R}^3 )\cap \mathcal{F
  }H^{-1/2}(\mathds{R}^3)$
  such that $\chi_{\sigma _0}(k)=\chi _{\sigma _0}(-k)$ for any
  $k\in \mathds{R}^3$. Then (again as a real functional on
  $L^2\oplus L^2$) the dressed Hamiltonian $\hat{\mathscr{E}}$ is
  defined as\footnote{We recall that:
    $g_{\infty
    }(k)=-i\frac{(2\pi)^{-3/2}}{\sqrt{2\omega(k)}}\frac{1-\chi_{\sigma_0}(k)}{\frac{k^2}{2M}+\omega(k)}$
    ;
    $V_{\infty }(x)=2\Re \int_{\mathds{R}^3}^{}\omega(k)\lvert g_{\infty
    }(k)\rvert_{}^2e^{-ik\cdot x} dk
    -4\Im\int_{\mathds{R}^3}^{}\frac{\bar{g}_{\infty
      }(k)}{(2\pi)^{3/2}}\frac{1}{\sqrt{2\omega(k)}}e^{-ik\cdot x} dk$
    .  Also, $D_{x}=-i\nabla_{x} $ ;
    $r_{\infty}(k)=-ik g_{\infty }(k)$ .}:
  \begin{equation*}
    \label{eq:50}
    \begin{split}
      \hat{\mathscr{E}}(u, \alpha ):=\Bigl\langle u  , \Bigl(-\tfrac{\Delta }{2M}+V\Bigr)u \Bigr\rangle_2+\langle \alpha   , \omega \alpha  \rangle_2 +\tfrac{1}{(2\pi )^{3/2}}\int_{\mathds{R}^6}^{}\tfrac{\chi _{\sigma _0}(k)}{\sqrt{2\omega (k)}}\Bigl(\bar{\alpha}(k)e^{-ik\cdot x}+ \alpha (k)e^{ik\cdot x}\Bigr)\lvert u(x)  \rvert_{}^2  dxdk \\+\tfrac{1}{2M}\int_{\mathds{R}^9}^{}\Bigl(r_{\infty }(k)\bar{\alpha}(k)e^{-ik\cdot x}+\bar{r}_{\infty }(k)\alpha (k)e^{ik\cdot x} \Bigr) \Bigl(r_{\infty }(l)\bar{\alpha}(l)e^{-il\cdot x}+\bar{r}_{\infty }(l)\alpha (l)e^{il\cdot x} \Bigr)  \lvert u(x)  \rvert_{}^2dxdkdl\\-\tfrac{2}{M}\Re\int_{\mathds{R}^6}^{}r_{\infty }(k) \bar{\alpha} (k)e^{-ik\cdot x}\bar{u}(x)D_xu(x)  dxdk +\tfrac{1}{2}\int_{\mathds{R}^6}^{}V_{\infty}(x-y)\lvert u(x)  \rvert_{}^2\lvert u(y)  \rvert_{}^2  dxdy\; .
    \end{split}
  \end{equation*}
\end{definition}

\begin{remark}
  \label{rem:4}
  We denote by $D(\mathscr{E})\subset L^2 \oplus L^2 $ the domain of
  definition of $\mathscr{E}$, and by
  $D(\hat{\mathscr{E}})\subset L^2 \oplus L^2 $ the domain of
  definition of $\hat{\mathscr{E}}$. We have that
  $D(\mathscr{E})\supset \mathcal{C}^{\infty }_0\oplus
  \mathcal{C}^{\infty }_0$
  and
  $D(\hat{\mathscr{E}})\supset \mathcal{C}^{\infty }_0\oplus
  \mathcal{C}^{\infty }_0$.
  Therefore both $\mathscr{E}$ and $\hat{\mathscr{E}}$ are densely
  defined, and $D(\mathscr{E})\cap D(\hat{\mathscr{E}})$ is dense in
  $L^2 \oplus L^2 $.
\end{remark}

We are interested in the action of $\mathscr{E}$ and
$\hat{\mathscr{E}}$ on $H^1\oplus \mathcal{F }H^{\frac{1}{2}}$, since
this emerges naturally as the energy space of the system, at least
when $V=0$.
\begin{lemma}
  \label{lemma:7}
  Let $\theta \in \mathds{R}$,
  $g\in \mathcal{F }H^{\frac{1}{2}} (\mathds{R}^3 )$. Then for
  any
  $u\in Q(V)\cap H^1(\mathds{R}^3)$, and
  $\alpha \in \mathcal{F }H^{\frac{1}{2}}(\mathds{R}^3)$:
  $\mathbf{D}_g(\theta )(u, \alpha )\in D(\mathscr{E})$.
\end{lemma}
\begin{proof}
  Let $u \in Q(V)$, and $\alpha \in L^2 (\mathds{R}^3 )$. Then
  \begin{equation*}
    \langle u_{\theta}  , V u_{\theta} \rangle_2=\langle u  , V u \rangle_2\; ;
  \end{equation*}
  where $u_{\theta}$ is defined in Equation~\eqref{eq:48}, and it is
  the first component of $\mathbf{D}_g(\theta )(u, \alpha )$. Also,
  for any $(u, \alpha) \in H^1\oplus \mathcal{F }H^{\frac{1}{2}}$ we
  have that:
  \begin{equation*}
    \begin{split}
      \Bigl\lvert\int_{\mathds{R}^6}^{}\lvert u(x)  \rvert_{}^2 \frac{1}{\sqrt{\omega (k)}}\alpha (k)e^{ik\cdot x}  dxdk  \Bigr\rvert_{}^{}=C\Bigl\lvert \int_{\mathds{R}^3}^{}\frac{1}{\lvert k  \rvert_{}^{}\omega (k)}\bigl(\omega^{1/2} \alpha \bigr)(k)\Bigl(\int_{\mathds{R}^3}^{}\bigl(D_x\lvert u(x)  \rvert_{}^2\bigr)e^{ik\cdot x}  dx\Bigr)  dk \Bigr\rvert_{}^{}\\
      \leq 2C \Bigl\lVert \frac{1}{\lvert k  \rvert_{}^{}\omega (k)}  \Bigr\rVert_2^{}\lVert \omega^{1/2} \alpha  \rVert_2^{}\lVert u  \rVert_2^{}\lVert u  \rVert_{H^1}^{}<+\infty \; .
    \end{split}
  \end{equation*}
  The result then follows since $\mathbf{D}_g(\theta )$ maps
  $H^1\oplus \mathcal{F }H^{\frac{1}{2}}$ into itself by
  Proposition~\ref{prop:6}.
\end{proof}

The functional $\mathscr{E}$ is independent of $g_{\infty }$, while
$\hat{\mathscr{E}}$ depends on it. In addition, we know that
$g_{\infty }$ has been fixed, at the quantum level, to renormalize the
Nelson Hamiltonian, and it is the function that appears in the
generator of the dressing transformation $U_{\infty }$. Hence, since
we are establishing a correspondence between the classical and quantum
theories, we expect it to be the function that appears in the
classical dressing too. Two features of $g_{\infty }$ are very
important in the classical setting: the first is that
$g_{\infty }\in \mathcal{F } H^{\frac{1}{2}} (\mathds{R}^3)$ for any
$\chi _{\sigma _0}\in L^{\infty } \cap \mathcal{F }H^{-\frac{1}{2}}$;
the second is that it is an even function, i.e.
$g_{\infty }(k)=g_{\infty }(-k)$ for any $k\in \mathds{R}^3$. Using
the first fact, one shows that $\mathbf{D}_{g_{\infty }}(\cdot )$ maps
the energy space into itself (and that will be convenient
when discussing global solutions); using the second property we can
simplify the explicit form of $\mathbf{D}_{g_{\infty }}(\cdot )$.
\begin{lemma}
  \label{lemma:8}
  Let $\theta \in \mathds{R}$, and $g\in L^2 (\mathds{R}^3 )$. If $g$
  is an even or odd function, then the map $\mathbf{D}_g(\theta )$
  defined by \eqref{eq:48} becomes:
  \begin{equation}
    \label{eq:51}
    \mathbf{D}_g(\theta )(u(x),\alpha(k) )= \Bigl(\:u(x) e^{-i\theta A_g(x)}\:,\: \alpha(k) -i\theta g(k)F(\lvert u  \rvert_{}^2)(k)\:\Bigr)\; .
  \end{equation}
\end{lemma}
\begin{proof}
  Consider
  $I(x):=\int_{\mathds{R}^3}^{}F(\lvert u \rvert_{}^2)(k)\lvert g(k)
  \rvert_{}^2e^{ik\cdot x} dk$.
  We will show that $\bar{I}(x)=I(x)$. We have that:
  \begin{equation*}
    \bar{I}(x)=\int_{\mathds{R}^6}^{}\lvert u(x')  \rvert_{}^2\lvert g(k)  \rvert_{}^2e^{-ik\cdot (x-x')}  dx'dk=\int_{\mathds{R}^6}^{}\lvert u(x')  \rvert_{}^2\lvert g(-k)  \rvert_{}^2e^{ik\cdot (x-x')}  dx'dk\; .
  \end{equation*}
  Now if $g$ is either even or odd,
  $\lvert g(-k) \rvert_{}^{}=\lvert g(k) \rvert_{}^{}$. Hence
  $\bar{I}(x)=I(x)$, therefore $\Im I(x)=0$.
\end{proof}

We conclude this section proving its main result: $\mathscr{E}$ and
$\hat{\mathscr{E}}$ are related by the $\mathbf{D}_{g_{\infty }}(1)$
classical dressing\footnote{We recall again that
  $g_{\infty }=-i
  \frac{(2\pi)^{-3/2}}{\sqrt{2\omega(k)}}\frac{1-\chi_{\sigma_0}(k)}{\frac{k^2}{2M}+\omega(k)}$ .}.
\begin{proposition}
  \label{prop:7}
  For any $u\in Q(V)\cap H^1(\mathds{R}^3)$,
  $\alpha \in \mathcal{F }H^{\frac{1}{2}}(\mathds{R}^3)$, and for any
  $\chi _{\sigma _0}\in L^{\infty } (\mathds{R}^3 )\cap \mathcal{F
  }H^{-\frac{1}{2}}(\mathds{R}^3)$:
  \begin{enumerate}[label=\color{myblue}(\arabic*)]
  \item\label{item:6} $(u, \alpha) \in D(\mathscr{E})$;
  \item\label{item:7} $(u, \alpha) \in D(\hat{\mathscr{E}})$;
  \item\label{item:8}
    $\hat{\mathscr{E}}(u, \alpha )=\mathscr{E}\circ
    \mathbf{D}_{g_{\infty }}(1)(u,\alpha )$.
  \end{enumerate}
\end{proposition}
\begin{remark}
  \label{rem:5}
  Relation \ref{item:8} of Proposition~\ref{prop:7} actually holds for
  any $(u, \alpha) \in \mathbf{D}_{g_{\infty }}(-1)D(\mathscr{E})$.
\end{remark}
\begin{proof}[Proof of Proposition~\ref{prop:7}]
  The statement \ref{item:6} is just an application of
  Lemma~\ref{lemma:7} when $\theta =0$. If \ref{item:8} holds
  formally, than \ref{item:7} follows directly, since by
  Lemma~\ref{lemma:7} the right hand side of \ref{item:8} is well
  defined. It remains to prove that the relation \ref{item:8} holds
  formally. This is done by means of a direct calculation, that we
  will briefly outline here.
  \begin{align}
    \notag
    \mathscr{E}\circ \mathbf{D}_{g_{\infty }}(1)(u,\alpha )=&\Bigl\langle ue^{-iA_{g_{\infty }}}  , \frac{D_x}{2M}D_x\bigl(ue^{-iA_{g_{\infty }}}\bigr) \Bigr\rangle_2+\langle u  , Vu \rangle_2+\langle \alpha   , \omega \alpha  \rangle_2\\
    \label{eq:53}\tag{a}
    &+2\Im \langle \alpha   , \omega g_{\infty } F_u \rangle_2+\frac{1}{(2\pi )^{3/2}}2\Re\int_{\mathds{R}^6}^{}\frac{1}{\sqrt{2\omega (k)}}\bar{\alpha}(k)e^{-ik\cdot x}\lvert u(x)  \rvert_{}^2  dxdk\\
    \label{eq:54}\tag{b}
    &+\lVert \omega g_{\infty }F_u  \rVert_2^2+\frac{1}{(2\pi )^{3/2}}2\Im\int_{\mathds{R}^6}^{}\frac{1}{\sqrt{2\omega (k)}}g_{\infty }(k)F_u(k)e^{ik\cdot x} \lvert u(x)  \rvert_{}^2  dxdk\; .
  \end{align}
  After some manipulation, taking care of the ordering, the first term
  on the right hand side becomes:
  \begin{align}
    \notag
    \Bigl\langle ue^{-iA_{g_{\infty }}}  , \frac{D_x}{2M}D_x\bigl(ue^{-iA_{g_{\infty }}}\bigr) \Bigr\rangle_2=&\Bigl\langle u  , -\frac{\Delta }{2M}u \Bigr\rangle_2\\
    \label{eq:56}\tag{c}
    &+\frac{1}{2M} \langle A_{r_{\infty }}u  , A_{r_{\infty }}u \rangle_2\\
    \label{eq:57}\tag{d}
    &-i\langle u  , A_{\frac{k^2}{2M}g_{\infty }}u \rangle_2\\
    \label{eq:58}\tag{e}
    &-\frac{1}{M}\Bigl\langle u  , \int_{\mathds{R}^3}^{}dk\Bigl(D_x\bar{r}_{\infty }(k)\alpha (k)e^{ik\cdot x}+ r_{\infty }(k)\bar{\alpha} (k)e^{-ik\cdot x}D_x\Bigr) u\Bigr\rangle_2\; .
  \end{align}
  The proof is concluded making the following identifications (the
  other terms sum to the free part):
  \begin{align*}
    &\text{\eqref{eq:53}}+\text{\eqref{eq:57}}=\frac{1}{(2\pi )^{3/2}}\int_{\mathds{R}^6}^{}\frac{\chi _{\sigma _0}}{\sqrt{2\omega (k)}}\Bigl(\bar{\alpha}(k)e^{-ik\cdot x}+ \alpha (k)e^{ik\cdot x}\Bigr)\lvert u(x)  \rvert_{}^2  dxdk\; ;\\
    &\text{\eqref{eq:54}}=\frac{1}{2}\int_{\mathds{R}^6}^{}V_{\infty}(x-y)\lvert u(x)  \rvert_{}^2\lvert u(y)  \rvert_{}^2  dxdy\; ;\\
    &\text{\eqref{eq:56}}=\frac{1}{2M}\int_{\mathds{R}^9}^{}\Bigl(r_{\infty }(k)\bar{\alpha}(k)e^{-ik\cdot x}+\bar{r}_{\infty }(k)\alpha (k)e^{ik\cdot x} \Bigr) \Bigl(r_{\infty }(l)\bar{\alpha}(l)e^{-il\cdot x}+\bar{r}_{\infty }(l)\alpha (l)e^{il\cdot x} \Bigr)  \lvert u(x)  \rvert_{}^2dxdkdl\; ;\\
    &\text{\eqref{eq:58}}=-\frac{2}{M}\Re\int_{\mathds{R}^6}^{}r_{\infty }(k) \bar{\alpha} (k)e^{-ik\cdot x}\bar{u}(x)D_xu(x)  dxdk\; .
  \end{align*}
\end{proof}

\subsection{Global existence results.}
\label{sec:glob-exist-results}

In this section we discuss uniqueness and global existence of the
classical dynamical system: using a well-known result on the undressed
dynamics, we prove uniqueness and existence also for the dressed
system.

The Cauchy problem associated to $\mathscr{E}$ by the Hamilton's
equations is\footnote{The Cauchy problem associated to
  $\hat{\mathscr{E}}$ is equivalent to \eqref{eq:73}, setting:
  $W=V_{\infty }$,
  $\varphi =(2\pi )^{-3/2}\mathcal{F }(\chi _{\sigma _0})$,
  $\xi =\frac{(2\pi )^{-3/2}}{\sqrt{2}M}\bigl(\mathcal{F
  }(\frac{k^2}{\sqrt{\omega }}g_{\infty })-\mathcal{F }(i
  \frac{k^2}{\omega }g_{\infty }) \bigr)$,
  $\rho =\frac{\sqrt{2}}{M}\mathcal{F }(\sqrt{\omega }kg_{\infty })$,
  and
  $\zeta =\frac{i}{\sqrt{M}}\mathcal{F }(\frac{k}{\sqrt{\omega
    }}g_{\infty })$.}
\eqref{eq:72}. Theorem~\ref{prop:8} below is a straightforward
extension of
\citep{MR2403699
  ,MR2906550
} that includes a (confining) potential on the NLS equation. As proved
in \citep{MR1016082
  ,Carles:2014aa
}, the quadratic potential is the maximum we can afford to still have
Strichartz estimates and global existence in the energy
space. Therefore we make the following standard  assumption on
$V$:
\begin{assumption}{$A_V$}
  \label{ass:2}
  $V\in \mathcal{C }^{\infty }(\mathds{R}^3,\mathds{R}_+)$, and
  $\partial ^{\alpha }V\in L^{\infty } (\mathds{R}^3 )$ for any
  $\alpha \in \mathds{N}^3$, with $\lvert \alpha \rvert_{}^{}\geq 2$
  (i.e. at most quadratic positive confining potential).
\end{assumption}
\begin{thm}[Undressed global existence]
  \label{prop:8}
  Assume \ref{ass:2}. Then there is a unique Hamiltonian
  flow solving \eqref{eq:72}:
  \begin{equation}
    \label{eq:66}
    \mathbf{E}:\mathds{R}\times \bigl(Q(-\Delta +V)\oplus \mathcal{F}H^{\frac{1}{2}}(\mathds{R}^3)\bigr)\to Q(-\Delta +V)\oplus \mathcal{F}H^{\frac{1}{2}}(\mathds{R}^3)\; .
  \end{equation}
  If $V=0$, then there is a unique Hamiltonian flow
  \begin{equation}
    \label{eq:65}
    \mathbf{E}:\mathds{R}\times \bigl(H^s(\mathds{R}^3)\oplus \mathcal{F}H^{\varsigma }(\mathds{R}^3)\bigr)\to H^s (\mathds{R}^3 )\oplus \mathcal{F}H^{\varsigma }(\mathds{R}^3)\; .
  \end{equation}
  for any $0\leq s\leq 1$,
  $s-\frac{1}{2}\leq \varsigma \leq s+\frac{1}{2}$.
\end{thm}
\begin{thm}[Dressed global existence]
  \label{prop:9}
  Assume \ref{ass:2}. Then for any
  $\chi _{\sigma _0}\in L^{\infty } (\mathds{R}^3 )\cap \mathcal{F
  }H^{-\frac{1}{2}}(\mathds{R}^3)$,
  there is a unique Hamiltonian flow:
  \begin{equation}
    \label{eq:68}
    \hat{\mathbf{E}}:\mathds{R}\times \bigl(Q(-\Delta +V)\oplus \mathcal{F}H^{\frac{1}{2}}(\mathds{R}^3)\bigr)\to Q(-\Delta +V)\oplus \mathcal{F}H^{\frac{1}{2}}(\mathds{R}^3)\; .
  \end{equation}
  If $V=0$, then there is a unique Hamiltonian flow
  \begin{equation}
    \label{eq:65a}
    \hat{\mathbf{E}}:\mathds{R}\times \bigl(H^s(\mathds{R}^3)\oplus \mathcal{F}H^{\varsigma }(\mathds{R}^3)\bigr)\to H^s (\mathds{R}^3 )\oplus \mathcal{F}H^{\varsigma }(\mathds{R}^3)\; .
  \end{equation}
  for any $0\leq s\leq 1$,
  $s-\frac{1}{2}\leq \varsigma \leq s+\frac{1}{2}$. For any $V$ that
  satisfies \ref{ass:2}, the flows $\hat{\mathbf{E}}$ and $\mathbf{E}$
  are related by:
  \begin{equation}
    \label{eq:69}
    \hat{\mathbf{E}}=\mathbf{D}_{g_{\infty }}(-1)\circ\mathbf{E}\circ \mathbf{D}_{g_{\infty }}(1)\quad ,\quad \mathbf{E}=\mathbf{D}_{g_{\infty }}(1)\circ \hat{\mathbf{E}} \circ \mathbf{D}_{g_{\infty }}(-1)\quad .
  \end{equation}
\end{thm}
\begin{proof}[Proof of Theorem~\ref{prop:9}]
  The theorem is a direct consequence of the global well-posedness
  result of Theorem~\ref{prop:8}, the relation
  $\hat{\mathscr{E}}=\mathscr{E}\circ\mathbf{D}_{g_{\infty }}(1)$
  proved in Proposition~\ref{prop:7}, and the regularity properties of
  the dressing proved in Proposition~\ref{prop:6}.
\end{proof}

\section{Classical renormalization of the Nelson model.}
\label{sec:class-renorm-nels}

In this section we would like to describe in some detail the classical
approach to renormalization outlined in
Section~\ref{sec:an-altern-renorm}, for the quantization of the S-KG
energy functional $\mathscr{E}$ of Definition~\ref{def:9}.

The Wick quantization
of $\mathscr{E}$---defined on
$L^2 (\mathds{R}^3 )\oplus L^2 (\mathds{R}^3 )$---yields a formal
operator on
$\Gamma _s(L^2 (\mathds{R}^3 )\oplus L^2 (\mathds{R}^3 ))$:
\begin{equation}
  \label{eq:76}
  \begin{split}
    (\mathscr{E})^{Wick}=\int_{\mathds{R}^3}^{}\psi^{*}(x)\Bigl(-\frac{\Delta}{2M}+V(x)\Bigr)\psi(x)  dx+\int_{\mathds{R}^3}^{}a^{*}(k)\omega(k)a(k)  dk\\+\frac{1}{(2\pi)^{3/2}}\int_{\mathds{R}^6}^{}\psi^{*}(x) \frac{1}{\sqrt{2\omega(k)}}\Bigl(a(k)e^{ik\cdot x}+a^{*}(k)e^{-ik\cdot x}\Bigr)\psi(x)  dxdk\; .
  \end{split}
\end{equation}
Such quantity $(\mathscr{E})^{Wick}$ makes sense only as a densely
defined quadratic form
\begin{equation}
  \label{eq:80}
  h(\cdot ,\cdot )=\mathrm{Wick}(\mathscr{E})(\cdot ,\cdot )=\langle \,\cdot\, , (\mathscr{E})^{Wick}\,\cdot\,\rangle_{}\; ,
\end{equation}
because the last term of the right hand side of \eqref{eq:76} creates,
roughly speaking, a Klein-Gordon particle with a wavefunction that is
not square-integrable. In addition, the quadratic form
$h(\cdot ,\cdot )$ is not bounded from below. We should therefore
define a suitable transformation (the classical dressing), to modify
the interaction part of $\mathscr{E}$ in the spirit of normal forms. It turns out that a near identity change of
coordinates in the phase space is sufficient (more precisely, a
non-linear canonical transformation). For any function
$\chi _{\sigma _0}:\mathds{R}^3\to \mathds{C}$ we can define
$\mathbf{D}_{\chi _{\sigma _0}}:=\mathbf{D}_{g_{\infty }}(1)$, where
$\mathbf{D}_{g_{\infty }}(1)$ is given by \eqref{eq:51} with
\begin{equation*}
  g_{\infty}(k)=-i\frac{(2\pi)^{-3/2}}{\sqrt{2\omega(k)}}\frac{1-\chi_{\sigma_0}(k)}{\frac{k^2}{2M}+\omega(k)}\; .
\end{equation*}
The map $\mathbf{D}_{\chi _{\sigma _0}}$ is a candidate to be the
classical dressing, since it cancels the singular term in
$\mathscr{E}$. As we will see, the choice of $\chi _{\sigma _0}$
modifies the self-adjoint Hamiltonian that results from Wick
quantization. The properties of $\mathbf{D}_{\chi _{\sigma _0}}$ have
been discussed in Sections~\ref{sec:classical-dressing}
and~\ref{sec:class-hamilt}. Since it is suitable that
$\mathbf{D}_{\chi _{\sigma _0}}$ map the domain of definition of
$\mathscr{E}$ into itself, we restrict $\chi _{\sigma _0}$ to be in
$ L^{\infty } (\mathds{R}^3 )\cap \mathcal{F
}H^{-\frac{1}{2}}(\mathds{R}^3)$.

Now, we exploit the $U(1)$ symmetry of $\mathscr{E}$ with respect to
the first variable. This symmetry leads to the conservation of the
``mass'' of $u$, i.e. of $\lVert u \rVert_2^2$.  The $U(1)$ symmetry
of $\mathscr{E}$ induces a conservation also at the quantum level: the
Wick quantization of $\mathscr{E}$ ``commutes'' with the Wick
quantization of $\lVert u \rVert_2^2$, and the latter is the
self-adjoint number operator $N_1\otimes 1$ for the nucleons. More
precisely, we have that for any $\Psi ,\Phi \in Q(h)$, for any
$n_1\leq n_2\in \mathds{N}$:
\begin{equation}
  \label{eq:79}
  h\bigl(\Psi ,\mathds{1}_{[n_1,n_2]}(N_1)\Phi \bigr)=h\bigl(\mathds{1}_{[n_1,n_2]}(N_1)\Psi ,\Phi \bigr)=h\bigl(\mathds{1}_{[n_1,n_2]}(N_1)\Psi ,\mathds{1}_{[n_1,n_2]}(N_1)\Phi \bigr)\; ;
\end{equation}
where the quadratic form $h$ is defined by~\eqref{eq:80}, and
$\mathds{1}_{[0,n]}(N_1)$ is a spectral projection of $N_1\otimes 1$.
Let $\mathcal{H}=\Gamma _s(L^2 \oplus L^2)$, it is possible to rewrite
$\mathcal{H}=\bigoplus _n\mathcal{H}_n$ with
$\mathcal{H}_n=\bigl(L^2(\mathds{R}^3)\bigr)^{\otimes_s n}\otimes
\Gamma_s(L^2 (\mathds{R}^3))$.
Equation~\eqref{eq:79} implies that $h$ factorizes into a direct sum
of forms on each $\mathcal{H}_n$, i.e. for any $n\in \mathds{N}$,
there exists a quadratic form $h_n$ densely defined on $\mathcal{H}_n$
such that $h=\bigoplus _{n}h_n$. Since
$\mathbf{D}_{\chi _{\sigma _0}}$ preserves the $U(1)$
invariance\footnote{$\mathscr{E}\circ\mathbf{D}_{\chi _{\sigma
      _0}}(e^{i\phi }u,\alpha )=\mathscr{E}\circ\mathbf{D}_{\chi
    _{\sigma _0}}(u,\alpha )$
  for any $\phi \in \mathds{R}$.}, also the form
$\hat{h}(\chi _{\sigma
  _0})=\mathrm{Wick}(\mathscr{E}\circ\mathbf{D}_{\chi _{\sigma_0}})$
factorizes as a direct sum of forms $\hat{h}_n(\chi _{\sigma _0})$ on
$\mathcal{H}_n$. However, even if we may hope\footnote{This is due to
  the fact that $\mathscr{E}$ is unbounded from below on
  $D(\mathscr{E})$, but bounded from below when restricted to any cylinder 
  with bounded $u$-norm $B_u(0,C)\cap D(\mathscr{E})$; see
  Section~\ref{sec:relat-betw-sigm} below.} that each
$\hat{h}_n(\chi _{\sigma _0})$ is bounded from below, this cannot be
expected for $\hat{h}(\chi _{\sigma _0})$ due to asymptotic
instability of bosonic systems.

Fix $n\in \mathds{N}$; we can prove that, for a suitable choice of
$\chi _{\sigma _0}$, $\hat{h}_n(\chi _{\sigma _0})$ is closed and
bounded from below by the KLMN theorem (it is exactly the quadratic
form associated to $\hat{H}^{(n)}_{\infty }$, and the bounds are
proved in Section~\ref{sec:renorm-hamilt}). The usual choice of the
function $\chi _{\sigma _0}$ is described at the beginning of
Section~\ref{sec:nelson-hamiltonian} (it is, roughly speaking, a
smooth characteristic function of
$0\leq \lvert k \rvert_{}^{}\leq \sigma _0$). We need also
$2K(1+n\varepsilon )< \sigma _0<\infty $, where $K$ is defined in
Theorem~\ref{thm:1}, to apply the KLMN theorem. The choice of
$\sigma _0$ depends on $n$, and thus we cannot define a unique
self-adjoint operator for any $n\in \mathds{N}$. Nevertheless, for any
$\sigma _0\in \mathds{R}_+$, we can define a unique self adjoint
operator $\hat{H}_{ren}(\sigma _0)$ on
$\Gamma _s(L^2 (\mathds{R}^d )\otimes L^2 (\mathds{R}^d ))$ by:
\begin{equation}
  \label{eq:81}
  \hat{H}_{ren}(\sigma _0)=\left\{
    \begin{aligned}
      &(\mathscr{E}\circ \mathbf{D}_{\chi _{\sigma_0}})^{Wick}\bigl\rvert_{\mathcal{H}_n} &\text{ if }n\leq \mathfrak{N}\\
      &0&\text{ if }n>\mathfrak{N}
    \end{aligned}
  \right.\; ,
\end{equation}
where $\mathfrak{N}$ is defined by \eqref{eq:55}, and
$(\mathscr{E}\circ \mathbf{D}_{\chi
  _{\sigma_0}})^{Wick}\bigl\rvert_{\mathcal{H}_n}$
is the unique self-adjoint operator defined by $\hat{h}_n$.

It remains to consider the quantization, if any, of
$\mathbf{D}_{\chi _{\sigma_0}}$. The classical dressing is the
Hamiltonian flow, at time $1$, generated by
$\mathscr{D}_{g_{\infty }}$. As stated in Lemma~\ref{lemma:1},
$(\mathscr{D}_{g_{\infty }})^{Wick}$ is a densely defined self-adjoint
operator on
$\Gamma _s(L^2 (\mathds{R}^d )\oplus L^2 (\mathds{R}^d ))$. Therefore
the quantization of $\mathbf{D}_{\chi _{\sigma_0}}$ would be the
unitary operator
$e^{-\frac{i}{\varepsilon }(\mathscr{D}_{g_{\infty }})^{Wick}}$.
Analogously, for any $\sigma _0\in \mathds{R}_+$ the quantization of
the dressed flow\footnote{In Section~\ref{sec:glob-exist-results}
  above the dressed flow was defined as $\hat{\mathbf{E}}(t)$.}
$\hat{\mathbf{E}}_{\sigma _0}(t)$ would be the renormalized unitary
evolution $e^{-\frac{i}{\varepsilon }t \hat{H}_{ren}(\sigma _0)}$. So
quantizing \eqref{eq:69} we can define the undressed unitary
dynamics for any
$\sigma _0\in\mathds{R}_+$:
\begin{equation}
  \label{eq:82}
  e^{-i\frac{t}{\varepsilon } H_{ren}(\sigma _0)}=e^{-\frac{i}{\varepsilon }(\mathscr{D}_{g_{\infty }})^{Wick}}e^{-i\frac{t}{\varepsilon } \hat{H}_{ren}(\sigma _0)}e^{\frac{i}{\varepsilon }(\mathscr{D}_{g_{\infty }})^{Wick}}\; ;
\end{equation}
and the undressed renormalized Hamiltonian $H_{ren}(\sigma _0)$ as its
generator.

\subsection{The relation between $\sigma_0$ and $\lVert u\rVert_2$.}
\label{sec:relat-betw-sigm}

The parameter $\sigma _0$ is necessary to make the renormalized
Hamiltonian operator self-adjoint; and for each
$\sigma _0\in\mathds{R}_+$ the quantum dynamics may be defined on the
subspace of $\mathcal{H}$ with at most $\mathfrak{N}(\sigma _0)$
nucleons. However, the choice of $\sigma _0$ remains in some sense
arbitrary. In this
section we argue that the number of nucleons is related to the
boundedness from below of the classical energy $\mathscr{E}$.

The classical energy $\mathscr{E}$, when defined on the whole
$D(\mathscr{E})$, is unbounded from below:
\begin{equation}
  \label{eq:77}
  \inf_{(u,\alpha )\in D(\mathscr{E})}\mathscr{E}(u,\alpha )=-\infty \; .
\end{equation}
However, we can take advantage
of the conservation of the $L^2$-norm of $u$. Indeed it can be easily seen
that on each hollow cylinder\footnote{We define a hollow cylinder as
  $S_u(0,C)=\{(u,\alpha )\in L^2 (\mathds{R}^d )\oplus L^2
  (\mathds{R}^d ), \lVert u \rVert_2^{}= C\}$.}
$S_u(0,C)\cap D(\mathscr{E})$ the energy becomes bounded from
below. One can also prove that that the energy is bounded from below on the
cylinder $B_u(0,\sqrt{\mathfrak{C}})\cap D(\mathscr{E})$:
\begin{equation}
  \label{eq:78}
  \inf_{\substack{(u,\alpha )\in D(\mathscr{E})\\\lVert u  \rVert_2^{}\leq \sqrt{\mathfrak{C}} }}\mathscr{E}(u,\alpha )>-\infty \; ,
\end{equation}
and both $B_u(0,\sqrt{\mathfrak{C}})$ and $D(\mathscr{E})$ are
conserved by the Hamiltonian flow $\mathbf{E}$. Therefore it seems
quite natural to consider the restriction of the S-KG theory to
$B_u(0,\sqrt{\mathfrak{C}})$; and in particular to consider classical
probability distributions concentrated in
$B_u(0,\sqrt{\mathfrak{C}})$. Since we would like to quantize our
theory in the Fock space $\Gamma _s(L^2 \oplus L^2)$, we may ask which
are the families of normal quantum states on the Fock space
that have Wigner measures concentrated  on
$B_u(0,\sqrt{\mathfrak{C}})$. The families that satisfy
Assumption~\eqref{eq:84}, or equivalently that have at most $[\mathfrak{C}/\varepsilon] $
nucleons (see Lemma~\ref{lemma:10}), have measures concentrated  on
$B_u(0,\sqrt{\mathfrak{C}})$ according to \citep[Lemma 2.14]{MR2802894} or
\cite[Theorem 3.1]{Ammari:2014ab}. If in addition we take into account that
to a classical theory with $U(1)$ invariance in $u$, corresponds a
quantum theory that preserves the number of nucleons; we
we are led to believe that a quantization of $\mathscr{E}$ should be
meaningful, as an operator bounded from below, on the subspace
$\bigoplus _{n\leq [\mathfrak{C}/\varepsilon ]}\mathcal{H}_n$ of the
Fock representation. As discussed above, we are able to define the
quantization $H_{ren}$ of $\mathscr{E}$ thanks to the classical
dressing renormalization. It is remarkable that we have also the
freedom, by choosing $\sigma _0$, to make $H_{ren}$ self-adjoint and
different from zero precisely on the aforementioned relevant subspace
$\bigoplus _{n\leq [\mathfrak{C}/\varepsilon ]}\mathcal{H}_n$.

We remark that there are families of quantum states with Wigner
measures concentrated on $B_u(0,\sqrt{\mathfrak{C}})$ that have non-zero
components in every subspace with fixed nucleons. In this case
however, the Wick quantization procedure yields an unboundedness from
below of the quantum energy, due to the fact that the number of
nucleons is conserved by the dynamics, but unbounded on the
state. This asymptotic instability or unboundedness from below of the energy, that is characteristic of the bosonic quantum system, seems to prevent the definition of a suitable non-trivial dynamics for these states.

\subsection{Symplectic character of $\mathbf{D}_{\chi_{\sigma_0}}$.}
\label{sec:sympl-prop-class}

To complete our description of the classical renormalization scheme,
we explicitly prove that the classical dressing is a (non-linear)
symplectic map for the real symplectic structure
$\bigl\{(L^2\oplus L^2)_{\mathds{R}},\Im\langle \,\cdot\, , \,\cdot\,
\rangle_{L^2\oplus L^2}\bigr\}$.
We denote by
$\mathrm{d}\mathbf{D}_g(\theta)_{(u,\alpha)}\in \mathcal{L}(L^2\oplus L^2)$ the
(Fréchet) derivative of $\mathbf{D}_g(\theta)$ at the point
$(u,\alpha)\in L^2\oplus L^2$.
\begin{proposition}
  Let $g\in L^2(\mathds{R}^3)$ be an even or odd function. Then for
  any $\theta\in \mathds{R}$, $\mathbf{D}_g(\theta)$ is differentiable
  at any point
  $(u,\alpha)\in L^2(\mathds{R}^3)\oplus L^2(\mathds{R}^3)$. In
  addition, it satisfies for any
  $(v_1,\beta_1),(v_2,\beta_2)\in L^2(\mathds{R}^3)\oplus
  L^2(\mathds{R}^3)$:
  \begin{equation}
    \label{eq:74}
    \Im\langle \mathrm{d}\mathbf{D}_g(\theta)_{(u,\alpha)}(v_1,\beta_1)  , \mathrm{d}\mathbf{D}_g(\theta)_{(u,\alpha)}(v_2,\beta_2) \rangle_{L^2\oplus L^2}=\Im\langle (v_1,\beta_1)  , (v_2,\beta_2) \rangle_{L^2\oplus L^2}\; .
  \end{equation}
\end{proposition}
\begin{proof}
  We recall that with the assumptions on $g$, $\mathbf{D}_g(\theta)$
  has the explicit form:
  \begin{equation*}
    \mathbf{D}_g(\theta )(u(x),\alpha(k) )= \Bigl(\:u(x) e^{-i\theta A_g(x)}\:,\: \alpha(k) -i\theta g(k)F(\lvert u  \rvert_{}^2)(k)\:\Bigr)\; ;
  \end{equation*}
  where $A_g$ and $F$ are defined by Equations~\eqref{eq:42}
  and~\eqref{eq:43} respectively. The Fréchet derivative of
  $\mathbf{D}_g(\theta)$ is easily computed, and yields
  \begin{equation*}
    \begin{split}
      \mathrm{d}\mathbf{D}_g(\theta )_{(u,\alpha)}(v(x),\beta(k))= \Bigl(\:\bigl(v(x) -i\theta B_g(x)u(x)\bigr) e^{-i\theta A_g(x)} \:,\: \beta(k) -2i\theta g(k)\Re \bigl(F(\bar{u}v)(k)\bigr)\:\Bigr)\\= \,\Bigr(\: \mathrm{i}(v,\beta)\:,\: \mathrm{ii}(v,\beta)\:\Bigr)\; ;
    \end{split}
  \end{equation*}
  where we recall that $B_g(x)$ is $A_g(x)$ with $\alpha$ substituted
  by $\beta$. Then we have:
  \begin{gather*}
    \Im\langle \mathrm{i}(v_1,\beta_1)  , \mathrm{i}(v_2,\beta_2) \rangle_{L^2}=\Im\langle v_1  , v_2 \rangle_{L^2}+2\theta\Re\Bigl(\langle B^{(1)}_g u , v_2 \rangle_{L^2}-\langle v_1,B^{(2)}_g u \rangle_{L^2}\Bigr)\; ,\\
    \Im\langle \mathrm{ii}(v_1,\beta_1)  , \mathrm{ii}(v_2,\beta_2) \rangle_{L^2}=\Im\langle \beta_1  , \beta_2 \rangle_{L^2}+2\theta\Re\Bigl(\langle g\,\Re F(\bar{u}v_1), \beta_2 \rangle_{L^2}-\langle\beta_1,g\,\Re F(\bar{u}v_2) \rangle_{L^2}\Bigr)\; .
  \end{gather*}
  The result then follows, noting that
  $\langle g\,\Re F(\bar{u}v_1), \beta_2 \rangle_{L^2}=\langle
  v_1,B^{(2)}_g u \rangle_{L^2}$
  and
  $\langle\beta_1,g\,\Re F(\bar{u}v_2) \rangle_{L^2}=\langle B^{(1)}_g
  u , v_2 \rangle_{L^2}$.
\end{proof}

\section{The classical limit of the renormalized Nelson model.}
\label{sec:class-limit-renorm-1}

In this section we discuss in detail the classical limit of the
renormalized Nelson model, both dressed and undressed. The Subsections
from \ref{sec:integral-formula} to \ref{sec:uniq-solut-transp} are
dedicated to prove the convergence of the dressed dynamics. The
obtained results are summarized by Theorem~\ref{thm:3}. In
Subsection~\ref{sec:class-limit-dress} we study the classical limit of
the dressing transformation. Finally, in
Section~\ref{sec:link-dress-dynam} we prove Theorem~\ref{thm:2} and
overview the results.

\subsection{The integral formula for the dressed Hamiltonian.}
\label{sec:integral-formula}

The results of this and the next subsection are similar in spirit to
the ones previously obtained in \citep[][Section
3]{Ammari:2014aa
} for the Nelson model with cutoff and in \cite[Section
3]{2011arXiv1111.5918A} for the mean field problem. However, some additional care has
to be taken, for in this more singular situation the manipulations
below are allowed only in the sense of quadratic forms. We start with
a couple of preparatory lemmas. The proof of the first can be essentially obtained
 following \citep[Lemma
6.1]{AB}; the second is an equivalent reformulation of Assumption
\eqref{eq:84}:
\begin{equation*}
  \exists \mathfrak{C}>0\, ,\, \forall \varepsilon \in (0,\bar{\varepsilon} )\, ,\, \forall k\in \mathds{N}\, ,\, \Tr[\varrho _{\varepsilon }N_1^k]\leq \mathfrak{C}^k\; .
\end{equation*}
We recall that the Weyl operator $W(\xi )$,
$L^2 \oplus L^2 \ni \xi =\xi _1\oplus \xi _2$, is defined as:
\begin{equation}
  \label{eq:104}
  W(\xi )=e^{\frac{i}{\sqrt{2}}\bigl(\psi ^{*}(\xi _1)+\psi (\xi _1)\bigr)}e^{\frac{i}{\sqrt{2}}\bigl(a ^{*}(\xi _2)+a (\xi _2)\bigr)}\; .
\end{equation}
\begin{lemma}\label{lemma:9}
  For any $\xi =\xi _1\oplus \xi _2$ such that
  $\xi _1\in Q(-\Delta +V)\subset H^1$ and
  $\xi _2\in D(\omega ^{1/2})\equiv \mathcal{F}H^{1/2}$, there exists
  $C(\xi )>0$ that depends only on $\lVert \xi _1 \rVert_{H^1}^{}$
  and $\lVert \xi _2 \rVert_{\mathcal{F}H^{1/2}}^{}$, such that for any
  $\varepsilon \in (0,\bar{\varepsilon} )$:
  \begin{gather*}
    \lVert H_0^{1/2}W(\xi )  \Psi \rVert_{}^{}\leq C(\xi )\lVert (H_0+\bar{\varepsilon} )^{1/2}\Psi   \rVert_{}^{}\; ,\; \forall \Psi \in Q(H_0)\; ;\\
    \lVert (H_0+1)^{1/2}(N_1+1)^{1/2}W(\xi )  \Psi \rVert_{}^{}\leq C(\xi )\lVert (H_0+\bar{\varepsilon} )^{1/2}(N_1+\bar{\varepsilon} )^{1/2}\Psi   \rVert_{}^{}\; ,\; \forall \Psi \in Q(H_0)\cap Q(N_1)\; .
  \end{gather*}
  In an analogous fashion, for any $\xi \in L^2 \oplus L^2$, $r>0$, there
  exist $C(\xi )>0$ that depends only on $\lVert \xi _1 \rVert_2^{}$
  and $\lVert \xi _2 \rVert_2^{}$, such that for any
  $\varepsilon \in (0,\bar{\varepsilon} )$:
  \begin{equation*}
    \lVert (N_1+N_2)^{r/2}W(\xi )  \Psi \rVert_{}^{}\leq C(\xi )\lVert (N_1+N_2+\bar{\varepsilon} )^{r/2}\Psi   \rVert_{}^{}\; ,\; \forall \Psi \in Q(N_1^r)\; .
  \end{equation*}
\end{lemma}

\begin{lemma}\label{lemma:10}
  Let
  $(\varrho _{\varepsilon })_{\varepsilon \in (0,\bar{\varepsilon} )}$
  be a family of normal states on $\mathcal{H}$. Then
  $(\varrho _{\varepsilon })_{\varepsilon \in (0,\bar{\varepsilon} )}$
  satisfies Assumption~\eqref{eq:84} if and only if for any
  $\varepsilon \in (0,\bar{\varepsilon} )$ there exists a sequence
  $(\Psi _i(\varepsilon ))_{i\in \mathds{N}}$ of orthonormal vectors
  in $\mathcal{H}$ with non-zero components only in
  $\bigoplus _{n=0}^{[\mathfrak{C}/\varepsilon ]}\mathcal{H}_n$ and a
  sequence $(\lambda _i(\varepsilon ))_{i\in \mathds{N}}\in l^1$, with
  each $\lambda _i(\varepsilon )>0$, such that:
  \begin{equation*}
    \varrho _{\varepsilon }=\sum_{i\in \mathds{N}}^{}\lambda _i(\varepsilon ) \lvert \Psi _i(\varepsilon )\rangle\langle \Psi _i(\varepsilon )\rvert\; .
  \end{equation*}
  The explicit $\varepsilon $-dependence of $\Psi _i$ and $\lambda _i$
  will be often omitted.
\end{lemma}
\begin{proof}
  We start assuming \eqref{eq:84}. Let
  $\varrho _{\varepsilon }=\sum_{i\in \mathds{N}}^{}\lambda _i \lvert
  \Psi _i\rangle\langle \Psi _i\rvert$
  be the spectral decomposition of $\varrho _{\varepsilon }$. Then
  \begin{equation*}
    \begin{split}
      \Tr\bigl[\varrho _{\varepsilon }N_1^k\bigr]=\sum_{i\in \mathds{N}}^{}\lambda _i \langle \Psi _i  , N^k_1\Psi _i \rangle_{}\leq \mathfrak{C}^k \Rightarrow \sum_{i\in \mathds{N}}^{}\lambda _i\langle \Psi _i  , (N_1/\mathfrak{C})^k\Psi _i \rangle_{}\leq 1\; .
    \end{split}
  \end{equation*}
  Let $\mathds{1}_{[L ,+\infty )}(N_1)$ be the spectral projection of
  $N_1$ on the interval $[L ,+\infty )$; and choose $L
  >\mathfrak{C}$. Then it follows that:
  \begin{equation*}
    \begin{split}
      1\geq \Tr\bigl[\varrho _{\varepsilon }\mathds{1}_{[L ,+\infty )}(N_1) (N_1/\mathfrak{C})^k\bigr]=\sum_{i\in \mathds{N}}^{}\lambda _i \langle \Psi _i  , \mathds{1}_{[L ,+\infty )}(N_1) (N_1/\mathfrak{C})^k\Psi _i \rangle_{}\\\geq \sum_{i\in \mathds{N}}^{}\lambda _i(L/\mathfrak{C})^k\langle \Psi _i  , \mathds{1}_{[L ,+\infty )}(N_1)\Psi _i \rangle_{}\; .
    \end{split}
  \end{equation*}
  Therefore
  $(L/\mathfrak{C})^k\langle \Psi _i , \mathds{1}_{[L ,+\infty
    )}(N_1)\Psi _i \rangle_{}\leq 1$
  for any $k\in \mathds{N}$ and for any $\Psi _i$. Now
  $(L/\mathfrak{C})^k$ diverges when $k\to \infty $, while
  $\langle \Psi _i , \mathds{1}_{[L ,+\infty )}(N_1)\Psi _i
  \rangle_{}$
  does not depend on $k$, so their product is uniformly bounded if and
  only if $\mathds{1}_{[L ,+\infty )}(N_1)\Psi _i=0$ for any
  $L>\mathfrak{C}$. The result follows immediately, recalling that the
  eigenvalues of $N_1$ are of the form $\varepsilon n_1$, with
  $n_1\in \mathds{N}$.

  The converse statement that Assumption~\eqref{eq:84} follows if
  $\varrho _{\varepsilon }=\sum_{i\in \mathds{N}}^{}\lambda _i \lvert
  \Psi _i\rangle\langle \Psi _i\rvert$,
  with each $\Psi _i$ with at most $[\mathfrak{C}/\varepsilon ]$
  particles is trivial to prove.
\end{proof}

In this subsection, we will consider only families  of states
$(\varrho _{\varepsilon })_{\varepsilon \in (0,\bar{\varepsilon} )}$
that satisfy Assumption~\eqref{eq:84} and the following assumption:
\begin{equation}
  \label{eq:arho}\tag{$A'_\rho$}
  \exists C>0  \, ,\, \forall \varepsilon \in (0,\bar{\varepsilon} )\, ,\, \Tr[\varrho _{\varepsilon }(N_1+H_0)]\leq C\; .
\end{equation}

\begin{definition}[$\varrho_{\varepsilon } (t)$, $\tilde{\varrho} _{\varepsilon }(t)$]\label{def:10}
  We define the dressed time evolution of a state
  $\varrho _{\varepsilon }$ to be
  \begin{equation*}
    \varrho _{\varepsilon }(t)=e^{-i \frac{t}{\varepsilon } \hat{H}_{ren}}\,\varrho _{\varepsilon } \,e^{i \frac{t}{\varepsilon } \hat{H}_{ren}}\; ,
  \end{equation*}
  where the dependence on $\sigma _0$ of $\hat{H}_{ren}$ is omitted,
  and the $\sigma _0$ is chosen such that the dynamics is non-trivial
  on the whole subspace with at most $[\mathfrak{C}/\varepsilon ]$
  nucleons (see Lemma~\ref{lemma:10} and the discussion in
  Section~\ref{sec:class-limit-renorm}). We also define the dressed
  evolution in the interaction picture to be
  \begin{equation*}
    \tilde{\varrho} _{\varepsilon }(t)=e^{i \frac{t}{\varepsilon } H_0}\,\varrho _{\varepsilon }(t)\, e^{-i \frac{t}{\varepsilon } H_0}\; .
  \end{equation*}
\end{definition}

To characterize the evolved Wigner measures corresponding to
$\tilde{\varrho} _{\varepsilon }(t)$, it is sufficient to study its
Fourier transform; this is done studying the evolution of
$\Tr[\tilde{\varrho} _{\varepsilon }(t)W(\xi )]$ by means of an
integral equation.

\begin{proposition}\label{prop:10}
  Let
  $(\varrho _{\varepsilon })_{\varepsilon \in (0,\bar{\varepsilon} )}$
  be a family of normal states on $\mathcal{H}$ satisfying
  Assumptions~\eqref{eq:84} and~\eqref{eq:arho}. Then for any
  $t\in \mathds{R}$,
  $Q(-\Delta +V)\oplus D(\omega ^{1/2})\ni \xi =\xi _1\oplus \xi _2$:
  \begin{equation}\label{eq:94}
    \begin{split}
      \Tr\Bigl[\tilde{\varrho}_{\varepsilon }(t)W(\xi ) \Bigr]=\Tr\Bigl[\varrho _{\varepsilon }W(\xi )\Bigr]+\frac{i}{\varepsilon }\int_0^t \Tr\Bigl[\varrho _{\varepsilon }(s)[\hat{H}_{ren,I},W(\tilde{\xi}_s )] \Bigr] \, ds\,,
    \end{split}
  \end{equation}
  where
  $\tilde{\xi}_s=e^{is(-\Delta +V)}\xi _1\oplus e^{-is\omega}\xi _2 $.
  The commutator $[\hat{H}_{ren,I},W(\tilde{\xi}_s )]$ has to be
  intended as a densely defined quadratic form with domain $Q(H_0)$,
  or equivalently as an operator from $Q(H_0)$ to $Q(H_0)^{*}$.
\end{proposition}
\begin{proof}
  The family
  $(\varrho _{\varepsilon })_{\varepsilon \in (0,\bar{\varepsilon} )}$
  satisfies Assumption~\eqref{eq:84}, therefore by
  Lemma~\ref{lemma:10}:
  \begin{equation*}
    \Tr\bigl[\tilde{\varrho}_{\varepsilon }(t)W(\xi ) \bigr]=\sum_{i\in \mathds{N}}^{}\lambda _i\langle e^{i \frac{t}{\varepsilon }H_0}e^{-i \frac{t}{\varepsilon }\hat{H}_{ren}} \Psi _i , W(\xi ) e^{i \frac{t}{\varepsilon }H_0}e^{-i \frac{t}{\varepsilon }\hat{H}_{ren}} \Psi _i\rangle_{}\; .
  \end{equation*}
  By Assumption~\eqref{eq:arho}, it follows that $\Psi _i\in Q(H_0)$
  for any $i\in \mathds{N}$. Hence the right hand side is
  differentiable in $t$ by Lemma~\ref{lemma:9}, since $Q(H_0)$ is the
  form domain of both $H_0$ and $\hat{H}_{ren}$. Using the Duhamel
  formula and the fact that
  $e^{-i \frac{s}{\varepsilon }H_0}W(\xi )e^{i \frac{s}{\varepsilon
    }H_0}=W(\tilde{\xi}_s )$, we then obtain:
  \begin{equation*}
    \begin{split}
      \Tr\bigl[\tilde{\varrho}_{\varepsilon }(t)W(\xi ) \bigr]=\sum_{i\in \mathds{N}}^{}\lambda _i\biggl(\langle \Psi _i  , W(\xi )\Psi _i \rangle_{}+\frac{i}{\varepsilon }\int_0^t\langle e^{-i \frac{s}{\varepsilon }\hat{H}_{ren}}\Psi _i  , [\hat{H}_{ren,I},W(\tilde{\xi}_s )]e^{-i \frac{s}{\varepsilon }\hat{H}_{ren}}\Psi _i \rangle_{}  ds\biggr)\; ;
    \end{split}
  \end{equation*}
  where $[\hat{H}_{ren,I},W(\tilde{\xi}_s )]$ makes sense as a
  quadratic form on $Q(H_0)$. The result is then obtained using
  Lebesgue's dominated convergence theorem on the right hand side, by
  virtue of Assumption~\eqref{eq:arho} and Lemma~\ref{lemma:9}.
\end{proof}

\subsection{The commutator $[\hat{H}_{ren,I},W(\tilde{\xi}_s)]$.}
\label{sec:comm-hath_r-i}

In this subsection we deal with the commutator
$[\hat{H}_{ren,I},W(\tilde{\xi}_s )]$. The goal is to show that each
of its terms converges in the limit $\varepsilon \to 0$, either to
zero or to a suitable phase space symbol.

For convenience, we recall some terminology related to quantization
procedures in infinite dimensional phase spaces (see
\citep{ammari:nier:2008
} for additional informations). Let $\mathcal{Z}$ be a Hilbert space
(the classical phase space).  In the language of quantization, we call
a densely defined functional
$\mathscr{A}: D\subset \mathcal{Z}\to \mathds{C}$ a (classical)
\emph{symbol}. We say that $A$ is a \emph{polynomial symbol} if there are densely defined
bilinear forms $b_{p,q}$ on
$\mathcal{Z}^{\otimes _s p} \times \mathcal{Z}^{\otimes _s q}$,
$0\leq p\leq \bar{p}$, $0\leq q\leq \bar{q}$ (with
$p,\bar{p},q,\bar{q}\in \mathds{N}$) such that
\begin{equation}
  \label{eq:93}
  \mathscr{A}(z)=\sum_{\substack{0\leq p\leq \bar{p}\\0\leq q\leq \bar{q}}}^{} b_{p,q}(z^{\otimes p},z^{\otimes q})\; .
\end{equation}
The Wick quantized quadratic form $(\mathscr{A})^{Wick}$ on
$\Gamma _s(\mathcal{Z})$ is then obtained, roughly speaking, replacing
each $z(\cdot )$ with the annihilation operator valued distribution
$a(\cdot )$; each $\bar{z}(\cdot )$ with the creation operator valued
distribution $a^{*}(\cdot )$; and putting all the $a^{*}(\cdot )$ to
the left of the $a(\cdot )$. We denote, with a straightforward
notation, the class of all polynomial symbols on $\mathcal{Z}$ by
$\bigoplus^{alg}_{(p,q)\in \mathds{N}^2}
\mathcal{Q}_{p,q}(\mathcal{Z})$.
If $\mathscr{A}:\mathcal{Z}\to \mathds{C}$ and the bilinear forms
$b_{p,q}(z^{\otimes p},z^{\otimes q})$ in \eqref{eq:93} can
all be written as
$\langle z^{\otimes q} , \tilde{b}_{p,q}z^{\otimes p}
\rangle_{\mathcal{Z}^{\otimes_s q }}$
for some bounded (resp.~compact) operator
$\tilde{b}_{p,q}:\mathcal{Z}^{\otimes_s p }\to \mathcal{Z}^{\otimes_s
  q }$,
we say that $\mathscr{A}$ is a bounded (resp.~compact) polynomial symbol. We denote the class of
all bounded (resp.~compact)  polynomial symbols by $\bigoplus^{alg}_{(p,q)\in \mathds{N}^2} \mathcal{P}_{p,q}(\mathcal{Z})$  \bigg(resp.
$\bigoplus^{alg}_{(p,q)\in \mathds{N}^2} \mathcal{P}_{p,q}^{\infty
}(\mathcal{Z})$\bigg).
We remark that $\mathscr{E}$, $\hat{\mathscr{E}}$ and $\mathscr{D}_g$
defined in Section~\ref{sec:classical-system} are all polynomial
symbols\footnote{In $L^2 (\mathds{R}^3 )\oplus L^2 (\mathds{R}^3 )$,
  we adopt the notation $z=(u,\alpha )$; and to each $u(x)$ it
  corresponds the operator valued distribution $\psi(x) $, to each
  $\alpha (k)$ the distribution $a(k)$. The Wick quantization is again
  obtained by substituting each $\bigl(u^{\#}(x),\alpha^{\#}(k) \bigr)$
  with $\bigl(\psi^{\#}(x) ,a^{\#}(k)\bigr)$, and using the normal
  ordering of creators to the left of annihilators.} on
$L^2\oplus L^2 $.

\begin{lemma}\label{lemma:12}
  Let
  $(\varrho _{\varepsilon })_{\varepsilon \in (0,\bar{\varepsilon} )}$
  satisfy the same assumptions as in Proposition~\ref{prop:10}. Then
  there exist maps
  $\mathscr{B}_j(\cdot ):Q(-\Delta +V)\oplus D(\omega ^{1/2})\to
  \bigoplus^{alg}_{(p,q)\in \mathds{N}^2}\mathcal{Q}_{p,q}\bigl(L^2
  \oplus L^2 \bigr)$,
  $j=0,\dotsc,3$, such that for any $t\in \mathds{R}$,
  $\xi \in Q(-\Delta +V)\oplus D(\omega ^{1/2})$:
  \begin{equation}\label{eq:92}
    \begin{split}
      \Tr\Bigl[\tilde{\varrho}_{\varepsilon }(t)W(\xi ) \Bigr]=\Tr\Bigl[\varrho _{\varepsilon }W(\xi )\Bigr]+\sum_{j=0}^3\varepsilon ^j\int_0^t \Tr\Bigl[\varrho _{\varepsilon }(s)W(\tilde{\xi}_s) \bigl(\mathscr{B}_j(\tilde{\xi}_s)\bigr)^{Wick} \Bigr]  ds\\
      =\Tr\Bigl[\varrho _{\varepsilon }W(\xi )\Bigr]+\sum_{j=0}^3\varepsilon ^j\int_0^t \Tr\Bigl[\varrho _{\varepsilon }(s)W(\tilde{\xi}_s) B_j(\tilde{\xi}_s) \Bigr]  ds\; ;
    \end{split}
  \end{equation}
  where the $\bigl(\mathscr{B}_j(\tilde{\xi}_s)\bigr)^{Wick}$ make
  sense as densely defined quadratic forms. To simplify the notation,
  we have set $B_j(\cdot ):=\bigl(\mathscr{B}_j(\cdot )\bigr)^{Wick}$.
\end{lemma}
\begin{proof} We only sketch the proof here since it follows the same lines as in \cite{ 2011arXiv1111.5918A,liard:2015}.   By \eqref{eq:94}, the only thing we have to prove is that,
  in the sense of quadratic forms,
  $\frac{i}{\varepsilon }[\hat{H}_{ren,I},W(\tilde{\xi}_s
  )]=\sum_{j=0}^3 W(\tilde{\xi}_s) B_j(\tilde{\xi}_s)$.
  First of all, we remark that $\hat{H}_{ren,I}$ is the Wick
  quantization of a polynomial symbol\footnote{To be precise, we are
    considering here the quadratic form $\hat{h}_{ren,I}$, defined and
    different from zero on the whole space $\mathcal{H}$, since it agrees
    with
    $\langle \,\cdot \, ,\, \hat{H}_{ren,I}\, \cdot \, \rangle_{}$
    when restricted to vectors that belong to
    $\bigoplus _{n\leq [\mathfrak{C}/\varepsilon ]}\mathcal{H}_n$
    (being here the case by Lemma~\ref{lemma:10}).}; i.e.
  $\hat{H}_{ren,I}=\bigl(\hat{\mathscr{E}}_I \bigr)^{Wick}$, with
  \begin{equation}
    \label{eq:95}
    \begin{split}
      \hat{\mathscr{E}}_I(u,\alpha )=\frac{1}{(2\pi )^{3/2}}\int_{\mathds{R}^6}^{}\frac{\chi _{\sigma _0}}{\sqrt{2\omega (k)}}\Bigl(\bar{\alpha}(k)e^{-ik\cdot x}+ \alpha (k)e^{ik\cdot x}\Bigr)\lvert u(x)  \rvert_{}^2  dxdk \\+\frac{1}{2M}\int_{\mathds{R}^9}^{}\Bigl(r_{\infty }(k)\bar{\alpha}(k)e^{-ik\cdot x}+\bar{r}_{\infty }(k)\alpha (k)e^{ik\cdot x} \Bigr) \Bigl(r_{\infty }(l)\bar{\alpha}(l)e^{-il\cdot x}\\+\bar{r}_{\infty }(l)\alpha (l)e^{il\cdot x} \Bigr)  \lvert u(x)  \rvert_{}^2dxdkdl\\-\frac{2}{M}\Re\int_{\mathds{R}^6}^{}r_{\infty }(k) \bar{\alpha} (k)e^{-ik\cdot x}\bar{u}(x)D_xu(x)  dxdk +\frac{1}{2}\int_{\mathds{R}^6}^{}V_{\infty}(x-y)\lvert u(x)  \rvert_{}^2\lvert u(y)  \rvert_{}^2  dxdy\; .
    \end{split}
  \end{equation}
  We also recall, according to \citep[Proposition 2.10 for bounded polynomial
  symbols]{ammari:nier:2008
  } and \cite[Proposition 2.1.30 for the general case]{liard:2015}, that essentially for any
  $\mathscr{A}\in \bigoplus^{alg}_{(p,q)\in \mathds{N}^2}
  \mathcal{Q}_{p,q}(L^2\oplus L^2)$ the following formula is true, in the sense of forms, for any
  suitably regular $\xi \in L^2\oplus L^2$:
  \begin{equation}
    \label{eq:96}
    W^{*}(\xi )\bigl(\mathscr{A}\bigr)^{Wick}W(\xi )=\Bigl(\mathscr{A}\bigl(\cdot + \tfrac{i\varepsilon }{\sqrt{2}}\xi \bigr)\Bigr)^{Wick}\; .
  \end{equation}
  Roughly speaking, the Weyl operators $W(\xi )$ translate each
  creation/annihilation operator by
  $\mp \tfrac{i\varepsilon }{\sqrt{2}}\xi$. The result then follows
  immediately on the states $\varrho_{\varepsilon}(s)$:
  \begin{equation*}
    [\hat{H}_{ren,I},W(\tilde{\xi}_s)]=W(\tilde{\xi}_s )\bigl(W^{*}(\tilde{\xi}_s )\hat{H}_{ren,I}W(\tilde{\xi}_s )-\hat{H}_{ren,I}\bigr)=W(\tilde{\xi}_s )\Bigl(\hat{\mathscr{E}}_I(\cdot + \tfrac{i\varepsilon }{\sqrt{2}}\tilde{\xi}_s)-\hat{\mathscr{E}}_I(\cdot )\Bigr)^{Wick}\; ;
  \end{equation*}
  finally we define
  $\sum_{j=0}^{3}\varepsilon ^j\mathscr{B}_j(\xi)
  (z)=\frac{i}{\varepsilon }\Bigl(\hat{\mathscr{E}}_I(z +
  \tfrac{i\varepsilon }{\sqrt{2}}\xi )-\hat{\mathscr{E}}_I(z)\Bigr)$,
  to factor out the $\varepsilon $-dependence.
\end{proof}

We state the next lemma without giving the tedious proof, that is
based on the same type of estimates given in
Section~\ref{sec:renorm-hamilt} for the full operator
$\hat{H}_{ren,I}$.

\begin{lemma}\label{lemma:11}
  For any $j=0,1,2,3$, $\xi \in Q(-\Delta +V)\cap D(\omega ^{1/2})$
  and $\mathfrak{C}>0$, there exists
  $C_j(\xi )>0$ such that for any
  $\Phi ,\Psi \in D(H_0^{1/2})\cap D(N_1)$, with $\Phi$ or $\Psi$ in
$\bigoplus _{n=0}^{[\mathfrak{C}/\varepsilon ]}\mathcal{H}_n$ and for
any $s\in \mathds{R}$ and $\varepsilon \in (0,\bar{\varepsilon} )$:
  \begin{equation}
    \label{eq:97}
    \lvert \langle \Phi   , B_j(\tilde{\xi}_s ) \Psi \rangle  \rvert_{}^{}\leq C_{j}(\xi )\lVert (N_1+H_0+\bar{\varepsilon} )^{1/2}\Phi   \rVert_{}^{}\cdot \lVert (N_1+H_0+\bar{\varepsilon} )^{1/2}\Psi   \rVert_{}^{}\; .
  \end{equation}
\end{lemma}

Thanks to this lemma we are now in a position to prove that the higher
order terms in $\varepsilon $ of Equation~\eqref{eq:92} (namely those
with $j>0$) vanish in the limit $\varepsilon \to 0$.

\begin{proposition}\label{prop:11}
  Let
  $(\varrho _{\varepsilon })_{\varepsilon \in (0,\bar{\varepsilon} )}$
  satisfy Assumptions~\eqref{eq:84} and~\eqref{eq:arho}; let
  $\xi \in Q(-\Delta +V)\cap D(\omega ^{1/2})$. Then the following
  limit holds for any $t\in \mathds{R}$:
  \begin{equation}
    \label{eq:98}
    \lim_{\varepsilon \to 0}\sum_{j=1}^3\varepsilon ^j\int_0^t \Tr\Bigl[\varrho _{\varepsilon }(s)W(\tilde{\xi}_s) B_j(\tilde{\xi}_s) \Bigr]  ds=0\; .
  \end{equation}
\end{proposition}
\begin{proof}
  By Lemma~\ref{lemma:10} we can write
  $\varrho _{\varepsilon }=\sum_{i}^{}\lambda _i\lvert\Psi
  _i\rangle\langle\Psi_i \rangle$,
  where each $\Psi _i$ has non-zero components only in the subspace
  $\bigoplus _{n\leq [\mathfrak{C}/\varepsilon ]}\mathcal{H}_n$, and
  each $\lambda _i>0$. Assumption~\eqref{eq:arho} then translates on
  the fact that each $\Psi _i$ is on the domain $Q(H_0)\cap Q(N_1)$,
  and in addition
  $\sum_i^{}\lambda _i\langle \Psi _i , (N_1+H_0)\Psi _i
  \rangle_{}\leq C$,
  uniformly with respect to $\varepsilon\in (0,\bar{\varepsilon} )
  $. Therefore we can write
  \begin{equation*}
    \Biggl\lvert \sum_{j=1}^3\varepsilon ^j\int_0^t \Tr\Bigl[\varrho _{\varepsilon }(s)W(\tilde{\xi}_s) B_j(\tilde{\xi}_s) \Bigr]  ds  \Biggr\rvert_{}^{}\leq \sum_{j=1}^3\varepsilon ^j\sum_i^{}\lambda _i\int_0^t\biggl\lvert \langle W^{*}(\tilde{\xi}_s)e^{-i \frac{s}{\varepsilon } \hat{H}_{ren}}\Psi _i  , B_j(\tilde{\xi}_s) e^{-i \frac{s}{\varepsilon } \hat{H}_{ren}}\Psi _i \rangle_{}  \biggr\rvert_{}^{}  ds\; .
  \end{equation*}
  Using now Lemma~\ref{lemma:11} and then Lemma~\ref{lemma:9} and the
  fact that $N_1$ commutes with $\hat{H}_{ren}$ we obtain
  \begin{equation*}
    \begin{split}
      \Biggl\lvert \sum_{j=1}^3\varepsilon ^j\int_0^t \Tr\Bigl[\varrho _{\varepsilon }(s)W(\tilde{\xi}_s) B_j(\tilde{\xi}_s) \Bigr]  ds  \Biggr\rvert_{}^{}\leq \sum_{j=1}^3\varepsilon ^jC_j(\xi )\sum_i^{}\lambda _i\int_0^t\lVert (N_1+H_0+\bar{\varepsilon} )^{1/2}W^{*}(\tilde{\xi}_s)e^{-i \frac{s}{\varepsilon } \hat{H}_{ren}}\Psi _i  \rVert_{}^{}\\\cdot \lVert (N_1+H_0+\bar{\varepsilon} )^{1/2}e^{-i \frac{s}{\varepsilon } \hat{H}_{ren}}\Psi _i  \rVert_{}^{}  ds\\
      \leq \sum_{j=1}^3\varepsilon ^jC(\xi )C_j(\xi )\sum_i^{}\lambda _i\int_0^t\langle e^{-i \frac{s}{\varepsilon } \hat{H}_{ren}}\Psi _i,(N_1+H_0+\bar{\varepsilon} )e^{-i \frac{s}{\varepsilon } \hat{H}_{ren}}\Psi _i  \rangle_{}^{}ds\\
      \leq \sum_{j=1}^3\varepsilon ^jC(\xi )C_j(\xi )\sum_i^{}\lambda _i\biggl(t\langle \Psi _i,(N_1+\bar{\varepsilon}) \Psi _i  \rangle_{}^{}+\int_0^t\langle e^{-i \frac{s}{\varepsilon } \hat{H}_{ren}}\Psi _i,H_0e^{-i \frac{s}{\varepsilon } \hat{H}_{ren}}\Psi _i  \rangle_{}^{}ds\biggr)\; .
    \end{split}
  \end{equation*}
  Now we consider the term
  $\langle e^{-i \frac{s}{\varepsilon } \hat{H}_{ren}}\Psi _i,H_0e^{-i
    \frac{s}{\varepsilon } \hat{H}_{ren}}\Psi _i \rangle_{}^{}$.
  First of all we write it as
  \begin{equation}\label{eq:100}
    \begin{split}
      \langle e^{-i \frac{s}{\varepsilon } \hat{H}_{ren}}\Psi _i,H_0e^{-i\frac{s}{\varepsilon } \hat{H}_{ren}}\Psi _i \rangle_{}^{}=\langle e^{-i \frac{s}{\varepsilon } \hat{H}_{ren}}\Psi _i,(\hat{H}_{ren}-\hat{H}_{ren,I})e^{-i\frac{s}{\varepsilon } \hat{H}_{ren}}\Psi _i \rangle_{}^{}\\=\sum_{n=0}^{[\mathfrak{C}/\varepsilon ]}\langle e^{-i \frac{s}{\varepsilon } \hat{H}_{\infty }^{(n)}}\Psi^{(n)} _i,\Bigl(\hat{H}^{(n)}_{\infty }-\hat{H}_{I}^{(n)}(\infty )\Bigr)e^{-i\frac{s}{\varepsilon } \hat{H}_{\infty }^{(n)}}\Psi _i^{(n)} \rangle_{}^{}\\
      = \sum_{n=0}^{[\mathfrak{C}/\varepsilon ]}\langle \Psi ^{(n)}_i  , \hat{H}^{(n)}_{\infty }\Psi ^{(n)}_i \rangle_{}-\langle e^{-i \frac{s}{\varepsilon } \hat{H}_{\infty }^{(n)}}\Psi^{(n)} _i,\hat{H}_{I}^{(n)}(\infty )e^{-i\frac{s}{\varepsilon } \hat{H}_{\infty }^{(n)}}\Psi _i^{(n)} \rangle_{}^{}\\
      \leq  \sum_{n=0}^{[\mathfrak{C}/\varepsilon ]}\Biggl(\Bigl\lvert \langle \Psi ^{(n)}_i  , \hat{H}^{(n)}_{\infty }\Psi ^{(n)}_i \rangle_{}  \Bigr\rvert_{}^{}+\Bigl\lvert  \langle e^{-i \frac{s}{\varepsilon } \hat{H}_{\infty }^{(n)}}\Psi^{(n)} _i,\hat{H}_{I}^{(n)}(\infty )e^{-i\frac{s}{\varepsilon } \hat{H}_{\infty }^{(n)}}\Psi _i^{(n)} \rangle_{}^{} \Bigr\rvert_{}^{}\Biggr)\; .
    \end{split}
  \end{equation}
  The idea now is to use the bound of Equation~\eqref{eq:31} on
  $\Bigl\lvert \langle e^{-i \frac{s}{\varepsilon } \hat{H}_{\infty
    }^{(n)}}\Psi^{(n)} _i,\hat{H}_{I}^{(n)}(\infty
  )e^{-i\frac{s}{\varepsilon } \hat{H}_{\infty }^{(n)}}\Psi _i^{(n)}
  \rangle_{}^{} \Bigr\rvert_{}^{}$.
  The crucial point is that since we have chosen $\sigma _0$ such that
  the dynamics is non-trivial for any
  $n\leq [\mathfrak{C}/\varepsilon ]$, it follows that there exist an
  $a<1$ and a $b<\infty $ \emph{both independent of $\varepsilon $ and
    $n$} such that the bound~\eqref{eq:31} holds for any
  $n\leq [\mathfrak{C}/\varepsilon ]$. Therefore we obtain
  \begin{equation}\label{eq:102}
    \begin{split}
      \langle e^{-i \frac{s}{\varepsilon } \hat{H}_{ren}}\Psi _i,H_0e^{-i\frac{s}{\varepsilon } \hat{H}_{ren}}\Psi _i \rangle_{}^{}\leq a \langle e^{-i \frac{s}{\varepsilon } \hat{H}_{ren}}\Psi _i,H_0e^{-i\frac{s}{\varepsilon } \hat{H}_{ren}}\Psi _i \rangle_{}^{} + b\langle \Psi _i  , \Psi _i \rangle_{}\\+\sum_{n=0}^{[\mathfrak{C}/\varepsilon ]}\Bigl\lvert \langle \Psi ^{(n)}_i  , \hat{H}^{(n)}_{\infty }\Psi ^{(n)}_i \rangle_{}  \Bigr\rvert_{}^{}\; .
    \end{split}
  \end{equation}
  Now since $a<1$, we may take it to the left hand side and use
  again~\eqref{eq:31} on
  $\Bigl\lvert \langle \Psi ^{(n)}_i , \hat{H}^{(n)}_{\infty }\Psi
  ^{(n)}_i \rangle_{} \Bigr\rvert_{}^{}$:
  \begin{equation}\label{eq:101}
    \begin{split}
      \langle e^{-i \frac{s}{\varepsilon } \hat{H}_{ren}}\Psi _i,H_0e^{-i\frac{s}{\varepsilon } \hat{H}_{ren}}\Psi _i \rangle_{}^{}\leq \frac{1}{1-a}\langle \Psi _i  , H_0\Psi _i \rangle_{}+\frac{2b}{1-a}\langle \Psi _i  , \Psi _i \rangle_{}\; .
    \end{split}
  \end{equation}
  Finally, since the state is normalized (i.e.
  $\sum_i^{}\lambda _i\langle \Psi _i , \Psi _i \rangle_{}=1$), we
  conclude:
  \begin{equation*}
    \begin{split}
      \Biggl\lvert \sum_{j=1}^3\varepsilon ^j\int_0^t \Tr\Bigl[\varrho _{\varepsilon }(s)W(\tilde{\xi}_s) B_j(\tilde{\xi}_s) \Bigr]  ds  \Biggr\rvert_{}^{}\leq t\sum_{j=1}^3\varepsilon ^jC(\xi )C_j(\xi )\sum_i^{}\lambda _i\biggl(\langle \Psi _i,N_1\Psi _i  \rangle_{}^{}+\tfrac{1}{1-a}\langle \Psi _i  , H_0\Psi _i \rangle_{}\\+(\tfrac{2b}{1-a}+\bar{\varepsilon} )\langle \Psi _i  , \Psi _i \rangle_{}\biggr)\\
      \leq t\sum_{j=1}^3\varepsilon ^jC(\xi )C_j(\xi )\biggl(\Bigl(1+\tfrac{1}{1-a}\Bigr)\sum_i^{}\lambda _i\langle \Psi _i,(N_1+H_0)\Psi _i  \rangle_{}^{}+\tfrac{2b}{1-a}+\bar{\varepsilon} \biggr)\\
      \leq t\sum_{j=1}^3\varepsilon ^jC(\xi )C_j(\xi )\biggl(\Bigl(1+\tfrac{1}{1-a}\Bigr)C+\tfrac{2b}{1-a}+\bar{\varepsilon} \biggr)\; .
    \end{split}
  \end{equation*}
  The right hand side has no implicit dependence on $\varepsilon $, so
  it converges to zero when $\varepsilon \to 0$.
\end{proof}
By the same argument used  from \eqref{eq:100} to
\eqref{eq:101} above, we can prove the following useful
lemma.
\begin{lemma}\label{lemma:13}
  If a family of states
  $(\varrho _{\varepsilon })_{\varepsilon \in (0,\bar{\varepsilon} )}$
  satisfies Assumptions~\eqref{eq:84} and \eqref{eq:arho}, then for
  any $t\in \mathds{R}$,
  $\bigl(\varrho _{\varepsilon }(t)\bigr)_{\varepsilon \in
    (0,\bar{\varepsilon} )}$
  and
  $\bigl(\tilde{\varrho} _{\varepsilon }(t)\bigr)_{\varepsilon \in
    (0,\bar{\varepsilon} )}$
  satisfy Assumptions~\eqref{eq:84} and \eqref{eq:arho}. In
  particular, there exist $a(\mathfrak{C})<1$ and $b(\mathfrak{C})>0$
  such that uniformly on $\varepsilon \in (0,\bar{\varepsilon} )$:
  \begin{align}
    \label{eq:103}
    \Tr[&\varrho _{\varepsilon }(t)N_1^k]\leq \mathfrak{C}^k\; ,\; \forall k\in \mathds{N}\;;\\
\label{eq:106}
\Tr[&\varrho _{\varepsilon }(t)(N_1+H_0)]\leq \frac{1}{1-a(\mathfrak{C})}C+\frac{2b(\mathfrak{C})}{1-a(\mathfrak{C})}\; ;
  \end{align}
  and the same bounds hold for
  $\bigl(\tilde{\varrho} _{\varepsilon }(t)\bigr)_{\varepsilon \in
    (0,\bar{\varepsilon} )}$.
\end{lemma}

It remains to study the limit of the $B_0(\cdot )$-term in
\eqref{eq:92}. As already pointed out in
Lemma~\ref{lemma:12}, we know that $B_0$ is a Wick quantization.  More
precisely, there exist a densely defined map from the one-particle
space to polynomial symbols in
$\bigoplus_{(p,q) \in \{(i,j)\rvert 0\leq i,j\leq 2;2\leq i+j\leq
  3\}}\mathcal{Q}_{p,q}\bigl(L^2 \oplus L^2\bigr)$.
In order to apply the convergence results of
\citet{ammari:nier:2008
}, we need to show that the symbol of $B_0$ may be approximated by a
compact one, with an error that vanishes in the limit
$\varepsilon \to 0$.

To improve readability, we will write $B_0(\xi )$ in a schematic
fashion. The precise structure of each term will be discussed and
analyzed in the proof of the sequent proposition. In addition, as seen
in Equation~\eqref{eq:12}, the dressed interaction quadratic form
$\hat{H}_I(\infty )$ can be split in three terms: the first is just
the interaction term $H_I(\sigma _0)$ of the Nelson model with cutoff
(with $\sigma$ replaced by $\sigma_0$), whose classical limit has been
analyzed by the authors in
\citep{Ammari:2014aa
}; the second is a ``mean-field'' term for the nucleons, of the same
type as the ones analyzed by Ammari and Nier in
\citep{2011arXiv1111.5918A
}; the last one has a structure similar to the interaction part of the
Pauli-Fierz model \citep[see
e.g.][]{BFS,BFS1,BFS2,MR2097788
}, and thus will be called of ``Pauli-Fierz type''. We will
concentrate on the analysis of the Pauli-Fierz type terms of $B_0$,
while for a precise treatment of the others the reader may refer to
\citep{Ammari:2014aa
  ,2011arXiv1111.5918A
}. In order to highlight the different parts of
$B_0(\xi )=B_0(\xi _1,\xi _2)$, we will underline with different
style and color the
\textcolor{OliveGreen}{\underline{\textcolor{black}{Nelson}}},
\textcolor{Fuchsia}{\dotuline{\textcolor{black}{mean-field}}} and
\textcolor{YellowOrange}{\dashuline{\textcolor{black}{Pauli-Fierz}}}
type terms:
\begin{equation}
  \label{eq:91}
  \begin{split}
    B_0(\xi _1,\xi _2 )=\bigl(\mathscr{B}_0(\xi _1,\xi _2) \bigr)^{Wick}=\textcolor{OliveGreen}{\underline{\textcolor{black}{(a^{*}+a)(\xi _1\psi ^{*}-\bar{\xi}_1\psi  )+\Im(\xi _2)\psi ^{*}\psi} }}+\textcolor{Fuchsia}{\dotuline{\textcolor{black}{ \bar{\xi} _{1}\psi ^{*}\psi \psi-\xi _1\psi ^{*}\psi ^{*}\psi}}}\\+\textcolor{YellowOrange}{\dashuline{\textcolor{black}{(a^{*}a^{*}+aa+a^{*}a)(\xi _1\psi ^{*}-\bar{\xi}_1\psi  )+(\xi _2a^{*}-\bar{\xi}_2a )\psi ^{*}\psi +(a^{*}D_x +D_xa)(\psi ^{*}\xi_1 -\bar{\xi}_1\psi  ) }}}\\\textcolor{YellowOrange}{\dashuline{\textcolor{black}{+\psi ^{*}D_x\xi _2\psi -\psi ^{*}\bar{\xi}_2D_x\psi  }}}\; .
  \end{split}
\end{equation}

\begin{proposition}\label{prop:12}
  There exists a family of maps
  $(\mathscr{B}_0^{(m)})_{m\in \mathds{N}}$ such that:
  \begin{itemize}[label=\color{myblue}*]
  \item For any $m\in \mathds{N}$
    \begin{equation*}
      \mathscr{B}_0^{(m)}(\cdot ):Q(-\Delta +V)\oplus D(\omega ^{3/4})\to\bigoplus_{\substack{(p,q) \in \{(i,j)\rvert0\leq i,j\leq 2;2\leq i+j\leq3\}}}\mathcal{P}^{\infty }_{p,q}\bigl(L^2\oplus L^2\bigr)\; ;
    \end{equation*}
  \item For any $\xi \in Q(-\Delta +V)\oplus D(\omega ^{3/4})$, there
    exist a sequence $\bigl(C^{(m)}(\xi )\bigr)_{m\in \mathds{N}}$ that depends
    only on
    $\lVert \xi \rVert_{Q(-\Delta +V)\oplus D(\omega ^{3/4})}^{}$ such
    that $\lim_{m\to \infty }C^{(m)}= 0$; and such that for any two
    vectors $\Phi,\Psi \in \mathcal{H}\cap D(N_1)$, and for any
    $\varepsilon \in (0,\bar{\varepsilon} )$:
    \begin{equation}\label{eq:99}
      \begin{split}
        \Bigl\lvert  \Bigl\langle (H_0+1)^{-1/2}\Phi   ,  \bigl(\mathscr{B}_{0}(\xi )-\mathscr{B}_0^{(m)}(\xi ) \bigr)^{Wick}(H_0+1)^{-1/2}\Psi \Bigr\rangle_{} \Bigr\rvert_{}^{}\leq C^{(m)}(\xi )\bigl\lVert (N_1+\bar{\varepsilon} )^{1/2}\Phi    \bigr\rVert_{}^{}\\\cdot\, \bigl\lVert  (N_1+\bar{\varepsilon} )^{1/2}\Psi  \bigr\rVert_{}^{}\; .
      \end{split}
    \end{equation}
  \end{itemize}
\end{proposition}
\begin{remark}\label{rem:9}
  Contrarily to what it was previously assumed throughout
  Section~\ref{sec:class-limit-renorm-1}, in this proposition we need
  additional regularity on $\xi _2$, namely
  $\xi _2\in D(\omega ^{3/4})\subset D(\omega ^{1/2})$. This will not
  be a problem in the following, since we will extend our results to
  any $\xi \in L^2 (\mathds{R}^3 )\oplus L^2 (\mathds{R}^3 )$ by a
  density argument, and $D(\omega ^{3/4})$ is still dense in
  $L^2 (\mathds{R}^3 )$.
\end{remark}
\begin{proof}[Proof of Proposition~\ref{prop:12}]
  To prove the proposition, we need to analyze each term of
  Equation~\eqref{eq:91}, and prove that either it has a compact
  symbol or it can be approximated by one, in a way that
  \eqref{eq:99} holds. The analysis for
  \textcolor{OliveGreen}{\underline{\textcolor{black}{Nelson}}} terms
  has been carried out in \citep[][Proposition 3.11 and Lemma
  3.15]{Ammari:2014aa
  }. In addition, using Lemma~\ref{lemma:2} we see that $V_{\infty }$
  satisfies the hypotheses of the mean field potentials in
  \citep{2011arXiv1111.5918A
  }, therefore \eqref{eq:99} holds also for the
  \textcolor{Fuchsia}{\dotuline{\textcolor{black}{mean-field}}} terms
  \citep[see in particular Section 3.2
  of][]{2011arXiv1111.5918A
  }. For the sake of completeness, we explicitly write the Nelson and
  mean field part of Equation~\eqref{eq:91}:
  \begin{equation*}
    \textcolor{OliveGreen}{\underline{\textcolor{black}{(a^{*}+a)(\xi _1\psi ^{*}-\bar{\xi}_1\psi  )}}}=-\tfrac{1}{\sqrt{2}(2\pi)^{3/2}}\int_{\mathds{R}^3}^{}\Bigl(a^{*}\bigl(\tfrac{e^{-ik\cdot x}}{\sqrt{2\omega}}\chi_{\sigma_0}\bigr)+a\bigl(\tfrac{e^{-ik\cdot x}}{\sqrt{2\omega}}\chi_{\sigma_0}\bigr)\Bigr)\Bigl(\xi_1(x)\psi ^{*}(x)-\bar{\xi}_1(x)\psi (x)  \Bigr)  dx\; ;
  \end{equation*}
  \begin{equation*}
    \textcolor{OliveGreen}{\underline{\textcolor{black}{\Im(\xi _2)\psi ^{*}\psi}}}=-\tfrac{1}{\sqrt{2}(2\pi)^{3/2}}\int_{\mathds{R}^6}^{}\biggl(\tfrac{\chi _{\sigma _0}(k)}{\sqrt{2\omega (k)}}\Bigl(\xi _2(k)e^{ik\cdot x}-\bar{\xi}_2(k)e^{-ik\cdot x} \Bigr)\biggr)\psi ^{*}(x)\psi (x)  dxdk\; ;
  \end{equation*}
  \begin{equation*}
    \textcolor{Fuchsia}{\dotuline{\textcolor{black}{\bar{\xi} _{1}\psi ^{*}\psi \psi-\xi _1\psi ^{*}\psi ^{*}\psi}}}= \tfrac{1}{\sqrt{2}}\int_{\mathds{R}^6}^{}V_{\infty }(x-y)\Bigl(\bar{\xi}_1(y) \psi ^{*}(x)\psi (x)\psi (y)-\xi_1(y) \psi ^{*}(x)\psi ^{*}(y)\psi (x)\Bigr)   dxdy\; .
  \end{equation*}
  Now we start to study in detail terms of
  \textcolor{YellowOrange}{\dashuline{\textcolor{black}{Pauli-Fierz}}}
  type.

  \paragraph*{\textcolor{myblue}{1)} $\bar{\xi}_1 aa\psi$,
    $\xi_{1} a^{*}a^{*}\psi^{*}$.}

  \begin{equation*}
    \textcolor{YellowOrange}{\dashuline{\textcolor{Black}{\bar{\xi}_1 aa\psi}}}=-\tfrac{1}{2\sqrt{2}M}\int_{\mathds{R}^3}^{}\Bigl(a\bigl(r_{\infty}e^{-ik\cdot x}\bigr)\Bigr)^2\bar{\xi} _1(x)\psi (x) dx\; .
  \end{equation*}
  We recall that $r_{\infty }\sim kg_{\infty }$, where $g_{\sigma }$
  is defined by \eqref{eq:3} for any $\sigma \leq \infty $.
  Let $\bar{\xi}_1 \alpha \alpha u$ be the symbol\footnote{We
    recall that for the Nelson model
    $\mathcal{Z}=L^2 (\mathds{R}^d )\oplus L^2 (\mathds{R}^d )$, thus
    we denote the variable $z$ by $u\oplus \alpha $.}  associated to
  $\bar{\xi}_1 aa\psi$, i.e.
  $\bar{\xi}_1 aa\psi=(\bar{\xi}_1 \alpha \alpha u)^{Wick}$.  Now,
  since $r_{\infty }\notin L^2 (\mathds{R}^3 )$, we cannot expect that
  $\bar{\xi}_1 \alpha \alpha u$ is defined for any
  $u,\alpha \in L^2 (\mathds{R}^3 )$, and therefore that it is a
  compact symbol. We introduce the approximated symbol
  $\bar{\xi}_1 \alpha \alpha u^{(m)}$ defined by
  \begin{equation*}
    \bar{\xi}_1 aa\psi^{(m)}=(\bar{\xi}_1 \alpha \alpha u^{(m)})^{Wick}=-\tfrac{1}{2\sqrt{2}M}\int_{\mathds{R}^3}^{}\Bigl(a\bigl(r_{\sigma _m }e^{-ik\cdot x}\bigr)\Bigr)^2\bar{\xi} _1(x)\psi (x)  dx\; ,
  \end{equation*}
  with $(\sigma _m)_{m\in \mathds{N}}\subset \mathds{R}$ such that
  $\lim_{m\to \infty }\sigma _m=\infty $. First of all, we prove that
  \eqref{eq:99} holds for
  $\bar{\xi}_1 aa\psi-\bar{\xi}_1 aa\psi^{(m)}$:
  \begin{equation*}
    \begin{split}
      \Bigl\lvert  \Bigl\langle (H_0+1)^{-1/2}\Phi   ,  (\bar{\xi}_1 aa\psi-\bar{\xi}_1 aa\psi^{(m)}) (H_0+1)^{-1/2}\Psi \Bigr\rangle_{} \Bigr\rvert_{}^{}\leq \tfrac{1}{2\sqrt{2}M}\lVert \xi_1   \rVert_2^{}\bigl\lVert  (d\Gamma (\omega )+1)^{1/2}\\(H_0+1)^{-1/2}\Phi \bigr\rVert_{}^{}\\\cdot \sup_{x\in \mathds{R}^3}\bigl\lVert (d\Gamma (\omega )+1)^{-1/2}\Bigl(a\bigl((r_{\infty }-r_{\sigma _m })e^{-ik\cdot x}\bigr)\Bigr)^2 (d\Gamma (\omega )+1)^{-1/2}\\ (d\Gamma (\omega )+1)^{1/2}(H_0+1)^{-1/2} (N_1+\varepsilon )^{1/2}\Psi \bigr\rVert_{}^{}\; .
    \end{split}
  \end{equation*}
  We use \eqref{eq:23} of Lemma~\ref{lemma:4} and the fact
  that $(d\Gamma (\omega )+1)^{1/2}(H_0+1)^{-1/2}$ is bounded with
  norm smaller than one to obtain
  \begin{equation*}
    \begin{split}
      \Bigl\lvert  \Bigl\langle (H_0+1)^{-1/2}\Phi   ,  (\bar{\xi}_1 aa\psi-\bar{\xi}_1 aa\psi^{(m)}) (H_0+1)^{-1/2}\Psi \Bigr\rangle_{} \Bigr\rvert_{}^{}\leq \tfrac{c}{2\sqrt{2}M}\lVert \xi_1   \rVert_2^{}\lVert \omega ^{-1/4}(r_{\infty }-r_{\sigma_m })  \rVert_2^2\bigl\lVert \Phi   \bigr\rVert_{}^{}\\\cdot \,\bigl\lVert (N_1+\bar{\varepsilon} )^{1/2}\Psi   \bigr\rVert_{}^{}\\
      \leq C^{(m)}(\xi_1 )\bigl\lVert (N_1+\bar{\varepsilon} )^{1/2}\Phi   \bigr\rVert_{}^{}\cdot \bigl\lVert (N_1+\bar{\varepsilon} )^{1/2}\Psi  \bigr\rVert_{}^{}\; ,
    \end{split}
  \end{equation*}
  with
  $C^{(m)}(\xi_1 )=C(\bar{\varepsilon},\xi_1 ) \lVert \omega
  ^{-1/4}(r_{\infty }-r_{\sigma_m }) \rVert_2^2$
  for some $C(\bar{\varepsilon},\xi_1 )>0$. The sequence
  $\bigl(C^{(m)}(\xi_1 )\bigr)_{m\in \mathds{N}}$ converges to zero
  since by our choice of $(\sigma _m)_{m\in \mathds{N}}$:
  \begin{equation*}
    \lim_{m\to \infty }\lVert \omega^{-1/4}(r_{\infty }-r_{\sigma_m }) \rVert_2^2=0\; .
  \end{equation*}
  It remains to show that $\bar{\xi}_1 \alpha \alpha u^{(m)}$ is a
  compact symbol. Such symbol can be written as
  \begin{equation*}
    \bar{\xi}_1 \alpha \alpha u^{(m)}=-\tfrac{1}{2\sqrt{2}M}\int_{\mathds{R}^9}^{}\bar{\xi} _1(x)\bar{r}_{\sigma _m }(k)\bar{r}_{\sigma _m }(k')e^{i(k+k')\cdot x}  \alpha(k)\alpha(k') u (x)  dxdkdk'\; .
  \end{equation*}
  Now we can define an operator
  $\tilde{b}_{\alpha \alpha u}:\bigl(L^2 \oplus L^2\bigr)^{\otimes_s
    3}\to \mathds{C}$
  in the following way. Let the maps
  $\pi _1,\pi _2:L^2 (\mathds{R}^3 )\oplus L^2 (\mathds{R}^3 )\to L^2
  (\mathds{R}^3 )$
  be the projections on the first and second space respectively.  Then
  we define the operator $\tilde{b}_{\alpha \alpha u}$ as:
  \begin{equation*}
    \tilde{b}_{\alpha \alpha u}:(u,\alpha )^{\otimes 3}\in \bigl(L^2\oplus L^2\bigr)^{\otimes_s 3}\underset{\pi _2\otimes \pi _2\otimes \pi _1}{\longrightarrow} \alpha(k) \alpha(k') u(x)\in L^2 (\mathds{R}^{9}   )\longrightarrow \langle f  , \alpha \alpha u \rangle_{L^2 (\mathds{R}^9  )}\in \mathds{C}\; ;
  \end{equation*}
  where
  $f(k,k',x)=-\tfrac{1}{2\sqrt{2}M}\bar{\xi} _1(x)\bar{r}_{\sigma _m
  }(k)\bar{r}_{\sigma _m }(k')e^{i(k+k')\cdot x}\in L^2 (\mathds{R}^9
  )$.
  Therefore $\tilde{b}_{\alpha \alpha u}$ is bounded and finite rank,
  and therefore compact. The proof for the corresponding adjoint term
  \begin{equation*}
    \textcolor{YellowOrange}{\dashuline{\textcolor{Black}{\xi_{1} a^{*}a^{*}\psi^{*}}}}=-\tfrac{1}{2\sqrt{2}M}\int_{\mathds{R}^3}^{}\Bigl(a^{*}\bigl(r_{\infty }e^{-ik\cdot x}\bigr)\Bigr)^2\xi _1(x)\psi^{*} (x)  dx
  \end{equation*}
  can be obtained directly from the above, using the following
  approximation with compact symbol:
  \begin{equation*}
    \xi_1 a^{*}a^{*}\psi^{*(m)}=(\xi_1 \bar{\alpha} \bar{\alpha} \bar{u}^{(m)})^{Wick}=-\tfrac{1}{2\sqrt{2}M}\int_{\mathds{R}^3}^{}\Bigl(a^{*}\bigl(r_{\sigma _m }e^{-ik\cdot x}\bigr)\Bigr)^2\xi _1(x)\psi^{*} (x)  dx\; .
  \end{equation*}

  \paragraph*{\textcolor{myblue}{2)} $\xi_{1} aa\psi^{*}$,
    $\bar{\xi}_{1} a^{*}a^{*}\psi$.}

  \begin{equation*}
    \textcolor{YellowOrange}{\dashuline{\textcolor{Black}{\xi_{1} aa\psi^{*}}}}=-\tfrac{1}{2\sqrt{2}M}\int_{\mathds{R}^3}^{}\Bigl(a\bigl(r_{\infty }e^{-ik\cdot x}\bigr)\Bigr)^2\xi _1(x)\psi^{*} (x)  dx\; .
  \end{equation*}
  Again we approximate this term by
  \begin{equation*}
    \xi_{1} aa\psi^{*(m)}=(\xi_{1} \alpha \alpha \bar{u}^{(m)})^{Wick}=-\tfrac{1}{2\sqrt{2}M}\int_{\mathds{R}^3}^{}\Bigl(a\bigl(r_{\sigma _m }e^{-ik\cdot x}\bigr)\Bigr)^2\xi _1(x)\psi^{*} (x)  dx
  \end{equation*}
  as above. The proof that it satisfies \eqref{eq:99} is
  perfectly analogous as the one for the previous term. Therefore we
  only prove that $\xi_{1} \alpha \alpha \bar{u}^{(m)}$ is a compact
  symbol. We define an operator
  $b_{\alpha \alpha \bar{u}}:\bigl(L^2 \oplus L^2\bigr)^{\otimes_s
    2}\to L^2 \oplus L^2$ by
  \begin{equation*}
    \begin{split}
      \tilde{b}_{\alpha \alpha \bar{u}}:(u,\alpha )^{\otimes 2}\in \bigl(L^2 \oplus L^2\bigr)^{\otimes_s 2}\underset{\pi _2\otimes \pi _2}{\longrightarrow} \alpha(k) \alpha(k') \in L^2 (\mathds{R}^{6}   )\\\underset{\tilde{c}_{\alpha \alpha \bar{u}}}{\longrightarrow} \biggl(\int_{\mathds{R}^6}^{}\bar{f}(k,k',\cdot )\alpha (k)\alpha(k')   dkdk'\,\oplus \,0\biggr)\in L^2 \oplus L^2\; ;
    \end{split}
  \end{equation*}
  where
  $f(k,k',x)=-\tfrac{1}{2\sqrt{2}M}\xi _1(x)r_{\sigma _m }(k)r_{\sigma
    _m }(k')e^{-i(k+k')\cdot x}$. By definition, we have that
  \begin{equation*}
    \xi_{1} \alpha \alpha \bar{u}^{(m)}=\langle (u, \alpha)  , \tilde{b}_{\alpha \alpha \bar{u}}(u,\alpha )^{\otimes 2} \rangle_{L^2 \oplus L^2}\; .
  \end{equation*}
  It is easily seen that the operator
  $\tilde{c}_{\alpha \alpha \bar{u}}$ is bounded. It is in fact
  compact: let $\beta _j\rightharpoonup \beta $ in
  $L^2 (\mathds{R}^6 )$ be a weakly convergent (bounded) sequence such
  that
  $\max\bigl\{\bigl(\sup_{j}\lVert \beta_{j} \rVert_{L^2 (\mathds{R}^6
    )}^{}\bigr),\lVert \beta \rVert_{L^2 (\mathds{R}^6
    )}^{}\bigr\}=X<\infty $; then
  \begin{equation*}
    \begin{split}
      \bigl\lVert \tilde{c}_{\alpha \alpha \bar{u}}(\beta -\beta _{j})  \bigr\rVert_{L^2 \oplus L^2 }^{}=\bigl\lVert \langle f(k,k',x)  , (\beta -\beta _{j})(k,k') \rangle_{L^2_{k,k'}(\mathds{R}^6)}  \bigr\rVert_{L^2_{x} (\mathds{R}^3  )}^{}\underset{j\to \infty }{\longrightarrow}0\; ,
    \end{split}
  \end{equation*}
  by Lebesgue's dominated convergence theorem, using the uniform bound
  \begin{equation*}
    \begin{split}
      \Bigl\lvert \langle f(k,k',x)  , (\beta -\beta _{j})(k,k') \rangle_{L^2_{k,k'}(\mathds{R}^6)}^2  \Bigr\rvert_{}^{}\leq \lVert f(k,k',x)  \rVert_{L^2_{k,k'}(\mathds{R}^6)}^2\Bigl(\lVert \beta\rVert_{L^2(\mathds{R}^6)}^2 +\lVert \beta_j  \rVert_{L^2(\mathds{R}^6)}^2\Bigr)\\\leq \tfrac{2X}{8M^2}\lVert r_{\sigma _m}  \rVert_2^4\lvert \xi _1(x)  \rvert_{}^2\in L^1_x (\mathds{R}^3  )\; .
    \end{split}
  \end{equation*}
  Therefore, since $\tilde{c}_{\alpha \alpha \bar{u}}$ is compact and
  $\pi _2\otimes \pi _2$ is bounded, it follows that
  $\tilde{b}_{\alpha \alpha \bar{u}}$ is compact. Again, that implies
  the result holds also for the adjoint term
  \begin{equation*}
    \textcolor{YellowOrange}{\dashuline{\textcolor{Black}{\bar{\xi}_{1} a^{*}a^{*}\psi}}}=-\tfrac{1}{2\sqrt{2}M}\int_{\mathds{R}^3}^{}\Bigl(a^{*}\bigl(r_{\infty }e^{-ik\cdot x}\bigr)\Bigr)^2\bar{\xi} _1(x)\psi (x)  dx\; .
  \end{equation*}

  \paragraph*{\textcolor{myblue}{3)} $\bar{\xi}_{1} a^{*}a\psi$,
    $\xi_{1} a^{*}a\psi^{*}$.}

  \begin{gather*}
    \textcolor{YellowOrange}{\dashuline{\textcolor{Black}{\bar{\xi}_{1} a^{*}a\psi}}}=-\tfrac{1}{\sqrt{2}M}\int_{\mathds{R}^3}^{}a^{*}\bigl(r_{\infty }e^{-ik\cdot x}\bigr) a\bigl(r_{\infty }e^{-ik\cdot x}\bigr)\bar{\xi} _1(x)\psi (x)  dx\; ,\\
    \textcolor{YellowOrange}{\dashuline{\textcolor{Black}{\xi_{1} a^{*}a\psi^{*}}}}=-\tfrac{1}{\sqrt{2}M}\int_{\mathds{R}^3}^{}a^{*}\bigl(r_{\infty }e^{-ik\cdot x}\bigr) a\bigl(r_{\infty }e^{-ik\cdot x}\bigr)\xi _1(x)\psi^{*} (x)  dx\; .
  \end{gather*}
  The proof for this couple of terms goes on exactly like the previous
  one, i.e. approximating $r_{\infty }$ with $r_{\sigma _m}$ and
  showing that the corresponding operator
  $\tilde{c}_{\bar{\alpha}\alpha u }$ is compact, for it maps weakly
  convergent  sequences into strongly convergent ones.

  \paragraph*{\textcolor{myblue}{4)} $\bar{\xi}_2a \psi ^{*}\psi$,
    $\xi_2a^{*} \psi ^{*}\psi$.}

  \begin{equation*}
    \textcolor{YellowOrange}{\dashuline{\textcolor{Black}{\bar{\xi}_2a \psi ^{*}\psi }}}=-\tfrac{\sqrt{2} i}{M}\int_{\mathds{R}^6}^{}\Im\bigl(\xi_2(k')\bar{r}_{\infty }(k')e^{ik'\cdot x} \bigr) a\bigl(r_{\infty }e^{-ik\cdot x}\bigr) \psi ^{*}(x)\psi (x)  dxdk'\; .
  \end{equation*}
  We approximate it by the symbol $\bar{\xi}_2\alpha \bar{u}u^{(m)} $
  defined by:
  \begin{equation*}
    \begin{split}
          \bar{\xi }_2a\psi ^{*}\psi ^{(m)}=(\bar{\xi}_2\alpha \bar{u}u^{(m)} )^{Wick}=-\tfrac{\sqrt{2} i}{M}\int_{\mathds{R}^6}^{}\psi ^{*}(x)\chi _m(D_x) \Im\bigl(\xi_2(k')\bar{r}_{\infty }(k')e^{ik'\cdot x} \bigr) a\bigl(r_{\sigma _m }e^{-ik\cdot x}\bigr)\\ \psi (x)  dxdk'\; ;
    \end{split}
  \end{equation*}
  where $\chi _m$ is the smooth cut-off function defined at the
  beginning of Section~\ref{sec:nelson-hamiltonian}, while
  $r_{\sigma _m}$ is the usual regularization of $r_{\infty }$ defined
  above. First of all we check that the approximation satisfies
  \eqref{eq:99}. By the chain rule, two parts have to be
  checked:
  \begin{eqnarray*}
      &&\Bigl\lvert \Bigl\langle (H_0+1)^{-1/2}\Phi   , (\bar{\xi }_2a\psi ^{*}\psi-\bar{\xi }_2a\psi ^{*}\psi ^{(m)})(H_0+1)^{-1/2}\Psi  \Bigr\rangle_{}  \Bigr\rvert_{}^{}\leq \tfrac{\sqrt{2}(2\pi )^{3/2}}{M}  \biggl(\\&&\hspace{.2in}\Bigl\lvert \Bigl\langle (H_0+1)^{-1/2}\Phi   , \int_{\mathds{R}^3}^{}dx\:\psi ^{*}(x)\bigl(1-\chi _m(D_x)\bigr) \Im\mathcal{F}^{-1}\bigl(\xi_2\bar{r}_{\infty }\bigr)(x) a\bigl(r_{\infty }e^{-ik\cdot x}\bigr) \psi (x)  (H_0+1)^{-1/2}\Psi  \Bigr\rangle_{}  \Bigr\rvert_{}^{}\\&&+\Bigl\lvert \Bigl\langle (H_0+1)^{-1/2}\Phi   , \int_{\mathds{R}^3}^{}dx\:\psi^{*}(x) \chi_m(D_x) \Im\mathcal{F}^{-1}\bigl(\xi_2\bar{r}_{\infty } \bigr)(x) a\bigl((r_{\infty }-r_{\sigma _m })e^{-ik\cdot x}\bigr) \psi (x)  (H_0+1)^{-1/2}\Psi  \Bigr\rangle_{}  \Bigr\rvert_{}^{}\biggr)\; ;
  \end{eqnarray*}
  and we will consider them separately. For the first part we have:
  \begin{eqnarray*}
     && \Bigl\lvert \Bigl\langle (H_0+1)^{-1/2}\Phi   , \int_{\mathds{R}^3}^{}dx\:\psi ^{*}(x)\bigl(1-\chi _m(D_x)\bigr) \Im\mathcal{F}^{-1}\bigl(\xi_2\bar{r}_{\infty } \bigr)(x) a\bigl(r_{\infty }e^{-ik\cdot x}\bigr) \psi (x) (H_0+1)^{-1/2}\Psi  \Bigr\rangle_{}  \Bigr\rvert_{}^{}\\
      &&\leq \sum_{n=0}^{\infty }n\varepsilon \Bigl\lvert \Bigl\langle (H^{(n)}_0+1)^{-1/2}\Phi_n   , (1-\chi _m(D_{x_1})) \Im\mathcal{F}^{-1}\bigl(\xi_2\bar{r}_{\infty } \bigr)(x_1) a\bigl(r_{\infty }e^{-ik\cdot x_{1}}\bigr) (H^{(n)}_0+1)^{-1/2}\Psi_n  \Bigr\rangle_{\mathcal{H}_n}  \Bigr\rvert_{}^{}\\
      &&\leq \sum_{n=0}^{\infty }n\varepsilon \bigl\lVert (1-D_x^2 )^{-1/2}\bigl(1-\chi _m(D_x)\bigr)  \bigr\rVert_{\mathcal{L}(L^2 (\mathds{R}^3  ))}^{}\cdot \bigl\lVert \mathcal{F}^{-1}(\xi _2\bar{r}_{\infty })  \bigr\rVert_{\infty }^{}\cdot \bigl\lVert \omega ^{-1/2}r_{\infty }  \bigr\rVert_2^{}\\&&\hspace{.5in}\cdot \bigl\lVert (1-D_{x_1}^2 )^{1/2}(H_0^{(n)}+1)^{-1/2}\Phi _n  \bigr\rVert_{\mathcal{H}_n}^{}\cdot \bigl\lVert d\Gamma (\omega )^{1/2}(H_0^{(n)}+1)^{-1/2}\Psi  _n  \bigr\rVert_{\mathcal{H}_n}^{}\\
      &&\leq (1+\bar{\varepsilon} ) \lVert \xi _2  \rVert_{\mathcal{F}H^{1/2}}^{}\cdot \bigl\lVert \omega ^{-1/2}r_{\infty }  \bigr\rVert_2^{2}\cdot \bigl\lVert (1-D_x^2 )^{-1/2}\bigl(1-\chi _m(D_x)\bigr)  \bigr\rVert_{\mathcal{L}(L^2 (\mathds{R}^3  ))}^{}\cdot \bigl\lVert (N_1+\bar{\varepsilon} )^{1/2}\Phi   \bigr\rVert_{}^{}\cdot \bigl\lVert N_1^{1/2}\Psi  \bigr\rVert_{}^{}\; ;
  \end{eqnarray*}
  where in the last inequality we have utilized the following bound:
  \begin{equation*}
    \begin{split}
      n\varepsilon \bigl\lVert (1-D_{x_1}^2 )^{1/2}(H_0^{(n)}+1)^{-1/2}\Phi _n  \bigr\rVert_{\mathcal{H}_n}^2=\bigl\langle \Phi _n  , (H_0^{(n)}+1)^{-1/2}d\Gamma (1-\Delta ) (H_0^{(n)}+1)^{-1/2}\Phi _n \bigr\rangle_{\mathcal{H}_n}\\
      \leq \bigl\lVert N_1^{1/2}\Phi _n  \bigr\rVert_{\mathcal{H}_n}^{}+\bigl\lVert d\Gamma (-\Delta )^{1/2}(H_0^{(n)}+1)^{-1/2}\Phi _n  \bigr\rVert_{\mathcal{H}_n}^{}\leq (1+\bar{\varepsilon} )\bigl\lVert (N_1+\bar{\varepsilon} )^{1/2}\Phi_n   \bigr\rVert_{\mathcal{H}_n}^{}\; .
    \end{split}
  \end{equation*}
  So the first part satisfies \eqref{eq:99}, since
  \begin{equation*}
    \lim_{m\to \infty }\bigl\lVert (1-D_x^2 )^{-1/2}\bigl(1-\chi _m(D_x)\bigr)  \bigr\rVert_{\mathcal{L}(L^2 (\mathds{R}^3  ))}^{}=0\; .
  \end{equation*}
  A similar procedure for the second part yields
  \begin{equation*}
    \begin{split}
      \Bigl\lvert \Bigl\langle (H_0+1)^{-1/2}\Phi   , \int_{\mathds{R}^3}^{}dx\:\psi ^{*}(x)\chi_m(D_x) \Im\mathcal{F}^{-1}\bigl(\xi_2\bar{r}_{\infty } \bigr)(x) a\bigl((r_{\infty }-r_{\sigma _m })e^{-ik\cdot x}\bigr) \psi (x)  (H_0+1)^{-1/2}\Psi  \Bigr\rangle_{}  \Bigr\rvert\\
      \leq \lVert \xi _2  \rVert_{\mathcal{F}H^{1/2}}^{}\cdot \bigl\lVert \omega ^{-1/2}r_{\infty }  \bigr\rVert_2^{}\cdot \bigl\lVert \omega ^{-1/2}(r_{\infty }-r_{\sigma _m})  \bigr\rVert_2^{}\bigl\lVert N_1^{1/2}\Phi   \bigr\rVert_{}^{}\cdot \bigl\lVert N_1^{1/2}\Psi  \bigr\rVert_{}^{}\; ;
    \end{split}
  \end{equation*}
  i.e. it satisfies \eqref{eq:99}, for
  $\lim_{m\to \infty }\lVert \omega ^{-1/2}(r_{\infty }-r_{\sigma _m})
  \rVert_2=0$.
  Now it remains to show that $\bar{\xi}_2\alpha \bar{u}u^{(m)} $ is a
  compact symbol:
  \begin{equation*}
    \bar{\xi}_2\alpha \bar{u}u^{(m)}=-\tfrac{(2\pi )^{3/2}\sqrt{2} i}{M}\int_{\mathds{R}^6}^{}\bar{u}(x)\chi _m(D_x) \Im\mathcal{F}^{-1}\bigl(\xi_2\bar{r}_{\infty }\bigr)(x) \bar{r}_{\sigma _m }(k)e^{ik\cdot x} \alpha(k) u (x)  dxdk\; .
  \end{equation*}
  As for the previous terms, we define an operator
  $b_{\alpha \bar{u}u}:\bigl(L^2 \oplus L^2 \bigr)^{\otimes_s 2}\to
  L^2 \oplus L^2 $ by
  \begin{equation*}
    \begin{split}
      \tilde{b}_{\alpha \bar{u}u}:(u,\alpha )^{\otimes 2}\in \bigl(L^2 \oplus L^2\bigr)^{\otimes_s 2}\underset{\pi _2\otimes \pi _1}{\longrightarrow} \alpha(k) u(x) \in L^2 (\mathds{R}^{6}   )\underset{\tilde{c}_{\alpha \bar{u}u}}{\longrightarrow} \Bigl(f'(x,D_{x})u(x)  \oplus \,0\Bigr)\in L^2 \oplus L^2\; ;
    \end{split}
  \end{equation*}
  where
  $f'(x,D_{x})=-\tfrac{(2\pi )^3\sqrt{2}i}{M} \chi_m(D_x)
  \mathcal{F}^{-1}\bigl(\bar{r}_{\sigma _m }\alpha
  \bigr)(x)\Im\mathcal{F}^{-1}\bigl(\xi_2\bar{r}_{\infty }\bigr)(x)$.
  We can easily prove that
  $f':L^2 (\mathds{R}^3 )\to L^2 (\mathds{R}^3 )$ is a compact
  operator. The cutoff function
  $\chi _m\in L^{\infty }_0 (\mathds{R}^3 )$ by hypothesis\footnote{We
    denote by $L_0^{\infty }(\mathds{R}^{3})$ the set of bounded
    functions on $\mathds{R}^{3}$  that vanish at infinity.}. Now both $\bar{r}_{\sigma _m }\alpha$
  and $\xi_2\bar{r}_{\infty }$ belong to $L^1 (\mathds{R}^3 )$, since
  $r_{\sigma_m },\alpha, \omega^{1/2}\xi_2,\omega ^{-1/2}r_{\infty }
  \in L^2 (\mathds{R}^3 )$.
  Therefore
  $\mathcal{F}^{-1}\bigl(\bar{r}_{\sigma _m }\alpha
  \bigr)\Im\mathcal{F}^{-1}\bigl(\xi_2\bar{r}_{\infty }\bigr)\in
  L^{\infty }_0 (\mathds{R}^3 )$,
  hence $f'(x,D_{x})\in \mathcal{K}(L^2 (\mathds{R}^3 ))$. It
  immediately follows that $\tilde{b}_{\alpha \bar{u}u}$ is compact,
  and the proof is complete. As usual, this result implies the one for
  the adjoint term
  \begin{equation*}
    \textcolor{YellowOrange}{\dashuline{\textcolor{Black}{\xi_2a^{*} \psi ^{*}\psi }}}=-\tfrac{\sqrt{2} i}{M}\int_{\mathds{R}^6}^{}\Im\bigl(\xi_2(k')\bar{r}_{\infty }(k')e^{ik'\cdot x} \bigr) a^{*}\bigl(r_{\infty }e^{-ik\cdot x}\bigr) \psi^{*}(x)\psi (x)  dxdk'\; .
  \end{equation*}

  \paragraph*{\textcolor{myblue}{5)} $D_xa\bar{\xi} _{1}\psi $,
    $a^{*}D_x\psi ^{*}\xi _1$, $D_xa\psi^{*}\xi _{1} $,
    $a^{*}D_x\bar{\xi} _1\psi$.}

  \begin{equation*}
    \textcolor{YellowOrange}{\dashuline{\textcolor{Black}{D_xa\bar{\xi} _{1}\psi}}}=\tfrac{1}{\sqrt{2}M}\int_{\mathds{R}^3}^{}\bar{\xi}_1(x)D_xa\bigl(r_{\infty }e^{-ik\cdot x}\bigr)\psi (x)   dx\; .
  \end{equation*}
  The approximated symbol $D_xa\bar{\xi} _{1}\psi^{(m)}$ is given by
  \begin{equation*}
    D_xa\bar{\xi} _{1}\psi^{(m)}=\tfrac{1}{\sqrt{2}M}\int_{\mathds{R}^3}^{}\bar{\xi}_1(x)D_xa\bigl(r_{\sigma _m }e^{-ik\cdot x}\bigr)\psi (x)   dx\; .
  \end{equation*}
  First of all we prove that \eqref{eq:99} is
  satisfied. Given $\Phi \in \mathcal{H}$, we denote by $\Phi _{n,p}$
  its restriction to the subspace
  $\mathcal{H}_{n,p}=\Bigl(L^2 (\mathds{R}^3 )\Bigr)^{\otimes _s
    n}\otimes \Bigl(L^2 (\mathds{R}^3 )\Bigr)^{\otimes _s p }$
  with $n$ nucleons and $p $ mesons. We also denote by
  $X_n=\{x_1,\dotsc, x_n\}$ a set of variables,
  $dX_n=dx_{1}\dotsm dx_n$ the corresponding Lebesgue measure (and
  analogously for $K_p$, $dK_p$). The proof is obtained by a direct
  calculation on the Fock space as follows:
  \begin{equation*}
    \begin{split}
      \biggl\lvert \Bigl\langle (H_0+1)^{-1/2}\Phi   , \int_{\mathds{R}^3}^{}\bar{\xi}_1(x)D_xa\bigl((r_{\infty }-r_{\sigma _m})e^{-ik\cdot x}\bigr)\psi (x)   d x (H_0+1)^{-1/2}\Psi  \Bigr\rangle_{} \biggr\rvert_{}^{}\\
      =\biggl\lvert \sum_{n,p=0}^{\infty }\varepsilon \sqrt{(n+1)(p+1)}\int_{\mathds{R}^{(n+p+2)d}}^{}\overline{\Bigl((H_0+1)^{-1/2}\Phi \Bigr)}_{n,p}(X_n;K_p)\bar{\xi}_{1}(x)D_x(\bar{r}_{\infty }-\bar{r}_{\sigma _m})(k)e^{ik\cdot x}\\\overline{\Bigl((H_0+1)^{-1/2}\Psi  \Bigr)}_{n+1,p+1}(x,X_n;k,K_p)   dxdX_ndkdK_p  \biggr\rvert_{}^{}\\
      \leq \sum_{n,p=0}^{\infty }\sqrt{\varepsilon(n+1)}\biggl\lvert\int_{\mathds{R}^{(n+p+2)d}}^{}\overline{\Bigl((H_0+1)^{-1/2}\Phi \Bigr)}_{n,p}(X_n;K_p)\overline{D_x\xi_1 }(x)\overline{\tfrac{r_{\infty }-r_{\sigma _m}}{\sqrt{\omega }}}(k)e^{ik\cdot x}\sqrt{\varepsilon(p+1)\omega(k) }\\\overline{\Bigl((H_0+1)^{-1/2}\Psi  \Bigr)}_{n+1,p+1}(x,X_n;k,K_p)   dxdX_ndkdK_p  \biggr\rvert_{}^{}\\
      \leq \sum_{n,p=0}^{\infty }\sqrt{\varepsilon(n+1)}\bigl\lVert (-\Delta +V)^{1/2}\xi _1  \bigr\rVert_2^{}\cdot \bigl\lVert \omega ^{-1/2}(r_{\infty }-r_{\sigma _m})  \bigr\rVert_2^{}\cdot \bigl\lVert (H_0+1)^{-1/2}\Phi_{n,p}   \bigr\rVert_{\mathcal{H}_{n,p}}^{}\\\cdot \,\bigl\lVert e^{ik\cdot x}\sqrt{\varepsilon(p+1)\omega(k_1) } (H_0+1)^{-1/2}\Psi_{n+1,p+1}(X_{n+1};K_{p+1})   \bigr\rVert_{\mathcal{H}_{n+1,p+1}}^{}\\
      \leq \bigl\lVert (-\Delta +V)^{1/2}\xi _1  \bigr\rVert_2^{}\cdot \bigl\lVert \omega ^{-1/2}(r_{\infty }-r_{\sigma _m})  \bigr\rVert_2^{}\cdot \bigl\lVert (N_1+\bar{\varepsilon} )^{1/2}\Phi   \bigr\rVert_{}^{}\cdot \bigl\lVert \Psi   \bigr\rVert_{}^{}\; ;
    \end{split}
  \end{equation*}
  where in the last bound we have used Schwartz's inequality and the
  fact that $p \omega (k_1)\equiv\sum_{j=1}^p\omega (k_j)$ when acting on
  vectors of $\mathcal{H}_{n,p}$. Now, since
  $\lim_{m\to \infty }\bigl\lVert \omega ^{-1/2}(r_{\infty }-r_{\sigma
    _m}) \bigr\rVert_2=0$,
  Equation~\eqref{eq:99} holds with
  $C^{(m)}(\xi _1)=\tfrac{1}{\sqrt{2}M}\bigl\lVert (-\Delta
  +V)^{1/2}\xi _1 \bigr\rVert_2^{}\cdot \bigl\lVert \omega
  ^{-1/2}(r_{\infty }-r_{\sigma _m}) \bigr\rVert_2$.
  It remains to show that the classical symbol
  \begin{equation*}
    D_x\alpha \bar{\xi} _{1}u^{(m)}=\tfrac{1}{\sqrt{2}M}\int_{\mathds{R}^6}^{}\bar{\xi}_1(x)D_x\alpha(k) \bar{r}_{\sigma _m }(k)e^{ik\cdot x}u (x)   dxdk
  \end{equation*}
  is compact. Here we have written
  $D_x\alpha \bar{\xi} _{1}u^{(m)}=\langle \xi _1 , D_xv \rangle_2$,
  with
  $v_x(x)=\tfrac{(2\pi )^{3/2}}{\sqrt{2}M}\mathcal{F}^{-1}\bigl(\alpha
  \bar{r}_{\sigma _m}\bigr)(x)u (x)$;
  and that is defined for any $v\in \dot{H}^1(\mathds{R}^3)$. However,
  since $\xi _1\in Q(-\Delta +V)\subset H^1(\mathds{R}^3)$ and $D_x$
  is self-adjoint, we can write
  $D_x\alpha \bar{\xi} _{1}u^{(m)}=\langle D_x\xi _1 , v \rangle_2$
  for any $v\in L^2 (\mathds{R}^3 )$. It follows that
  $D_x\alpha \bar{\xi} _{1}u^{(m)}$ is defined for any
  $u,\alpha \in L^2 (\mathds{R}^3)$, since
  $\alpha ,r_{\sigma _m}\in L^2$ implies
  $\alpha \bar{r}_{\sigma _m}\in L^1$, and therefore
  $\mathcal{F}^{-1}\bigl(\alpha \bar{r}_{\sigma _m}\bigr)\in L^{\infty
  }$.
  It follows that the operator
  $\tilde{b}_{D_{x}\alpha u}:\bigl(L^2 \oplus L^2 \bigr)^{\otimes_s
    2}\to \mathds{C}$ defined as
  \begin{equation*}
    \tilde{b}_{D_{x}\alpha u}:(u,\alpha )^{\otimes 2}\in \bigl(L^2 \oplus L^2\bigr)^{\otimes_s 2}\underset{\pi _2\otimes \pi _1}{\longrightarrow} \alpha(k) u(x)\in L^2 (\mathds{R}^{6}   )\longrightarrow \langle f''  , \alpha u \rangle_{L^2 (\mathds{R}^6  )}\in \mathds{C}\; ,
  \end{equation*}
  with
  $f''(x,k)=\tfrac{1}{\sqrt{2}M}(D_{x}\xi_1 )(x)r_{\sigma
    _m}(k)e^{-ik\cdot x}$,
  is bounded and finite rank, and therefore compact.
  \begin{equation*}
    \textcolor{YellowOrange}{\dashuline{\textcolor{Black}{a^{*}D_x\bar{\xi}_{1}\psi }}}=\tfrac{1}{\sqrt{2}M}\int_{\mathds{R}^3}^{}\bar{\xi}_1(x) a^{*}\bigl(r_{\infty }e^{-ik\cdot x}\bigr)D_x\psi  (x)  dx\; .
  \end{equation*}
  Again, the approximated symbol $a^{*}D_x\bar{\xi}_{1}\psi$ is given
  by
  \begin{equation*}
    a^{*}D_x\bar{\xi}_{1}\psi^{(m)}=\tfrac{1}{\sqrt{2}M}\int_{\mathds{R}^3}^{}\bar{\xi}_1(x) a^{*}\bigl(r_{\sigma_m  }e^{-ik\cdot x}\bigr)D_x\psi  (x)  dx\; .
  \end{equation*}
  Equation~\eqref{eq:99} is satisfied, and the proof follows the same
  guidelines as the one for the previous term
  $D_xa\bar{\xi} _{1}\psi$. We give the compactness proof for the
  symbol
  \begin{equation*}
    \bar{\alpha} D_x\bar{\xi}_{1}u^{(m)}=\tfrac{1}{\sqrt{2}M}\int_{\mathds{R}^6}^{}\bar{\xi}_1(x) \bar{\alpha}(k) r_{\sigma_m  }(k)e^{-ik\cdot x}D_xu  (x)  dxdk\; .
  \end{equation*}
  We rewrite it as
  $\bar{\alpha} D_x\bar{\xi}_{1}u^{(m)}=\langle (u,\alpha ) ,
  \tilde{b}_{\bar{\alpha}D_x u }(u,\alpha ) \rangle_{L^2\oplus L^2}$,
  with
  $\tilde{b}_{\bar{\alpha}D_x u }:L^2 \oplus L^2 \to L^2 \oplus L^2 $
  defined as
  \begin{equation*}
    \begin{split}
      \tilde{b}_{\bar{\alpha} D_x u}:(u,\alpha )\in L^2 \oplus L^2\underset{\pi _1}{\longrightarrow} u(x) \in L^2 (\mathds{R}^{3}   )\underset{\tilde{c}_{\bar{\alpha} D_x u}}{\longrightarrow} \Bigl( 0 \, \oplus \, f'''(k) \Bigr)\in L^2 \oplus L^2\; ,
    \end{split}
  \end{equation*}
  where
  $f'''(k)=\tfrac{1}{\sqrt{2}M}r_{\sigma _m}(k)\Bigl(k\langle
  e^{ik\cdot x}\xi_1 , u \rangle_{L^2_x}+\langle e^{ik\cdot
    x}D_{x}\xi_1 , u \rangle_{L^2_x}\Bigr)$.
  Now suppose that $u_j\rightharpoonup u$ is a weakly convergent
  (bounded) sequence with bound $X$. It follows that, uniformly in $j$,
  \begin{equation*}
    \begin{split}
      \lvert f_j'''(k)  \rvert_{}^2=\Bigl\lvert \tfrac{1}{\sqrt{2}M}r_{\sigma _m}(k)\Bigl(k\langle
      e^{ik\cdot x}\xi_1 , u_j \rangle_{L^2_x}+\langle e^{ik\cdot
        x}D_{x}\xi_1 , u_j \rangle_{L^2_x}\Bigr)  \Bigr\rvert_{}^{2}\\\leq \tfrac{1}{2M^{2}}X^2\lvert r_{\sigma _m}(k)  \rvert_{}^2(k^2+1)\lVert \xi _1  \rVert_{H^1}^2\in L^1_k (\mathds{R}^3  )\; .
    \end{split}
  \end{equation*}
  In addition,
  $\lim_{j\to \infty }\lvert f'''(k)-f_j'''(k) \rvert_{}^2=0$;
  therefore $\tilde{c}_{\bar{\alpha} D_x u}$ is a compact operator by
  Lebesgue's dominated convergence theorem. So
  $\tilde{b}_{\bar{\alpha} D_x u}$ is compact. The proofs above extend
  immediately to the adjoint terms
  \begin{gather*}
    \textcolor{YellowOrange}{\dashuline{\textcolor{Black}{a^{*}D_x\psi ^{*}\xi _1}}}=\tfrac{1}{\sqrt{2}M}\int_{\mathds{R}^3}^{}\psi^{*}(x) a^{*}\bigl(r_{\infty }e^{-ik\cdot x}\bigr)D_x\xi_1(x)   dx\; ;\\
    \textcolor{YellowOrange}{\dashuline{\textcolor{Black}{D_xa\psi^{*}\xi _{1} }}}=\tfrac{1}{\sqrt{2}M}\int_{\mathds{R}^3}^{}\psi ^{*}(x)D_xa\bigl(r_{\infty }e^{-ik\cdot x}\bigr)\xi_1(x)   dx\; .
  \end{gather*}

  \paragraph*{\textcolor{myblue}{6)} $\psi ^{*}D_x\xi _2\psi$,
    $\psi ^{*}\bar{\xi}_2D_x\psi$.}

  \begin{equation*}
    \textcolor{YellowOrange}{\dashuline{\textcolor{Black}{\psi ^{*}D_x\xi _2\psi}}}=\tfrac{(2\pi )^{3/2}}{\sqrt{2}M}\int_{\mathds{R}^3}^{}\psi^{*}(x) D_x \mathcal{F}^{-1}\Bigl(\xi_2\bar{r}_{\infty }\Bigr)(x)\psi (x)   dx\; .
  \end{equation*}
  The approximated symbol, as for the terms of point 4, contains
  $\chi _m(D_x)$:
  \begin{equation*}
    \psi ^{*}D_x\xi _2\psi^{(m)}=\tfrac{(2\pi )^{3/2}}{\sqrt{2}M}\int_{\mathds{R}^3}^{}\psi^{*}(x) \chi _m(D_x) D_x \mathcal{F}^{-1}\Bigl(\xi_2\bar{r}_{\infty }\Bigr)(x)\psi (x)   dx\; .
  \end{equation*}
  As usual, we start proving that \eqref{eq:99} holds. We
  remark that this is the only term where we need
  $\xi _2\in D(\omega ^{3/4})$ instead of $D(\omega ^{1/2})$.
  \begin{equation*}
    \begin{split}
      \biggl\lvert \Bigl\langle (H_0+1)^{-1/2}\Phi   , \int_{\mathds{R}^3}^{}\psi ^{*}(x)\bigl(1-\chi _m(D_x)\bigr) D_{x}\mathcal{F}^{-1}\Bigl(\xi_2\bar{r}_{\infty }\Bigr)(x)\psi (x)   dx (H_0+1)^{-1/2}\Psi  \Bigr\rangle_{} \biggr\rvert_{}^{}\\
      \leq \sum_{n=0}^{\infty }n\varepsilon \biggl\lvert \Bigl\langle (H_0+1)^{-1/2}\Phi_{n}   , \bigl(1-\chi _m(D_{x_1})\bigr) D_{x_1}\mathcal{F}^{-1}\Bigl(\xi_2\bar{r}_{\infty }\Bigr)(x_1) (H_0+1)^{-1/2}\Psi_n  \Bigr\rangle_{} \biggr\rvert\\
      \leq \sum_{n=0}^{\infty }n\varepsilon \bigl\lVert (1-\Delta )^{-1/2}\bigl(1-\chi _m(D_{x})\bigr)  \bigr\rVert_{\mathcal{L}(L^2 (\mathds{R}^3  ))}^{}\cdot \Bigl(\bigl\lVert \mathcal{F}^{-1}\Bigl(\xi_2\bar{r}_{\infty }\Bigr)  \bigr\rVert_{\infty }^{}+\bigl\lVert \mathcal{F}^{-1}\Bigl(k\xi_2\bar{r}_{\infty }\Bigr)  \bigr\rVert_{\infty }^{}\Bigr)\\\cdot\, \bigl\lVert (1-\Delta_{x_1} )^{1/2}(H_0+1)^{-1/2}\Phi_{n}  \bigr\rVert_{\mathcal{H}_n}^{}\cdot \Bigl( \bigl\lVert D_{x_1}(H_0+1)^{-1/2}\Psi_{n}  \bigr\rVert_{\mathcal{H}_n}^{}+\bigl\lVert (H_0+1)^{-1/2}\Psi_{n}  \bigr\rVert_{\mathcal{H}_n}^{}\Bigr)\\
      \leq 2\bigl\lVert (1-D_x^2 )^{-1/2}\bigl(1-\chi _m(D_{x})\bigr)  \bigr\rVert_{\mathcal{L}(L^2 (\mathds{R}^3  ))}^{}\cdot \bigl\lVert \omega ^{3/4}\xi _2\bigr  \rVert_2^{}\cdot \bigl\lVert \omega ^{-1/4}r_{\infty }\bigr  \rVert_2^{}\cdot \bigl\lVert \Phi   \bigr\rVert_{}^{}\cdot \\\Bigl(\bigl\lVert (N_1+\bar{\varepsilon} )^{1/2}\Psi   \bigr\rVert_{}^{}+\bigl\lVert \Psi   \bigr\rVert_{}^{}\Bigr)\; ;
    \end{split}
  \end{equation*}
  hence the result follows with
  \begin{equation*}
    C^{(m)}(\xi _2)=\tfrac{2 \sqrt{2}(2\pi )^{3/2}}{M}\bigl\lVert
    (1-D_x^2 )^{-1/2}\bigl(1-\chi _m(D_{x})\bigr)
    \bigr\rVert_{\mathcal{L}(L^2 (\mathds{R}^3 ))}^{} \bigl\lVert \omega
    ^{3/4}\xi _2\bigr \rVert_2^{} \bigl\lVert \omega ^{-1/4}r_{\infty
    }\bigr \rVert_2^{}\; ,
  \end{equation*}
  since
  $\lim_{m\to \infty }\bigl\lVert (1-D_x^2 )^{-1/2}\bigl(1-\chi
  _m(D_x)\bigr) \bigr\rVert_{\mathcal{L}(L^2 (\mathds{R}^3
    ))}^{}=0$. It remains to show that the symbol
  \begin{equation*}
    \bar{u}D_x\xi _2u^{(m)}=\tfrac{(2\pi )^{3/2}}{\sqrt{2}M}\int_{\mathds{R}^3}^{}\bar{u}(x) \chi _m(D_x) D_x \mathcal{F}^{-1}\Bigl(\xi_2\bar{r}_{\infty }\Bigr)(x)u (x)   dx
  \end{equation*}
  is compact. We introduce the operator
  $\tilde{b}_{\bar{u}D_x u}:L^2 \oplus L^2 \to L^2 \oplus L^2 $ such
  that
  $\bar{u}D_x\xi _2u^{(m)}=\langle (u,\alpha ) , \tilde{b}_{\bar{u}D_x
    u}(u,\alpha ) \rangle_{L^2\oplus L^2}$:
  \begin{equation*}
    \begin{split}
      \tilde{b}_{\bar{u}D_x u}:(u,\alpha )\in L^2 \oplus L^2\underset{\pi _1}{\longrightarrow} u(x) \in L^2 (\mathds{R}^{3}   )\underset{\tilde{c}_{\bar{u}D_x u}}{\longrightarrow} \Bigl( f''''(x,D_x)u(x) \,\oplus \,0\Bigr)\in L^2 \oplus L^2\; ,
    \end{split}
  \end{equation*}
  where
  $f''''(x,D_x)=\tfrac{(2\pi )^{3/2}}{\sqrt{2}M} D_x\,\chi _m(D_x)
  \mathcal{F}^{-1}\Bigl(\xi_2\bar{r}_{\infty }\Bigr)(x)$.
  Now $f''''(x,D_x)$ is a compact operator: both $x\,\chi _m(x)$ and
  $\mathcal{F}^{-1}\Bigl(\xi_2\bar{r}_{\infty }\Bigr)(x)$ are in
  $L^{\infty }_0(\mathds{R}^3)$. Therefore $\tilde{b}_{\bar{u}D_x u}$
  is compact. The proof extends immediately to the adjoint term
  \begin{equation*}
    \textcolor{YellowOrange}{\dashuline{\textcolor{Black}{\psi ^{*}\bar{\xi} _2D_x\psi}}}=\tfrac{(2\pi )^{3/2}}{\sqrt{2}M}\int_{\mathds{R}^3}^{}\psi^{*}(x) \mathcal{F}\Bigl(\bar{\xi}_2r_{\infty }\Bigr)(x) D_x\psi (x)   dx\; .
  \end{equation*}
\end{proof}

\subsection{Defining the time-dependent family of Wigner measures.}
\label{sec:defin-wign-meas}

The last tool we need in order to take the limit $\varepsilon \to 0$ of the
integral formula~\eqref{eq:92} are Wigner measures. Throughout this
section, we will leave some statements unproven; the reader may refer
to \citep[][Section
6]{ammari:nier:2008
} for the proofs, and a detailed discussion of Wigner measures
properties. We recall the definition of a Wigner measure
associated with a family of states on
$\mathcal{H}=\Gamma _s\bigl(L^2 (\mathds{R}^3 )\oplus L^2 (\mathds{R}^3 ) \bigr)$.
\begin{definition}
  Let
  $(\varrho _{\varepsilon })_{\varepsilon \in (0,\bar{\varepsilon}
    )}\subset \mathcal{L}^1\bigl(\mathcal{H}\bigr)$
  be a family of normal  states;
  $\mu \in \mathfrak{P}\bigl(L^2 \oplus L^2\bigr)$ a Borel probability
  measure. We say that $\mu $ is a Wigner (or semiclassical) measure
  associated to
  $(\varrho _{\varepsilon })_{\varepsilon \in (0,\bar{\varepsilon}
    )}$,
  or in symbols
  $\mu \in \mathcal{M}\bigl(\varrho _{\varepsilon },\varepsilon \in
  (0,\bar{\varepsilon} )\bigr)$,
  if there exist a sequence
  $(\varepsilon _k)_{k\in \mathds{N}}\subset (0,\bar{\varepsilon} )$
  such that $\lim_{k\to \infty }\varepsilon _k= 0$ and
  \begin{equation}
    \label{eq:105}
    \lim_{k\to \infty }\Tr\Bigl[\varrho _{\varepsilon _k}W(\xi )\Bigr]=\int_{L^2\oplus L^2}^{}e^{i \sqrt{2}\Re\langle \xi   , z \rangle_{L^2\oplus L^2}}  d\mu (z) \; ,\; \forall \xi \in L^2\oplus L^2\; .
  \end{equation}
  We remark that the right-hand side is essentially the Fourier transform of the
  measure $\mu $, so considering  the sequence
  $(\varepsilon _k)_{k\in \mathds{N}}$ there is at most one
  probability measure that could satisfy \eqref{eq:105}. If
  \eqref{eq:105} is satisfied, we say that to the sequence
  $(\varrho _{\varepsilon _k})_{k\in \mathds{N}}$ corresponds a single Wigner (or semiclassical) measure   $\mu $, or simply $\varrho _{\varepsilon _k}\rightarrow \mu $.
\end{definition}
First of all, it is necessary to ensure that such a definition of
Wigner measures is meaningful, i.e. that under suitable conditions
the set of Wigner measures $\mathcal{M}$ associated to a family of
states is not empty. Since $m_0>0$, it turns out that Assumption
\eqref{eq:arho} is sufficient. Assumption~\eqref{eq:86} would be
sufficient as well, even if we will not use it for the moment.
\begin{lemma}\label{lemma:14}
  Let
  $(\varrho _{\varepsilon })_{\varepsilon \in (0,\bar{\varepsilon} )}$
  be a family of normal states on $\mathcal{H}$, that satisfies
  Assumptions~\eqref{eq:arho} and \eqref{eq:84}. Then for any
  $t\in \mathds{R}$:
\begin{itemize}
\item [(i)] $ \mathcal{M}\bigl(\varrho _{\varepsilon }(t), \varepsilon \in (0,\bar{\varepsilon} )\bigr)\neq\emptyset\; ;\; \mathcal{M}\bigl(\tilde{\varrho} _{\varepsilon }(t), \varepsilon \in (0,\bar{\varepsilon} )\bigr)\neq\emptyset\; .$
\item [(ii)]  Any $\mu\in\mathcal{M}\bigl(\varrho _{\varepsilon }(t), \varepsilon \in (0,\bar{\varepsilon} )\bigr)$ or in $ \mathcal{M}\bigl(\tilde{\varrho} _{\varepsilon }(t), \varepsilon \in (0,\bar{\varepsilon} )\bigr)$\footnote{In this
    section, we have used mostly the notation $D(\omega ^{1/2})$;
    however $D(\omega^{1/2} )=\mathcal{F}H^{1/2}$, where the latter is
    defined in Definition~\ref{def:5}.} satisfies:
  \begin{equation*}
     \mu\Bigl(B_u(0,\sqrt{\mathfrak{C}})\cap Q(-\Delta +V)\oplus D(\omega ^{1/2})\Bigr)=1
  \end{equation*}
 \item  [(iii)] Moreover \begin{equation*}
     \int_{z=(u,\alpha )\in L^2\oplus L^2} \lVert (-\Delta +V)^{1/2}u  \rVert_2^2 + \lVert \alpha   \rVert_{\mathcal{F}H^{1/2}}^{2}\,d\mu(z) <+\infty\,\; .
  \end{equation*}
\end{itemize}
  We recall that
  $B_u(0,\sqrt{\mathfrak{C}})=\Bigl\{(u,\alpha )\in L^2\oplus
  L^2,\lVert u \rVert_2^{}\leq \sqrt{\mathfrak{C}} \Bigr\}$.
\end{lemma}
\begin{proof}
By \eqref{eq:106} of Lemma~\ref{lemma:13}, we see that  $\varrho
_{\varepsilon }(t)$ and $\tilde\varrho _{\varepsilon }(t)$ satisfy
\eqref{eq:84} and  \eqref{eq:arho} at any time.  Now $(i)$ follows
by \citep[Theorem 6.2]{ammari:nier:2008} and $(ii)$ by  $(iii)$ and \citep[Lemma 2.14]{MR2802894}. The third point is essentially a consequence of \citep[Lemma
  3.12]{2011arXiv1111.5918A}. However the latter result requires more
  regularity on the states $\varrho
_{\varepsilon }$. So we indicate here how to adapt the argument to our case.
It is enough to assume $t=0$ and $\{\mu\}=\mathcal{M}\bigl({\varrho} _{\varepsilon }, \varepsilon \in (0,\bar{\varepsilon} )\bigr)$.
The operators $-\frac{\Delta}{2M}+V$ and $\omega$ are positive (self-adjoint). So one can found
 non-decreasing sequences of finite rank operators $A_k$ and $B_k$ that converge weakly to
 $-\frac{\Delta}{2M}+V$ and $\omega$ respectively. In particular
\begin{equation*}
b_k^{wick}=d\Gamma(A_k) \otimes 1+1\otimes d\Gamma(B_k)\leq d\Gamma(-\frac{\Delta}{2M}+V)\otimes 1+1\otimes d\Gamma(\omega)=H_0\,,
\end{equation*}
where $b_k(u,\alpha)=\langle u, A_k u\rangle + \langle \alpha, B_k \alpha\rangle\in \mathcal{P}^\infty_{1,1}(L^2\oplus L^2)$. Let $P_k$ and $Q_k$ be the orthogonal projections
on ${\rm Ran}(A_k)$ and ${\rm Ran}(B_k)$ respectively. Using the Fock space decomposition
$\Gamma_s(L^2\oplus L^2)\equiv\Gamma_s(P_k L^2\oplus Q_k L^2)\otimes \Gamma_s(P_k^{\perp}L^2\oplus Q_k^{\perp}L^2)$ where $P_k^{\perp}=1-P_k$ and $Q_k^{\perp}=1-Q_k$; one can writes
$b_k^{Wick}\equiv(b_k)^{Wick}_{|\Gamma_s(P_kL^2\oplus Q_kL^2)}\otimes 1_{
\Gamma_s(P_k^{\perp}L^2\oplus Q_k^{\perp}L^2)}$ and $\varrho _{\varepsilon}\equiv
\hat\varrho _{\varepsilon}$. Hence
\begin{eqnarray*}
\Tr\Bigl[\varrho _{\varepsilon} b_k^{Wick}\Bigr]= \Tr\Bigl[\hat\varrho _{\varepsilon} b^{Wick}_{|\Gamma_s(P_kL^2\oplus Q_kL^2)}\otimes 1_{
\Gamma_s(P_k^{\perp}L^2\oplus Q_k^{\perp}L^2)}\Bigr]=
\Tr_{\Gamma_s(P_kL^2\oplus Q_kL^2)}\Bigl[\varrho^k _{\varepsilon} b_k^{Wick}\Bigr]\,,
\end{eqnarray*}
where $\varrho^k _{\varepsilon_j}$ is a given reduced density matrix which is  trace-class in $\Gamma_s(P_kL^2\oplus Q_kL^2)$. So the problem is in some sense reduced to finite dimension. Now using Wick calculus (in finite dimension) $b_k^{Wick}$ can be written as an Anti-Wick operator by moving all the
$a^*$ to the right of $a$. So, one obtains  $b_k^{Wick}= b_k^{A-Wick}+\varepsilon T$ with $T (d\Gamma(P_k\oplus Q_k)+1)^{-1}$ is bounded uniformly with respect to $\varepsilon\in(0,\bar\varepsilon)$.
Hence
\begin{equation*}
\overline{\lim_{\varepsilon\to 0}}\Tr_{\Gamma_s(P_kL^2\oplus Q_kL^2)}\Bigl[\varrho^k _{\varepsilon} b_k^{A-Wick}\Bigr]=\overline{\lim_{\varepsilon\to 0}}\Tr_{\Gamma_s(P_kL^2\oplus Q_kL^2)}\Bigl[\varrho^k _{\varepsilon} b_k^{Wick}\Bigr]
\leq \overline{\lim_{\varepsilon\to 0}}\Tr\Bigl[\varrho_{\varepsilon} H_0\Bigr]\leq C\,.
\end{equation*}
For details on the Anti-Wick quantization we refer the reader to \cite{ammari:nier:2008}; in particular it is a positive quantization (see e.g. \cite[Proposition 3.6]{ammari:nier:2008}). Hence, we see that
\begin{equation*}
\Tr_{\Gamma_s(P_kL^2\oplus Q_kL^2)}\Bigl[\varrho^k _{\varepsilon} (b_{k,\chi})^{A-Wick}\Bigr]\leq \Tr_{\Gamma_s(P_kL^2\oplus Q_kL^2)}\Bigl[\varrho^k _{\varepsilon} b_k^{A-Wick}\Bigr]
 \end{equation*}
where $b_{k,\chi}(u,\alpha)=\chi(u) \langle u, A_k u\rangle + \chi(\alpha)\langle \alpha, B_k \alpha\rangle$  for any cutoff function $\chi\in C_0^\infty(\mathds{R}^3)$, $0\leq \chi\leq 1$.
Finally \cite[Theorem 6.2]{ammari:nier:2008} gives
\begin{equation*}
\int_{z=(u,\alpha )\in L^2\oplus L^2} b_{k,\chi}(u,\alpha)\,d\mu(z) =\lim_{\varepsilon\to 0}\Tr\Bigl[\varrho_{\varepsilon} (b_{k,\chi})^{A-Wick}\Bigr]=
\lim_{\varepsilon\to 0}\Tr_{\Gamma_s(P_kL^2\oplus Q_kL^2)}\Bigl[\varrho^k _{\varepsilon} b_k^{A-Wick}\Bigr]\leq C\,,
\end{equation*}
and the monotone convergence theorem proves $(iii)$.
\end{proof}
As we said above, our aim is to take the limit $\varepsilon_k \to 0$
on the integral equation~\eqref{eq:92}, for a suitable sequence
contained in $(0,\bar{\varepsilon} )$. We may suppose that the
sequence $(\varepsilon_k)_{k\in \mathds{N}}$ is chosen in such a way
that there exist
$\mu _0\in \mathcal{M}\bigl(\varrho _{\varepsilon },\varepsilon \in
(0,\bar{\varepsilon} )\bigr)$
such that \eqref{eq:105} holds, i.e.
$\mathcal{M}\bigl(\varrho _{\varepsilon_k },k \in
\mathds{N}\bigr)=\{\mu _0\}$.
However, nothing a priori ensures that the sequence, or one of its
subsequences
$(\varepsilon_{k_i})_{i\in \mathds{N}}\subset (\varepsilon_k)_{k\in
  \mathds{N}}$, is such that \emph{for any $t\in \mathds{R}$:}
\begin{equation*}
  \lim_{i\to \infty }\Tr\Bigl[\tilde{\varrho} _{\varepsilon _{k_i}}(t)W(\xi )\Bigr]=\int_{L^2\oplus L^2}^{}e^{i \sqrt{2}\Re\langle \xi   , z \rangle}  d\tilde{\mu}_{t} (z) \; ,\; \forall \xi \in L^2 (\mathds{R}^3 )\oplus L^2(\mathds{R}^3 )\; ;
\end{equation*}
where
$\tilde{\mu}_t:\mathds{R}\to \mathfrak{P} \bigl(L^2 \oplus L^2\bigr)$
is a map such that $\tilde{\mu}_0=\mu_0 $. The possibility of
extracting such a common subsequence is crucial, since the integral
equation involves all measures from zero to an arbitrary time $t$. To
prove it is possible, we exploit the uniform continuity
properties of $\Tr\Bigl[\tilde{\varrho} _{\varepsilon}(t)W(\xi )\Bigr]$ in both
$t$ and $\xi $, proved in the following lemma.
\begin{lemma}\label{lemma:15}
  Let
  $(\varrho _{\varepsilon })_{\varepsilon \in (0,\bar{\varepsilon} )}$
  be a family of quantum states on $\mathcal{H}$ that satisfies
  Assumptions \eqref{eq:84} and \eqref{eq:arho}. Then the family of functions  
  $(t,\xi)\mapsto\tilde{G}_{\varepsilon }(t,\xi ):=\Tr\Bigl[\tilde{\varrho}
  _{\varepsilon}(t)W(\xi )\Bigr]$
  is uniformly equicontinuous on bounded subsets of
  $\mathds{R}\times \bigl( Q(-\Delta +V)\oplus D(\omega ^{1/2})\bigr) $.
\end{lemma}
\begin{proof}
  Let   $(t,\xi),(s,\eta)\in \mathds{R}\times
  \bigl(Q(-\Delta +V)\oplus D(\omega ^{1/2})\bigr)$. Without loss of generality, we
  may suppose that $s\leq t$. We write
  \begin{equation*}
    \Bigl\lvert  \tilde{G}_{\varepsilon }(t,\xi )- \tilde{G}_{\varepsilon }(s,\eta )\Bigr\rvert_{}^{}\leq \Bigl\lvert  \tilde{G}_{\varepsilon }(t,\eta )- \tilde{G}_{\varepsilon }(s,\eta)\Bigr\rvert_{}^{}+\Bigl\lvert  \tilde{G}_{\varepsilon }(t,\xi )- \tilde{G}_{\varepsilon }(t,\eta)\Bigr\rvert_{}^{}\; ;
  \end{equation*}
  and define
  $X_1:=\Bigl\lvert \tilde{G}_{\varepsilon }(t,\eta )-
  \tilde{G}_{\varepsilon }(s,\eta)\Bigr\rvert_{}^{}$,
  $X_2:=\Bigl\lvert \tilde{G}_{\varepsilon }(t,\xi )-
  \tilde{G}_{\varepsilon }(t,\eta)\Bigr\rvert_{}^{}$.
  Consider $X_1$; we get by standard manipulations and
  Lemma~\ref{lemma:10}:
  \begin{equation*}
    \begin{split}
      X_1\leq \sum_{j=0}^3 \varepsilon ^j\sum_{i\in \mathds{N}}^{}\lambda _i\int_{s}^t \Bigl\lvert \bigl\langle e^{-i \frac{s}{\varepsilon }\hat{H}_{ren}}\Psi _i  , W\bigl(\tilde{(\eta)}_s \bigr)B_j\bigl(\tilde{(\eta)}_s \bigr) e^{-i \frac{s}{\varepsilon }\hat{H}_{ren}}\Psi _i\bigr\rangle_{}  \Bigr\rvert_{}^{} ds\; .
    \end{split}
  \end{equation*}
  Now using Lemma~\ref{lemma:11} we obtain
  \begin{equation*}
    \begin{split}
            X_1\leq \sum_{j=0}^3 \varepsilon ^jC_j(\eta )\sum_{i\in \mathds{N}}^{}\lambda _i\int_{s}^t \Bigl\lVert (N_1+H_0+\bar{\varepsilon} )^{1/2}W^{*}\bigl(\tilde{(\eta)}_s \bigr)  e^{-i \frac{s}{\varepsilon }\hat{H}_{ren}}\Psi _i\Bigr\rVert_{}^{}\\\cdot \Bigl\lVert (N_1+H_0+\bar{\varepsilon} )^{1/2}  e^{-i \frac{s}{\varepsilon }\hat{H}_{ren}}\Psi _i\Bigr\rVert_{}^{} ds\; ;
    \end{split}
  \end{equation*}
  then using Lemma~\ref{lemma:9}, and the fact that
  $\bigl\lVert \tilde{(\eta_1)}_s \bigr\rVert_{H^1}^{}=\bigl\lVert
  \eta_1 \bigr\rVert_{H^1}^{}$,
  $\bigl\lVert \tilde{(\eta_2)}_s
  \bigr\rVert_{\mathcal{F}H^{1/2}}^{}=\bigl\lVert \eta_2
  \bigr\rVert_{\mathcal{F}H^{1/2}}^{}$ we get
  \begin{equation*}
    \begin{split}
            X_1\leq C(\eta )\sum_{j=0}^3 \varepsilon ^jC_j(\eta )\int_{s}^t \Tr\Bigl[\varrho _{\varepsilon }(s)(N_1+H_0+\bar{\varepsilon} ) \Bigr] ds\\\leq \lvert t-s  \rvert_{}^{}C(\eta) \sum_{j=0}^3 \bar{\varepsilon}^jC_j(\eta )\Bigl(\tfrac{C}{1-a(\mathfrak{C})}+\tfrac{2b(\mathfrak{C})}{1-a(\mathfrak{C})}+\bar{\varepsilon}  \Bigr)\; ;
          \end{split}
  \end{equation*}
  where in the last inequality we used Equation~\eqref{eq:106} of
  Lemma~\ref{lemma:13}. Now let's consider $X_2$; a standard
  manipulation using Weyl's relation yields
  \begin{equation*}
    \begin{split}
      X_2\leq \Bigl\lVert \bigl(e^{i \frac{\varepsilon }{2}\Im\langle \xi   , \eta \rangle_{L^2\oplus L^2}}W(\xi -\eta )-1\bigr)(N_1+N_2+1)^{-1}  \Bigr\rVert_{\mathcal{L}\bigl(\Gamma _s(L^2\oplus L^2)\bigr)}^{}\Tr\Bigl[\tilde{\varrho}_{\varepsilon }(t)(N_1+N_2+1)\Bigr]\; .
    \end{split}
  \end{equation*}
  Now we use the estimate in  \citep[][Lemma 3.1]{ammari:nier:2008} and obtain
  \begin{equation*}
    \begin{split}
      X_2\leq \lVert \xi -\eta  \rVert_{L^2\oplus L^2}^{}\Bigl(\bar{\varepsilon} \lVert \eta   \rVert_{L^2\oplus L^2}^{}+1\Bigr)\Tr\Bigl[\tilde{\varrho}_{\varepsilon }(t)(N_1+N_2+1)\Bigr]\\
      \leq \lVert \xi -\eta  \rVert_{L^2\oplus L^2}^{}\Bigl(\bar{\varepsilon} \lVert \eta   \rVert_{L^2\oplus L^2}^{}+1\Bigr)\Bigl(\tfrac{C}{1-a(\mathfrak{C})}+\tfrac{2b(\mathfrak{C})}{1-a(\mathfrak{C})}+1\Bigr)\; ,
    \end{split}
  \end{equation*}
  where in the last inequality we used again Equation~\eqref{eq:106}
  of Lemma~\ref{lemma:13}, keeping in mind that
  $N_2\leq d\Gamma (\omega )\leq H_0$.
\end{proof}
Now using Lemma~\ref{lemma:15} with the estimates on $X_1, X_2$ above and a diagonal extraction argument, we prove the following proposition. We omit the proof since it is similar  to
 \cite[Proposition 3.9]{2011arXiv1111.5918A}.
\begin{proposition}\label{prop:13}
  Let
  $(\varrho _{\varepsilon })_{\varepsilon \in (0,\bar{\varepsilon} )}$
  be a family of quantum states on $\mathcal{H}$ that satisfies
  Assumptions \eqref{eq:84} and \eqref{eq:arho}. Then for any sequence
  $(\varepsilon_k)_{k\in \mathds{N}}\subset (0,\bar{\varepsilon} ) $
  with $\lim_{k\to \infty }\varepsilon _k=0$, there exists a
  subsequence $(\varepsilon _{k_i})_{i\in \mathds{N}}$ such that there
  exists a map
  $\mu_t:\mathds{R}\to \mathfrak{P}\bigl(L^2\oplus L^2 \bigr) $
  verifying the following statements:
  \begin{gather}
    \label{eq:107}
    \varrho _{\varepsilon _{k_i}}(t) \rightarrow \mu _t\; ,\; \forall t\in \mathds{R}\; ;\\
    \label{eq:108}
    \tilde{\varrho} _{\varepsilon _{k_i}}(t) \rightarrow \tilde{\mu} _t\; ,\; \forall t\in \mathds{R}\;,\; \text{with}\; \tilde{\mu}_t=\mathbf{E}_0(-t)_{\#}\mu_{t}\; ;\\
    \label{eq:109}
    \varrho _{\varepsilon _{k_i}}(t)W(\tilde{\xi}_t ) \rightarrow \mu _{\xi ,t}\; ,\; \forall t\in \mathds{R}\;\text{and}\; \forall \xi \in L^2\oplus L^2\;,\;  \text{with}\; d\mu_{\xi ,t}(z)=e^{i \sqrt{2}\Re\langle \tilde{\xi }_{t}   , z \rangle}d\mu_{t}(z)\; ;
  \end{gather}
  where
  $\mathbf{E}_0(t)z =e^{-it(-\Delta +V)}u \oplus e^{-it\omega }\alpha
  $
  is the Hamiltonian flow associated with the free classical energy
  $\mathscr{E}_0$, and $\tilde{\xi }_t=\mathbf{E}_0(-t)\xi $. Moreover,
  $\mu_t$ and $\tilde{\mu}_t$ are both Borel probability measures on $Q(-\Delta +V)\oplus D(\omega ^{1/2})$.

\end{proposition}

\subsection{The classical limit of the integral formula.}
\label{sec:class-limit-integr}

We are finally ready to discuss the limit $\varepsilon \to 0$ of the
integral formula~\eqref{eq:92}. As a final preparation, we state  a
couple of preliminary lemmas. The first is a slight improvement of \citep[][Theorem
6.13]{ammari:nier:2008
}. The second can be easily proved by standard estimates on the symbol
$\mathscr{B}_0^{(m)}(\xi )$ which we recall for convenience:
\begin{equation}
  \label{eq:111}
  \begin{split}
    \mathscr{B}_0^{(m)}(\xi )(u,\alpha )=2 i\sqrt{2}\Bigl\langle \Re\mathcal{F}\bigl(\tfrac{\chi _{\sigma _0}\bar{\alpha} }{\sqrt{2\omega }}\bigr)(x)  , \Im\bigl(\bar{\xi}_1u \bigr)(x) \Bigr\rangle_2+ i\sqrt{2}\Bigl\langle u(x)  , \chi _m(D_{x})\Im\Bigl(\mathcal{F}\bigl(\tfrac{\chi _{\sigma _0}\bar{\xi }_2 }{\sqrt{2\omega }}\bigr)\Bigr)(x) u(x)\Bigr\rangle_2\\
    +i\sqrt{2}\Im\Bigl\langle u(x)  , \bigl(\chi _m(D_{(\cdot )})V_{\infty }*\bar{\xi}_{1}u\bigr)(x) u(x)\Bigr\rangle_2+\tfrac{i(2\pi )^{3/2}}{2M}\Im\Bigl\langle \xi_1 (x)  , \Bigl(\mathcal{F}^{-1}(\bar{r}_{\sigma _m}\alpha )^2+\mathcal{F}(r_{\sigma _m}\bar{\alpha} )^2\\+\mathcal{F}^{-1}(\bar{r}_{\sigma _m}\alpha ) \mathcal{F}(r_{\sigma _m}\bar{\alpha} )\Bigr)(x) u(x)\Bigr\rangle_2\\
    -\tfrac{2 \sqrt{2}(2\pi )^3}{M}\Im \Bigl\langle u(x)  , \chi _m(D_{x})\Im\Bigl(\mathcal{F}^{-1}(\bar{r }_{\infty } \xi _2)\Bigr)(x) \mathcal{F}^{-1}(\bar{r }_{\sigma _m }  \alpha )(x)u(x)\Bigr\rangle_2\\
    -\tfrac{i \sqrt{2}(2\pi )^{3/2}}{M}\Im\Bigl\langle \xi_1 (x)  , D_x\mathcal{F}^{-1}(\bar{r}_{\sigma _m}\alpha )(x) u(x)\Bigr\rangle_2  -\tfrac{i \sqrt{2}(2\pi )^{3/2}}{M}\Im\Bigl\langle \xi_1 (x)  , \mathcal{F}(r_{\sigma _m}\bar{\alpha} )(x) D_x u(x)\Bigr\rangle_2 \\
    +\tfrac{i \sqrt{2}(2\pi )^{3/2}}{M}\Im\Bigl\langle u (x)  , \chi_m(D_x)D_{x} \mathcal{F}^{-1}(\bar{r}_{\infty }\xi_2  )(x) u(x)\Bigr\rangle_2\; .
  \end{split}
\end{equation}
\begin{lemma}\label{lemma:16}
  Let
  $(\varepsilon_j )_{j\in \mathds{N}}\subset (0,\bar{\varepsilon} )$,
  $\lim_{j\to \infty }\varepsilon_j=0 $, and
  $\delta>0$.
  Furthermore, let $(\varrho _{\varepsilon_j })_{j\in \mathds{N}}$ be
  a sequence of normal states in $\mathcal{H}$   such that for some $C(\delta )>0$,
  \begin{equation}\label{ass:lem15}
  \Bigl\|(N_1+N_2)^{\delta/2 }\varrho _{\varepsilon_j } (N_1+N_2)^{\delta/2 }\Bigr\|_{\mathcal{L}^1(L^2\oplus L^2)}\leq
  C(\delta )\,,
  \end{equation}
   uniformly in $\varepsilon \in (0,\bar{\varepsilon} )$. Suppose that
  $\varrho_{\varepsilon_j }\rightarrow \mu\in \mathfrak{P}(L^2\oplus L^2)$ then the
  following statement is true:
\begin{equation*}
   \biggl(\forall \mathscr{A}\in \bigoplus_{\substack{(p,q)\in \mathds{N}^2\\ p+q< 2\delta}  }\mathcal{P}_{p,q}^{\infty }\bigl(L^2\oplus L^2 \bigr)\; ,\; \lim_{j\to \infty }\Tr\Bigl[\varrho _{\varepsilon_j}(\mathscr{A})^{Wick}\Bigr]=\int_{L^2\oplus L^2}^{}\mathscr{A}(z)d\mu (z)\biggr)\; .
\end{equation*}
\end{lemma}
\begin{proof}
By linearity it is enough to assume $\mathscr{A}\in\mathcal{P}_{p,q}^{\infty }\bigl(L^2\oplus L^2 \bigr)$ for $(p,q)\in \mathds{N}^2$ with $p+q< 2\delta$.
Let $(P_R)_{R>0}$ be an increasing family of finite rank orthogonal projections on
 $L^2$ such that the strong limit $s-\lim_{R\to+\infty}P_R=1$ holds.
Let $ \mathscr{A}_R(z):=\mathscr{A}(P_R\oplus P_R z)$ for any $z\in L^2\oplus L^2$.
One writes
\begin{eqnarray}
\left|\Tr\Bigl[\varrho _{\varepsilon_j} (\mathscr{A})^{Wick}\Bigr]-
\int_{L^2\oplus L^2}^{}\mathscr{A}(z)d\mu (z)\right|&\leq&
\left|\Tr\Bigl[\varrho _{\varepsilon_j} (\mathscr{A})^{Wick}\Bigr]-\Tr\Bigl[\varrho _{\varepsilon_j} (\mathscr{A}_R)^{Wick}\Bigr]
\right| \label{eq:124}\\
&+&\left|\Tr\Bigl[\varrho _{\varepsilon_j} (\mathscr{A}_R)^{Wick}\Bigr]-
\int_{L^2\oplus L^2}^{}\mathscr{A}_R(z)d\mu (z)\right| \label{eq:125}\\
&+&\left|\int_{L^2\oplus L^2}^{}\mathscr{A}_R(z)d\mu (z)-
\int_{L^2\oplus L^2}^{}\mathscr{A}(z)d\mu (z)\right|\,.\label{eq:126}
\end{eqnarray}
Using standard number estimates and the regularity of the states $(\varrho_{\varepsilon_j})_j$, one shows
\begin{eqnarray*}
\left|\Tr\Bigl[\varrho _{\varepsilon_j} (\mathscr{A}-\mathscr{A}_R)^{Wick}\Bigr]
\right|\leq ||(N_1+N_2)^{\delta/2 }\varrho _{\varepsilon_j } (N_1+N_2)^{\delta/2}||_{\mathcal{L}^1(L^2\oplus L^2)}\;\;||\mathscr{\tilde A}-\mathscr{\tilde A}_R||\,,
\end{eqnarray*}
where $\mathscr{\tilde A}$ and $\mathscr{\tilde A}_R$ denote the compact operators
satisfying $\mathscr{A}(z)=\langle z^{\otimes q}, \mathscr{\tilde A} z^{\otimes p}\rangle$
and $\mathscr{ A}_R(z)=\langle z^{\otimes q}, \mathscr{\tilde A}_R z^{\otimes p}\rangle$
  respectively. Since $\mathscr{\tilde A}_R=(P_R\oplus P_R)^{\otimes q} \mathscr{\tilde A}
  (P_R\oplus P_R)^{\otimes p} $ and $\mathscr{\tilde A}$ is compact, one shows that
  $\lim_{R\to+\infty}||\mathscr{\tilde A}-\mathscr{\tilde A}_R||= 0$.
  So the right hand side of \eqref{eq:124} can be made arbitrary small by choosing $R$ large enough. \\
  According to \cite[Theorem 6.2]{ammari:nier:2008}, the regularity of $(\varrho _{\varepsilon_j})_j$ insures the bound
  \begin{equation*}
  \int_{L^2\oplus L^2}^{} ||z||_{L^2\oplus L^2}^{2\delta} \,d\mu (z)\leq C(\delta)\,.
  \end{equation*}
Hence by dominated convergence the right hand side of \eqref{eq:126} can also be made arbitrary small when $R$ is large enough since $\mathscr{A}(z)$ and $\mathscr{A}_R(z)$ are both bounded by
$c ||z||_{L^2\oplus L^2}^{p+q}$ and $\mathscr{A}_R(z)$ converges pointwise to
$\mathscr{A}(z)$.\\
To handle the right hand side of \eqref{eq:125}, we use a further regularization.
Let $\chi\in C_0^\infty(\mathds{R})$, $0\leq\chi\leq 1$, $\chi(x)=1$ in a neighborhood of  $0$ and $\chi_m(x)=\chi(\frac{x}{m})$ for $m>0$. Recall that the Fock space has the decomposition
$\Gamma_s(L^2\oplus L^2)\equiv\Gamma_s(P_RL^2\oplus P_RL^2)\otimes \Gamma_s(P_R^{\perp}L^2\oplus P_R^{\perp}L^2)$ where $P_R^{\perp}=1-P_R$. In this representation
$\mathscr{A}_R^{Wick}\equiv(\mathscr{A}_R)^{Wick}_{|\Gamma_s(P_RL^2\oplus P_RL^2)}\otimes 1_{
\Gamma_s(P_R^{\perp}L^2\oplus P_R^{\perp}L^2)}$ and
$\varrho _{\varepsilon_j}\equiv \hat\varrho _{\varepsilon_j}$. Hence using reduced density matrices $\varrho^R _{\varepsilon_j}$ that are normalized positive trace-class operators in $\Gamma_s(P_RL^2\oplus P_RL^2)$, one writes
\begin{eqnarray*}
\Tr\Bigl[\varrho _{\varepsilon_j} (\mathscr{A}_R)^{Wick}\Bigr]= \Tr\Bigl[\hat\varrho _{\varepsilon_j} (\mathscr{A}_R)^{Wick}_{|\Gamma_s(P_RL^2\oplus P_RL^2)}\otimes 1_{
\Gamma_s(P_R^{\perp}L^2\oplus P_R^{\perp}L^2)}\Bigr]=
\Tr_{\Gamma_s(P_RL^2\oplus P_RL^2)}\Bigl[\varrho^R _{\varepsilon_j} (\mathscr{A}_R)^{Wick}\Bigr]\,.
\end{eqnarray*}
As in the proof of Lemma \ref{lemma:14}, the Wick calculus  gives that $(\mathscr{A}_R)^{Wick}$ can be written as an Anti-Wick operator by moving all the
$a^*$ to the right of $a$. So, one obtains  $(\mathscr{A}_R)^{Wick}=
(\mathscr{A}_R)^{A-Wick}+\varepsilon T$ with $T
(d\Gamma(P_R\oplus P_R)+1)^{-\frac{p+q}{2}}$ is bounded uniformly with
respect to $\varepsilon\in(0,\bar\varepsilon)$. We refer the reader
to \cite{ammari:nier:2008} where Weyl and Anti-Wick
quantization are explained for ``cylindrical'' symbols.
Hence
\begin{eqnarray*}
\lim_{j\to\infty}\Tr\Bigl[\varrho _{\varepsilon_j} (\mathscr{A}_R)^{Wick}\Bigr]=\lim_{j\to\infty}\Tr_{\Gamma_s(P_RL^2\oplus P_RL^2)}\Bigl[\varrho^R _{\varepsilon_j} (\mathscr{A}_R)^{Wick}\Bigr]= \lim_{j\to\infty}\Tr_{\Gamma_s(P_RL^2\oplus P_RL^2)}\Bigl[\varrho^R _{\varepsilon_j} (\mathscr{A}_R)^{A-Wick}\Bigr]\,.
\end{eqnarray*}
Now we define $\chi_{m,R}(z):=\chi_m(|P_R\oplus P_R z|^2)$ and
$\varrho^{R,m} _{\varepsilon_j} :=\chi_{m,R}(z)^{Weyl} \,
\varrho^{R} _{\varepsilon_j} \,\chi_{m,R}(z)^{Weyl}$. So one writes
\begin{eqnarray}
\left|\Tr\Bigl[\varrho^R_{\varepsilon_j} (\mathscr{A})^{A-Wick}\Bigr]-
\int_{L^2\oplus L^2}^{}\mathscr{A}(z)d\mu (z)\right|&\leq&
\left|\Tr\Bigl[(\varrho^R _{\varepsilon_j}-\varrho^{R,m} _{\varepsilon_j}) (\mathscr{A})^{A-Wick}\Bigr]
\right| \label{eq:127}\\
&&\hspace{-.3in}+\left|\Tr\Bigl[\varrho^{R,m} _{\varepsilon_j} (\mathscr{A}_R)^{A-Wick}\Bigr]-
\int \chi_{m,R}^2( z)\mathscr{A}_R(z)d\mu (z)\right| \label{eq:128}\\
&&\hspace{-.3in}+\left|\int \chi_{m,R}^2(z)\mathscr{A}_R(z)d\mu (z)-
\int\mathscr{A}_{R}(z)d\mu (z)\right|\,,\label{eq:129}
\end{eqnarray}
where the traces are on the Fock space $\Gamma_s(P_RL^2\oplus P_RL^2)$
and the integrals are over $L^2\oplus L^2$. By dominated convergence the right hand side of
\eqref{eq:129} tends  to $0$ when $m\to \infty$ at fixed $R$.
The right hand side of \eqref{eq:127} can be made arbitrary small when
$m\to\infty$ using the following decomposition
\begin{eqnarray*}
(\varrho^{R,m} _{\varepsilon_j}-\varrho^{R} _{\varepsilon_j})
=\underbrace{ (\chi_{m,R}^{Weyl}-1)\, \varrho^R _{\varepsilon_j}\, \chi_{m,R}^{Weyl}}_{(A)}
+\underbrace{ \varrho^R _{\varepsilon_j} \,(\chi_{m,R}^{Weyl}-1)}_{(B)}
\,,
\end{eqnarray*}
which gives   $\Tr\Bigl[(A)\,
(\mathscr{A}_R)^{A-Wick}\Bigr]=\Tr\Bigl[ T_{1} T_{2} T_{3}
T_{4}\Bigr]$ and a similar expression for $(B)$ with
\begin{eqnarray*}
&T_{1}=(N_{R}+1)^{\frac{p+q}{4}} (\chi_{m,R}^{Weyl}-1)
(N_{R}+1)^{-\frac{\delta}{2}} \,,
& \quad T_{2}=
(N_{R}+1)^{\frac{\delta}{2}} \varrho^R _{\varepsilon_j}
(N_{R}+1)^{\frac{\delta}{2}} \\
&T_{3}=(N_{R}+1)^{-\frac{\delta}{2}}
\chi_{m,R}^{Weyl} (N_{R}+1)^{\frac{p+q}{4}}\,,  & \quad
T_{4}=  (N_{R}+1)^{-\frac{p+q}{4}} (\mathscr{A}_R)^{A-Wick} (N_{R}+1)^{-\frac{p+q}{4}}\,,
\end{eqnarray*}
where $N_{R}=d\Gamma(P_{R}\oplus P_{R})$. The
Weyl-H\"ormander Pseudo-differential calculus gives that
$T_{1}\to_{m\to\infty} 0$  in norm (since $\delta >p+q$) and that
$T_{i}$, $i=2,3,4$, are uniformly bounded with respect $j\in \mathds{N}$ and $m>0$
at fixed $R$ (see e.g.
\cite[Proposition 3.2 and 3.3]{ammari:nier:2008}).\\
To complete the proof, we remark
that $ \Tr\Bigl[\varrho^{R,m} _{\varepsilon_j}
(\mathscr{A}_R)^{A-Wick}\Bigr]=
\Tr\Bigl[\varrho^{R} _{\varepsilon_j} \;\chi_{m,R}^{Weyl}\;
(\mathscr{A}_R)^{A-Wick} \;\chi_{m,R}^{Weyl}\Bigr]$. So again by
pseudo-differential
calculus we know  that $(\mathscr{A}_R)^{A-Wick}=
(\mathscr{A}_R)^{Weyl}+\varepsilon \;b(\varepsilon) ^{Weyl}$ with
$b(\varepsilon)$ belonging to the  Weyl–H\"ormander class symbol
$S_{P_{R}\oplus P_{R}}(\langle z\rangle^{p+q-2},
\frac{dz^{2}}{\langle z\rangle^{2}})$
uniformly in $\varepsilon$ (see \cite[Section 3.2 and
3.4]{ammari:nier:2008}). Therefore
\begin{equation*}
\lim_{j\to \infty}\Tr\Bigl[\varrho^{R,m} _{\varepsilon_j}
(\mathscr{A}_R)^{A-Wick}\Bigr]=\lim_{j\to \infty}\Tr\Bigl[
\varrho^{R} _{\varepsilon_j} \;\chi_{m,R}^{Weyl}
(\mathscr{A}_R)^{Weyl} \;\chi_{m,R}^{Weyl}\Bigr]\,,
\end{equation*}
since $ (d\Gamma(P_{R}\oplus P_{R})+1)^{-(q+p)/2}  \;b(\varepsilon)^{Weyl}\;
(d\Gamma(P_{R}\oplus P_{R})+1)^{-(p+q)/2} $ is uniformly bounded with
respect to $\varepsilon$. The Weyl-H\"ormander pseudo-differential
calculus gives $ \;\chi_{m,R}^{Weyl}
(\mathscr{A}_R)^{Weyl}
\;\chi_{m,R}^{Weyl}=(\chi_{m,R}^{2}\mathscr{A}_R )^{Weyl}+\varepsilon \,
c(\varepsilon)^{Weyl} $ with
$c(\varepsilon)\in S_{P_{R}\oplus P_{R}}(1,dz^{2})$
uniformly in $\varepsilon$
(see e.g. \cite[Proposition 3.2]{ammari:nier:2008}). Hence,
according to \cite[Theorem 6.2]{ammari:nier:2008} one obtains
\begin{equation*}
\lim_{j\to \infty}\Tr\Bigl[\varrho^{R,m} _{\varepsilon_j}
(\mathscr{A}_R)^{A-Wick}\Bigr]=\lim_{j\to \infty}\Tr\Bigl[
\varrho _{\varepsilon_j} \;
(\chi_{m,R}^{2}\,\mathscr{A}_R)^{Weyl} \Bigr]= \int_{L^{2}\oplus L^{2}} \chi_{m,R}^2( z)\mathscr{A}_R(z)d\mu (z)\,.
\end{equation*}
This yields the intended bound on \eqref{eq:125} and completes the proof.
\end{proof}

\begin{lemma}\label{lemma:17}
  There exists  $C(\sigma _0)>0$ depending only on $\sigma _0\in\mathds{R}_+$ such that the
  following bound holds for $\mathscr{B}_0^{(m)}$ uniformly in
  $m\in \mathds{N}$:
  \begin{equation}\label{eq:112}
    \begin{split}
      \Bigl\lvert \mathscr{B}_0^{(m)}(\xi )(u,\alpha )\Bigr  \rvert_{}^{}\leq C(\sigma _0)\lVert \xi   \rVert_{L^2\oplus L^2}^{}\Bigl(\lVert u  \rVert_2^{2}+ \lVert (-\Delta +V)^{1/2}u  \rVert_2^2+\lVert \alpha   \rVert_{\mathcal{F}H^{1/2}}^{2}+\lVert u  \rVert_2\cdot \lVert (-\Delta +V)^{1/2}u  \rVert_2^2\\
      +\lVert u  \rVert_2\cdot \lVert \alpha   \rVert_{\mathcal{F}H^{1/2}}^2 +\lVert u  \rVert_2\cdot \lVert (-\Delta +V)^{1/2}u  \rVert_2^{}\cdot  \lVert \alpha   \rVert_{\mathcal{F}H^{1/2}}\Bigr)\; .
    \end{split}
  \end{equation}
  It follows that:
  \begin{itemize}[label=\color{myblue}*]
  \item For any $\xi \in L^2\oplus L^2$, for any
    $(u,\alpha )\in Q(-\Delta +V)\oplus D(\omega^{1/2} )$, $
      \lim_{m\to \infty }\mathscr{B}_0^{(m)}(\xi )(u,\alpha )=\mathscr{B}_0(\xi )(u,\alpha )\;; $ and therefore the bound \eqref{eq:112} holds also for
    $\mathscr{B}_0$.
  \item For any $m\in \mathds{N}$,
    $\mathscr{B}_0^{(m)}(\cdot ),\mathscr{B}_0(\cdot )$ are are jointly continuous with respect to $\xi\in L^2\oplus L^2 $ and
    $(u,\alpha )\in Q(-\Delta
    +V)\oplus D(\omega ^{1/2})$.
  \end{itemize}
\end{lemma}
Recall that for any $\sigma_0\geq2K (\mathfrak{C}+1+\bar\varepsilon)$ there exists $b>0$ such that the operator $\hat H_{ren}(\sigma_0)+b$ is non-negative uniformly for $\varepsilon\in(0,\bar \varepsilon)$. Let
$(\varrho_{\varepsilon } )_{\varepsilon \in (0,\bar{\varepsilon} )}$
be a family of normal states on
$\Gamma _s(L^2 (\mathds{R}^3 )\oplus L^2 (\mathds{R}^3 ))$, we
consider the additional assumption:
\begin{gather}
  \label{eq:130}\tag{$A''_\rho$}
  \exists C>0  \, ,\, \forall \varepsilon \in (0,\bar{\varepsilon} )\, ,\, \Tr[\varrho _{\varepsilon }\,(\hat H_{ren}(\sigma_0)+b)^2]\leq C\; ;
\end{gather}
\begin{proposition}\label{prop:14}
  Let
  $(\varrho _{\varepsilon })_{\varepsilon \in (0,\bar{\varepsilon}
    )}\subset \mathcal{L}^1(\mathcal{H})$
  be a family of normal states that satisfy Assumptions~\eqref{eq:84},
  \eqref{eq:arho} and~\eqref{eq:130} such that\footnote{We recall that $\mathfrak{C}$
  appears in Assumption \eqref{eq:84} and $\sigma_{0}$ in
  Definition \ref{def:7} of $\hat
  H_{ren}(\sigma_{0})$. The condition $\sigma_{0}\geq K (\mathfrak{C}+1)$ ensures  that the dressed dynamics is non-trivial on $\bigoplus _{n=0}^{[\mathfrak{C}/\varepsilon ]}\mathcal{H}_n$ and hence non-trivial on the state $\varrho_\varepsilon$ according to Lemma \ref{lemma:10}.
  } $\sigma_{0}\geq 2K (\mathfrak{C}+1+\bar\varepsilon)$. Then:\\
$(i)$ For any sequence
  $(\varepsilon _k)_{k\in \mathds{N}}\subset (0,\bar{\varepsilon} )$
  converging to zero, there exist a subsequence
  $(\varepsilon_{k_{\iota}} )_{\iota\in \mathds{N}}$ and a map $\mu_t:\mathds{R}\to \mathfrak{P}\bigl(L^2\oplus L^2 \bigr) $ such
  that $\varrho _{\varepsilon _{k_\iota}}(t)\rightarrow \mu _t$ and
  $\tilde{\varrho}_{\varepsilon _{k_\iota}}(t)\rightarrow
  \tilde{\mu}_t=\mathbf{E}_0(-t)_{\#}\mu _t$
  , for any $t\in \mathds{R}$ . \\
$(ii)$ The action of $e^{-i\frac{t}{\varepsilon} \hat
  H_{ren}(\sigma_{0})}$ is non-trival on the states $\varrho_{\varepsilon}$.\\
$(iii)$ The Fourier transform of
  $\tilde{\mu}_{(\cdot )} $ satisfies the following transport equation
  $\forall \xi \in L^2\oplus L^2$:
  \begin{equation}
    \label{eq:113}
    \int_{L^2\oplus L^2}^{}e^{i \sqrt{2}\Re\langle \xi   , z \rangle_{}}  d\tilde{\mu}_t(z)= \int_{L^2\oplus L^2}^{}e^{i \sqrt{2}\Re\langle \xi   , z \rangle_{}}  d\mu_0(z)+\int_0^t  \biggl(\int_{L^2\oplus L^2}^{}\mathscr{B}_0(\tilde{\xi}_s )(z)e^{i \sqrt{2}\Re\langle \tilde{\xi}_s   , z \rangle_{}}  d\mu_s(z)\biggr)ds\; ;
  \end{equation}
  where the right hand side makes sense since
  $\mathscr{B}_0(\tilde{\xi}_t )(z)e^{i \sqrt{2}\Re\langle
    \tilde{\xi}_t , z \rangle_{}}\in L^{\infty }_{t}\Bigl(\mathds{R},
  L^1_{z}\bigl[L^2\oplus L^2,d\mu_t(z)\bigr]\Bigr)$
  for any $\xi \in L^2\oplus L^2$.
\end{proposition}
\begin{proof}
  The first part of the proposition $(i)-(ii)$ is just a partial restatement of
  Proposition~\ref{prop:13}. We discuss the last assertion in $(iii)$ about
  $\mathscr{B}_0(\tilde{\xi}_t )(z)e^{i \sqrt{2}\Re\langle
    \tilde{\xi}_t , z \rangle_{}}$,
  before proving \eqref{eq:113}. Recall the fact that for any
  $\xi \in L^2\oplus L^2$ and for any $t\in \mathds{R}$,
  $\lVert \tilde{\xi}_{t} \rVert_{L^2\oplus L^2}^{}=\lVert \xi
  \rVert_{L^2\oplus L^2}^{}$.
  Using bound~\eqref{eq:112} of Lemma~\ref{lemma:17} we obtain,
  setting $Q(-\Delta +V)\oplus D(\omega ^{1/2}) \ni z=(u,\alpha )$:
  \begin{equation*}
    \begin{split}
          \Bigl\lvert \mathscr{B}_0(\tilde{\xi}_t )(z)e^{i \sqrt{2}\Re\langle\tilde{\xi}_t , z \rangle_{}}  \Bigr\rvert_{}^{}\leq C(\sigma _0)\lVert \xi   \rVert_{L^2\oplus L^2}^{}\Bigl(\lVert u  \rVert_2^{2}+ \lVert (-\Delta +V)^{1/2}u  \rVert_2^2+\lVert \alpha   \rVert_{\mathcal{F}H^{1/2}}^{2}\\+\lVert u  \rVert_2\cdot \lVert (-\Delta +V)^{1/2}u  \rVert_2^2
      +\lVert u  \rVert_2\cdot \lVert \alpha   \rVert_{\mathcal{F}H^{1/2}}^2 +\lVert u  \rVert_2\cdot \lVert (-\Delta +V)^{1/2}u  \rVert_2^{}\cdot  \lVert \alpha   \rVert_{\mathcal{F}H^{1/2}}\Bigr)\; .
    \end{split}
  \end{equation*}
  Now
  $\mu _t\in \mathcal{M}\bigl(\varrho _{\varepsilon }(t), \varepsilon
  \in (0,\bar{\varepsilon} )\bigr)$,
  therefore by Lemma~\ref{lemma:14},
  $\mu_t \Bigl(B_u(0,\sqrt{\mathfrak{C}})\cap
  Q(-\Delta +V)\oplus D(\omega ^{1/2})\Bigr)=1$
  for any $t\in \mathds{R}$. Then it follows that there
  exists
  $C(\mathfrak{C})>0$ such that
  \begin{equation*}
    \begin{split}
      \biggl\lvert\int_{L^2\oplus L^2}^{}  \mathscr{B}_0(\tilde{\xi}_t )(z)e^{i \sqrt{2}\Re\langle\tilde{\xi}_t , z \rangle_{}}d\mu_t(z)   \biggr\rvert_{}^{}\leq C(\mathfrak{C}) \lVert \xi   \rVert_{L^2\oplus L^2}^{}\int_{L^2\oplus L^2}^{}\Bigl(\lVert (-\Delta +V)^{1/2}u  \rVert_2^2+\lVert \alpha   \rVert_{\mathcal{F}H^{1/2}}^{2}\Bigr)  d\mu _t(z)\\\leq C(\mathfrak{C}) \lVert \xi   \rVert_{L^2\oplus L^2}^{}J(t)\; ;
    \end{split}
  \end{equation*}
  where $J(t)<\infty $ by Lemma~\ref{lemma:14}. Actually, using the
  fact that the bound~\eqref{eq:106} is independent of $t$, it is
  easily proved that $J(t)$ does not depend on $t$ as well, i.e.
  $J(t)\in L^{\infty }(\mathds{R})$.

  We prove \eqref{eq:113} by successive
  approximations. Consider
  $\Tr\Bigl[\tilde{\varrho}_{\varepsilon _{k_{\iota}}}(t)W(\xi )
  \Bigr]$,
  $\xi \in L^2\oplus L^2$. We can approximate $\xi $ with
  $(\xi^{(l)})_{l\in \mathds{N}} \subset Q(-\Delta +V)\oplus D(\omega
  ^{3/4})$,
  since the latter is dense in $L^2\oplus L^2$, and
  $\lim_{l\to \infty }\Tr\Bigl[\tilde{\varrho}_{\varepsilon
    _{k_{\iota}}}(t)\bigl(W(\xi ) -W(\xi ^{(l)})\bigr)\Bigr]=0$
  uniformly in $\varepsilon _{k_{\iota}}$ by
  Lemma~\ref{lemma:15}. Now, for
  $\Tr\Bigl[\tilde{\varrho}_{\varepsilon _{k_{\iota}}}(t)W(\xi^{(l)} )
  \Bigr]$
  the integral equation~\eqref{eq:92} holds. Proposition~\ref{prop:13}
  implies that
  $\tilde{\varrho}_{\varepsilon _{k_{\iota}}}(t)\rightarrow
  \tilde{\mu}_t=\mathbf{E}_0(t)_{\#}\mu _t $,
  for any $t\in \mathds{R}$. Therefore the left-hand side of
  \eqref{eq:92} converges when $\iota\to \infty $ to
  $\int_{L^2\oplus L^2}^{}e^{i \sqrt{2}\Re\langle \xi^{(l)} , z
    \rangle_{}} d\tilde{\mu}_t(z)$;
  and that in turn converges when $l\to \infty $ to
  $\int_{L^2\oplus L^2}^{}e^{i \sqrt{2}\Re\langle \xi , z \rangle_{}}
  d\tilde{\mu}_t(z)$ by dominated convergence theorem. In addition,
  \begin{equation*}
    \lim_{\iota\to \infty }\sum_{j=1}^3\varepsilon ^j\int_0^t\Tr\Bigl[\varrho_{\varepsilon _{k_{\iota}}}(s)W(\tilde{\xi^{(l)}}_s)B_j(\tilde{\xi^{(l)}}_s) \Bigr]  ds=0\; ;
  \end{equation*}
  by Proposition~\ref{prop:11}. It remains to show the convergence of
  the $B_0$ term in \eqref{eq:92}. We approximate
  $\mathscr{B}_0$ by the compact $\mathscr{B}_0^{(m)}$, because using
  Lemma~\ref{lemma:10} and \eqref{eq:99} of
  Proposition~\ref{prop:12} we obtain
  \begin{equation*}
    \begin{split}
      \Bigl\lvert \Tr\Bigl[\varrho_{\varepsilon _{k_{\iota}}}(s)W(\tilde{\xi^{(l)}}_s)\Bigl(B_0(\tilde{\xi^{(l)}}_s)-B^{(m)}_0(\tilde{\xi^{(l)}}_s)\Bigr) \Bigr]  \Bigr\rvert_{}^{}\leq \sum_{i\in \mathds{N}}^{}\lambda _i\Bigl\lvert \Bigl\langle W^{*}(\tilde{\xi^{(l)}}_s) e^{-i \frac{s}{\varepsilon_{k_{\iota}} }\hat{H}_{ren}}\Psi _i , \Bigl(B_0(\tilde{\xi^{(l)}}_s)\\-B^{(m)}_0(\tilde{\xi^{(l)}}_s)\Bigr) e^{-i \frac{s}{\varepsilon_{k_{\iota}} }\hat{H}_{ren}}\Psi _i \Bigr\rangle_{}  \Bigr\rvert_{}^{}\\
      \leq \sum_{i\in \mathds{N}}^{}\lambda _i C^{(m)}(\tilde{\xi^{(l)}}_s) \Bigl\lVert (H_0+1)^{1/2}(N_1+\bar{\varepsilon} )^{1/2}W^{*}(\tilde{\xi^{(l)}}_s) e^{-i \frac{s}{\varepsilon_{k_{\iota}} }\hat{H}_{ren}}\Psi _i  \Bigr\rVert_{}^{}\\\cdot \Bigl\lVert (H_0+1)^{1/2}(N_1+\bar{\varepsilon} )^{1/2}e^{-i \frac{s}{\varepsilon_{k_{\iota}} }\hat{H}_{ren}}\Psi _i  \Bigr\rVert_{}^{}\; .
    \end{split}
  \end{equation*}
  Now, using the fact that $C^{(m)}(\tilde{\xi^{(l)}}_s)$ depends only
  on
  $\lVert \tilde{\xi^{(l)}}_s \rVert_{Q(-\Delta +V)\oplus D(\omega
    ^{3/4})}^{}=\lVert \xi^{(l)} \rVert_{Q(-\Delta +V)\oplus D(\omega
    ^{3/4})}^{}$ and Lemma~\ref{lemma:9} we obtain
  \begin{equation*}
    \begin{split}
      \Bigl\lvert \Tr\Bigl[\varrho_{\varepsilon _{k_{\iota}}}(s)W(\tilde{\xi^{(l)}}_s)\Bigl(B_0(\tilde{\xi^{(l)}}_s)-B^{(m)}_0(\tilde{\xi^{(l)}}_s)\Bigr) \Bigr]  \Bigr\rvert_{}^{}\leq \sum_{i\in \mathds{N}}^{}\lambda _i C^{(m)}(\xi^{(l)})C(\xi^{(l)}) \Bigl\lVert (H_0+1)^{1/2}e^{-i \frac{s}{\varepsilon_{k_{\iota}} }\hat{H}_{ren}}\\(N_1+\bar{\varepsilon} )^{1/2}\Psi _i  \Bigr\rVert_{}^2\; .
    \end{split}
  \end{equation*}
  We then use Equation~\eqref{eq:106} of Lemma~\ref{lemma:13}:
  \begin{equation*}
    \begin{split}
      \Bigl\lvert \Tr\Bigl[\varrho_{\varepsilon _{k_{\iota}}}(s)W(\tilde{\xi^{(l)}}_s)\Bigl(B_0(\tilde{\xi^{(l)}}_s)-B^{(m)}_0(\tilde{\xi^{(l)}}_s)\Bigr) \Bigr]  \Bigr\rvert_{}^{}\leq \sum_{i\in \mathds{N}}^{}\lambda _i C^{(m)}(\xi^{(l)})C(\xi^{(l)}) (\mathfrak{C}+\bar{\varepsilon} ) \tfrac{1}{1-a(\mathfrak{C})}C\\+\tfrac{2b(\mathfrak{C})}{1-a(\mathfrak{C})}\; .
    \end{split}
  \end{equation*}
  The right hand side goes to zero when $m\to \infty $ uniformly with
  respect to $\varepsilon_{k_{\iota}} $ and $s$ by
  Proposition~\ref{prop:12}, and therefore
  \begin{equation*}
    \lim_{m\to \infty }\int_0^t \Tr\Bigl[\varrho_{\varepsilon _{k_{\iota}}}(s)W(\tilde{\xi^{(l)}}_s)\Bigl(B_0(\tilde{\xi^{(l)}}_s)-B^{(m)}_0(\tilde{\xi^{(l)}}_s)\Bigr) \Bigr]  ds=0\; .
  \end{equation*}
So the next step is to prove
\begin{eqnarray*}
\lim_{\iota\to \infty } \Tr\Bigl[\varrho_{\varepsilon _{k_{\iota}}}(s)W(\tilde{\xi^{(l)}}_s)\Bigl(\mathscr{B}_0^{(m)}(\tilde{\xi^{(l)}}_s)\Bigr)^{Wick} \Bigr] =
\int_{L^2\oplus L^2}^{}\mathscr{B}_0^{(m)}(\tilde{\xi}^{(l)}_s )(z)\,e^{i \sqrt{2}\Re\langle \tilde{\xi}^{(l)}_s   , z \rangle_{}}  d\mu_s(z).
\end{eqnarray*}
This statement follows by applying  Lemma~\ref{lemma:16} with $\delta=2$ and by checking
the assumption
\begin{equation}
\label{ass:wlem15}
||(N_1+N_2)\,\varrho_{\varepsilon _{k_{\iota}}}(s) W(\tilde{\xi^{(l)}}_s)\,(N_1+N_2)||_{\mathcal{L}^1(L^2\oplus L^2)}\leq C\,,
\end{equation}
uniformly in $k_{\iota}$ for some $C>0$. In fact \eqref{ass:wlem15} holds true  by
Assumptions~\eqref{eq:84}-\eqref{eq:130}, the Higher order estimate of Proposition \ref{prop:A1} and Lemma \ref{lemma:9}. Remark that while $\varrho_{\varepsilon _{k_{\iota}}}(s)W(\tilde{\xi^{(l)}}_s)$ is not a non-negative trace-class operator, one can still apply Lemma~\ref{lemma:16}. In fact, one can write
 \begin{equation*}
    \Tr\Bigl[\varrho_{\varepsilon _{k_{\iota}}}(s)W(\tilde{\xi^{(l)}}_s)\;B^{(m)}_0(\tilde{\xi^{(l)}}_s) \Bigr] =
    \Tr\Bigl[W(\eta)\varrho_{\varepsilon _{k_{\iota}}}(s) W(\eta) \;\mathscr{A}^{Wick} \Bigr]\, ,
  \end{equation*}
 for some $\mathscr{A}\in \bigoplus_{p+q< 4}  \mathcal{P}_{p,q}^{\infty }\bigl(L^2\oplus L^2 \bigr)$ and with  $\eta=\frac{1}{2} \tilde{\xi^{(l)}}_s $. Remark now that
  $W(\eta)\varrho_{\varepsilon _{k_{\iota}}}(s) W(\eta)$ decomposes explicitly  into a linear combination of non-negative trace-class operators  satisfying all the assumption \eqref{ass:lem15} of Lemma~\ref{lemma:16}. Note that the Wigner measures of $\varrho_{\varepsilon _{k_{\iota}}}(s)W(\tilde{\xi^{(l)}}_s)$ are identified through \eqref{eq:109}. Hence the dominated  convergence theorem yields:
  \begin{equation*}
    \lim_{\iota\to \infty }\int_0^t \Tr\Bigl[\varrho_{\varepsilon _{k_{\iota}}}(s)W(\tilde{\xi^{(l)}}_s)B^{(m)}_0(\tilde{\xi^{(l)}}_s) \Bigr]  ds=\int_0^t  \biggl(\int_{L^2\oplus L^2}^{}\mathscr{B}_0^{(m)}(\tilde{\xi^{(l)}}_s )(z)e^{i \sqrt{2}\Re\langle \tilde{\xi^{(l)}}_s   , z \rangle_{}}  d\mu_s(z)\biggr)ds\; .
  \end{equation*}
  By Lemma~\ref{lemma:17},
  $\lim_{m\to \infty }\mathscr{B}_0^{(m)}(\tilde{\xi^{(l)}}_s
  )(z)=\mathscr{B}_0(\tilde{\xi^{(l)}}_s )(z)$,
  so by dominated convergence theorem
  \begin{equation*}
    \lim_{m\to \infty }\int_0^t  \biggl(\int_{L^2\oplus L^2}^{}\mathscr{B}_0^{(m)}(\tilde{\xi^{(l)}}_s )(z)e^{i \sqrt{2}\Re\langle \tilde{\xi^{(l)}}_s   , z \rangle_{}}  d\mu_s(z)\biggr)ds=\int_0^t  \biggl(\int_{L^2\oplus L^2}^{}\mathscr{B}_0(\tilde{\xi^{(l)}}_s )(z)e^{i \sqrt{2}\Re\langle \tilde{\xi^{(l)}}_s   , z \rangle_{}}  d\mu_s(z)\biggr)ds\; .
  \end{equation*}
  Above it is possible to apply the dominated convergence theorem due
  to a reasoning analogous to the one done at the beginning of this
  proof: roughly speaking, we have that
  $\mathscr{B}_0^{(m)}(\tilde{\xi^{(l)}}_t)(z) e^{i \sqrt{2}\Re\langle
    \tilde{\xi^{(l)}}_t , z \rangle_{}}\in L^{\infty
  }_{t}\Bigl(\mathds{R}, L^1_{z}\bigl[L^2\oplus
  L^2,d\mu_t(z)\bigr]\Bigr) $
  uniformly with respect to $m\in \mathds{N}$. In an analogous fashion
  we finally obtain
  \begin{equation*}
    \lim_{l\to \infty }\int_0^t  \biggl(\int_{L^2\oplus L^2}^{}\mathscr{B}_0(\tilde{\xi^{(l)}}_s )(z)e^{i \sqrt{2}\Re\langle \tilde{\xi^{(l)}}_s   , z \rangle_{}}  d\mu_s(z)\biggr)ds=\int_0^t  \biggl(\int_{L^2\oplus L^2}^{}\mathscr{B}_0(\tilde{\xi}_s )(z)e^{i \sqrt{2}\Re\langle \tilde{\xi}_s   , z \rangle_{}}  d\mu_s(z)\biggr)ds\; .
  \end{equation*}
\end{proof}
\begin{corollary}\label{cor:3}
  The transport equation~\eqref{eq:113} may be rewritten as
  \begin{equation}
    \label{eq:114}
    \int_{L^2\oplus L^2}^{}e^{i \sqrt{2}\Re\langle \xi   , z \rangle_{}}  d\tilde{\mu}_t(z)= \int_{L^2\oplus L^2}^{}e^{i \sqrt{2}\Re\langle \xi   , z \rangle_{}}  d\mu_0(z)+i\sqrt{2}\int_0^t  \biggl(\int_{L^2\oplus L^2}^{}e^{i \sqrt{2}\Re\langle \xi  , z \rangle_{}}\Re\langle \xi   , \mathbf{V}(s)(z) \rangle_{}  d\tilde{\mu}_s(z)\biggr)ds\; ;
  \end{equation}
  with the velocity vector field
  $\mathbf{V}(t)(z)=-i\mathbf{E}_0(-t)\circ\partial_{\bar{z}}\bigl(\hat{\mathscr{E}}-\mathscr{E}_0
  \bigr)\circ\mathbf{E}_0(t)(z)$.   In addition
    $\tilde{\mu}_t=\mathbf{E}_0(-t)_{\#}\hat{\mathbf{E}}(t)_{\#}\mu
    _0$ is a solution of Equation~\eqref{eq:114}.
\end{corollary}
\begin{proof}
It is proved by direct calculation, since
  $\mu_t\bigl( Q(-\Delta +V)\oplus \mathcal{F}H^{1/2}\bigr)=1$ for any
  $t\in \mathds{R}$ by Lemma~\ref{lemma:14}; and
  $\hat{\mathbf{E}}(t),\mathbf{E}_0(t)$ are globally well-defined on
  this space (for $\hat{\mathbf{E}}(t)$, it is proved in
  Theorem~\ref{prop:9}; for $\mathbf{E}_0(t)$ it is trivial). The second point is proved by  differentiating with respect to time and using Lemma \ref{lemma:17}  and Lemma \ref{lemma:14}
   $(iii)$.
\end{proof}

\subsection{Uniqueness of the solution to the transport equation.}
\label{sec:uniq-solut-transp}

As discussed in Corollary~\ref{cor:3}, the dressed flow
yields in the classical limit a solution of the transport equation~\eqref{eq:114}. The second
part of the same corollary suggests that it is important to study uniqueness
properties of~\eqref{eq:114}: it is by means of uniqueness that we can
close the argument and reach a satisfactory characterization of the
dynamics of classical states (Wigner measures). This subsection is devoted to prove that
the family of Wigner measures $\tilde{\mu}_t $ of
Proposition~\ref{prop:14} satisfies sufficient conditions, induced by
the properties of
$(\varrho _{\varepsilon })_{\varepsilon \in (0,\bar{\varepsilon} )}$,
to be uniquely identified with
$\mathbf{E}_0(-t)_{\#} \hat{\mathbf{E}}(t)_{\#}\mu _0$. We use an
optimal transport technique initiated  by
\citet{AGS
} then extended by \citet{2011arXiv1111.5918A} to propagation of Wigner measures; and under minimal assumptions proved to be sufficient by
\citet{liard:2015
} (see also \cite{MR2335089}).

In order to do that, we need to introduce a suitable topology on
$\mathfrak{P}\bigl(L^2\oplus L^2\bigr)$. Let $(e_j)_{j\in \mathds{N}}\subset L^2\oplus L^2$
be an orthonormal basis. Then
\begin{equation}
  \label{eq:122}
  d_{w}(z_1,z_2)=\biggl(\sum_{j\in \mathds{N}}^{}\frac{\lvert \langle z_1-z_2  , e_j \rangle_{L^2\oplus L^2}  \rvert_{}^2}{(1+j)^2}\biggr)^{1/2}\; ,
\end{equation}
where $z_1,z_2\in L^2\oplus L^2$, defines a distance on
$L^2\oplus L^2$. The topology induced by
$\bigl(L^2\oplus L^2,d_w \bigr)$ is homeomorphic to the weak topology
on bounded sets.

\begin{definition}[Weak narrow convergence of probability measures]\label{def:11}
  Let $(\mu _i)_{i\in \mathds{N}}\subset \mathfrak{P}\bigl(L^2\oplus L^2 \bigr)$.
  Then $(\mu _i)_{i\in \mathds{N}}$ weakly narrowly converges to
  $\mu \in \mathfrak{P}\bigl(L^2\oplus L^2 \bigr)$, in symbols
  $\mu _i\overset{n}{\rightharpoonup}\mu $, if
  \begin{equation*}
    \forall f\in \mathcal{C}_b\Bigl(\bigl(L^2\oplus L^2,d_w\bigr),\mathds{R}\Bigr)\; , \; \lim_{i\to \infty }\int_{L^2\oplus L^2}^{}f(z)  d\mu _i(z)=\int_{L^2\oplus L^2}^{}f(z)  d\mu (z)\; ;
  \end{equation*}
  where
  $\mathcal{C}_b\Bigl(\bigl(L^2\oplus L^2,d_w\bigr),\mathds{R}\Bigr)$
  is the space of bounded continuous real-valued functions on
  $\bigl(L^2\oplus L^2,d_w\bigr)$.
\end{definition}
It is actually more convenient to use cylindrical functions to prove
narrow continuity properties. Therefore we also define the cylindrical
Schwartz and compact support space of functions on $L^2\oplus L^2$.
\begin{definition}[Spaces of cylindrical functions]\label{def:12}
  Let $f:L^2\oplus L^2\to \mathds{R}$. Then
  $f\in \mathcal{S}_{cyl}\bigl(L^2\oplus L^2 \bigr)$ if there exists
  an orthogonal projection
  $\mathbf{p}:L^2\oplus L^2\to L^2\oplus L^2$,
  $\mathrm{dim}(\ran \mathbf{p})=d<\infty $, and a rapid decrease
  function $g\in \mathcal{S}(\ran\mathbf{p})$ such that
  \begin{equation*}
    \forall z\in L^2\oplus L^2 \, ,\, f(z)=g(\mathbf{p}z)\;.
  \end{equation*}
  Analogously, if $g\in \mathcal{C}_0^{\infty }(\ran\mathbf{p})$, then
  $f\in \mathcal{C}^{\infty }_{0,cyl}\bigl(L^2\oplus L^2 \bigr)$, the
  cylindrical smooth functions with compact support.
\end{definition}
We remark that neither $\mathcal{S}_{cyl}\bigl(L^2\oplus L^2 \bigr)$
nor $\mathcal{C}^{\infty }_{0,cyl}\bigl(L^2\oplus L^2 \bigr)$ possess
a vector space structure. Finally, for cylindrical Schwartz functions
we define the Fourier transform:
\begin{equation*}
  \mathcal{F}[f](\eta )=\int_{\ran\mathbf{p}}^{}e^{-2\pi i\Re\langle \eta   , z \rangle_{L^2\oplus L^2}}f(z)  dL_{\mathbf{p}}(z)\; ,
\end{equation*}
where $dL_{\mathbf{p}}$ denotes integration with respect to the
Lebesgue measure on $\ran\mathbf{p}$. The inversion formula is then
\begin{equation*}
  f(z )=\int_{\ran\mathbf{p}}^{}e^{2\pi i\Re\langle \eta   , z \rangle_{L^2\oplus L^2}} \mathcal{F}[f](\eta )  dL_{\mathbf{p}}(\eta )\; .
\end{equation*}
With these definitions in mind, we can prove the following
lemma.
\begin{lemma}\label{lemma:19}
  Let
  $(\varrho _{\varepsilon })_{\varepsilon \in (0,\bar{\varepsilon}
    )}\subset \mathcal{L}^1(\mathcal{H})$
  be a family of normal states that satisfies Assumptions~\eqref{eq:84},
  \eqref{eq:arho} and~\eqref{eq:130};
  $\tilde{\mu}_t:\mathds{R}\to \mathfrak{P}\bigl(L^2\oplus L^2 \bigr)$ such that
  for any $t\in \mathds{R}$ ,
  $\tilde{\mu} _t\in \mathcal{M}\bigl(\tilde{\varrho}_{\varepsilon
  }(t),\varepsilon \in (0,\bar{\varepsilon} ) \bigr)$.
  If, in addition, $\tilde{\mu}_t$ satisfies the integral
  equation~\eqref{eq:114}, then the following statements are true:
  \begin{itemize}[label=\textcolor{myblue}*]
  \item For any $t\in \mathds{R}$, and for any
    $(t_i)_{i\in \mathds{R}}\subset \mathds{R}$ such that
    $\lim_{i\to \infty }t_i=t$,
    \begin{equation*}
      \tilde{\mu} _{t_i}\overset{n}{\rightharpoonup}\tilde{\mu}_t\; ;
    \end{equation*}
    i.e. $\tilde{\mu} _t$ is a weakly narrowly continuous map.
  \item The map $\tilde{\mu}_t $ solves the transport equation \footnote{Recall that
   $\mathbf{V}(t)(z)=-i\mathbf{E}_0(-t)\circ\partial_{\bar{z}}\bigl(\hat{\mathscr{E}}-\mathscr{E}_0
  \bigr)\circ\mathbf{E}_0(t)(z)$.}
    \begin{equation*}
      \partial _t\tilde{\mu}_t+ \nabla ^T\Bigl(\mathbf{V}(t)\tilde{\mu}_t \Bigr)=0
    \end{equation*}
    in the weak sense, i.e.
    \begin{equation}\label{eq:123}
      \forall f\in \mathcal{C}^{\infty }_{0,cyl}\Bigl(\mathds{R}\times \bigl(L^2\oplus L^2\bigr) \Bigr)\;,\; \int_{\mathds{R}}^{}\int_{L^2\oplus L^2}^{}\Bigl(\partial _tf +\Re\langle \nabla f  , \mathbf{V}(t) \rangle_{}\Bigr)  d\tilde{\mu}_t   dt=0\; .
    \end{equation}
  \end{itemize}
\end{lemma}
\begin{proof}
  Let $f\in \mathcal{S}_{cyl}\bigl(L^2\oplus L^2 \bigr)$. Fubini's
  theorem gives
  \begin{equation*}
    \int_{L^2\oplus L^2}^{}f(z)  d\tilde{\mu}_t(z)=\int_{\ran\mathbf{p}}^{}\mathcal{F}[f](\xi )\,\biggl(\int_{L^2\oplus L^2}^{}e^{2\pi i\Re\langle \xi  , z \rangle_{}}  d\tilde{\mu}_t(z) \biggr)\,  dL_{\ran\mathbf{p}}(\xi)\; ,
  \end{equation*}
  where $dL_{\ran\mathbf{p}}$ is the Lebesgue measure on ${\ran\mathbf{p}}$ and
  $   \mathcal{F}(f)(\xi)=\int_{\ran\mathbf{p}}^{}f(z)e^{-2\pi i \Re\langle \xi  , z \rangle}
  dL_{\ran\mathbf{p}}(z)$.
  Now we define
  $\tilde{G}_0(t,\xi ):=\int_{L^2\oplus L^2}^{}e^{2\pi i\Re\langle
    \xi , z \rangle_{}} d\tilde{\mu}_t(z)$. Hence Equation
  \eqref{eq:113} of Proposition \ref{prop:14} gives
\begin{equation}\label{eq:131}
    \tilde{G}_0(t,\xi )- \tilde{G}_0(s,\xi )=\int_{s}^{t}
 \biggl(\int_{L^2\oplus L^2}^{}\mathscr{B}_0(\tilde{\xi}_\tau )(z)e^{i
   \sqrt{2}\Re\langle \tilde{\xi}_\tau   ,
 z \rangle_{}}  d\mu_\tau(z)\biggr)d\tau\; ;
  \end{equation}
and this proves that $t\mapsto \tilde{G}_0(t,\xi )$ is continuous for
any $\xi\in L^2\oplus L^{2}$ since  the integrand in the right hand
side of \eqref{eq:131} is bounded with respect to $\tau$ by
Proposition \ref{prop:14}.  Remark that $\tilde{G}_0(t,\xi )$ is
bounded by one for any
$(t,\xi)\in\mathds{R}\times (L^{2}\oplus L^{2})$.
Therefore the map  $t\mapsto \int_{L^2\oplus L^2}^{}f(z) d\tilde{\mu}_t(z)$ is
  continuous for any
  $f\in \mathcal{S}_{cyl}\bigl(L^2\oplus L^2 \bigr)$. Finally, by an
  argument analogous to the one used at the beginning of the proof of
  Proposition~\ref{prop:14}, it is easy to prove that
  $\int_{L^2\oplus L^2}^{} \lVert z \rVert^2_{L^2\oplus
    L^2}d\tilde{\mu}_t(z) \in L^{\infty }_{t}(\mathds{R})$. In fact,
  we know that $\tilde{\mu}_t\Bigl(B_u(0,\sqrt{\mathfrak{C}})\cap
  Q(-\Delta +V)\oplus D(\omega ^{1/2})\Bigr)=1$ by
  Lemma~\ref{lemma:14}; and  if $z=(u,\alpha )$ then
  the functions $\alpha\mapsto\lVert \alpha \rVert_2^2\leq \lVert \alpha
  \rVert_{\mathcal{F}H^{1/2}}^2$,
  belong to $L_z^1\Bigl[L^2\oplus L^2,d\tilde{\mu}_t(z) \Bigr]$
  uniformly in $t$ by Lemmas~\ref{lemma:14} and~\ref{lemma:13}. Then
  it follows that $\tilde{\mu}_t $ is weakly narrowly continuous by
  \citep[][Lemma 5.1.12 -
  f]{AGS
  }, thus proving the first point.

  Now we prove the second point by a similar argument as in
  \cite{2011arXiv1111.5918A} which we reproduce here for completeness. Let
  $g\in \mathcal{C}_{0,cyl}^{\infty }\bigl(L^2\oplus L^2 \bigr)$; we
  integrate Equation~\eqref{eq:114} with respect to the measure
  $\mathcal{F}[g](\eta )dL_{\mathbf{p}}$ obtaining
  \begin{equation*}
    \begin{split}
          \int_{L^2\oplus L^2}^{}g(z)  d\tilde{\mu}_t(z)=\int_{L^2\oplus L^2}^{}g(z)  d\tilde{\mu}_0(z) +2\pi i\int_0^t\int_{\ran\mathbf{p}}^{}\biggl(\int_{L^2\oplus L^2}^{}\Re\langle \eta   , \mathbf{V}(s)(z) \rangle_{}  d\tilde{\mu}_s(z)\biggr)\\\mathcal{F}[g](\eta )  dL_{\mathbf{p}}(\eta )  ds\; .
    \end{split}
  \end{equation*}
  Let $\nabla g$ be the differential of
  $g:L^2\oplus L^2\to \mathds{R}$, where here $L^2\oplus L^2$ is
  considered as a real Hilbert space with scalar product
  $\Re\langle \cdot , \cdot \rangle_{L^2\oplus L^2}$. Then, by
  Fubini's theorem and the properties of the Fourier transform, we get
  \begin{equation*}
        \int_{L^2\oplus L^2}^{}g(z)  d\tilde{\mu}_t(z)=\int_{L^2\oplus L^2}^{}g(z)  d\tilde{\mu}_0(z) +\int_0^t\int_{L^2\oplus L^2}^{}\Re\langle \nabla g (z)   , \mathbf{V}(s)(z) \rangle_{}  d\tilde{\mu}_s(z)  ds\; .
  \end{equation*}
By Lebesgue's differentiation theorem (with respect to $t$), we obtain
  \begin{equation*}
    \partial _t \int_{L^2\oplus L^2}^{}g(z)  d\tilde{\mu}_t(z)- \int_{L^2\oplus L^2}^{}\Re\langle \nabla g (z)   , \mathbf{V}(t)(z) \rangle_{}  d\tilde{\mu}_t(z)=0\;.
  \end{equation*}
  Equation~\eqref{eq:123} is then obtained for
  $f(t,z)=\varphi (t)g(z)$, multiplying by
  $\varphi (t)\in \mathcal{C}_0^{\infty }(\mathds{R},\mathds{R})$,
  integrating with respect to $t$, and finally using integration by
  parts. The result for a generic
  $f\in \mathcal{C}_{0,cyl}^{\infty }\Bigl(\mathds{R}\times (L^2\oplus L^2)
  \Bigr)$
  follows immediately: $f(t,z)=g(t,\mathbf{p}z)$ for some
  $g\in \mathcal{C}_0^{\infty }\Bigl(\mathds{R}\times
  \ran\mathbf{p}\Bigr)$,
  and the latter can be approximated by a sequence
  $\Bigl(g_j(t ,\mathbf{p}z )\Bigr)_{j\in \mathds{N}}\subset
  \mathcal{C}_0^{\infty }(\mathds{R})\overset{alg}{\otimes
  }\mathcal{C}_0^{\infty }(\ran\mathbf{p})$.
\end{proof}
We need to check an hypothesis on the velocity vector field
$\mathbf{V}(t)$ (introduced in Corollary~\ref{cor:3}) to prove the
sought uniqueness result. This is done in the following lemma.
\begin{lemma}\label{lemma:20}
  Let
  $(\varrho _{\varepsilon })_{\varepsilon \in (0,\bar{\varepsilon}
    )}\subset \mathcal{L}^1(\mathcal{H})$
  be a family of normal states that satisfies
  Assumptions~\eqref{eq:84} and~\eqref{eq:arho};
  $\tilde{\mu}_t:\mathds{R}\to \mathfrak{P}\bigl(L^2\oplus L^2 \bigr)$ such that
  for any $t\in \mathds{R}$ ,
  $\tilde{\mu} _t\in \mathcal{M}\bigl(\tilde{\varrho}_{\varepsilon
  }(t),\varepsilon \in (0,\bar{\varepsilon} ) \bigr)$.
  Then
  $\lVert \mathbf{V}(t)(z) \rVert_{L^2\oplus L^2}^{} \in L^{\infty
  }_{t}\Bigl(\mathds{R}, L^1_{z}\bigl[L^2\oplus
  L^2,d\mu_t(z)\bigr]\Bigr)$,
  i.e. the norm of the velocity vector field is integrable with
  respect to $\tilde{\mu}_t $, uniformly in $t\in \mathds{R}$.
\end{lemma}
\begin{proof}
  By Equation~\eqref{eq:112} of Lemma~\ref{lemma:17} and the
  definition of $\mathbf{V}(t)$ we have that for any
  $\xi \in L^2\oplus L^2$:
\begin{equation*}
    \begin{split}
          \Bigl\lvert \Re\langle \xi   , \mathbf{V}(t)(z) \rangle_{}  \Bigr\rvert_{}^{}\leq C(\sigma _0)\lVert \xi   \rVert_{L^2\oplus L^2}^{}\Bigl(\lVert u  \rVert_2^{2}+ \lVert (-\Delta +V)^{1/2}u  \rVert_2^2+\lVert \alpha   \rVert_{\mathcal{F}H^{1/2}}^{2}+\lVert u  \rVert_2\cdot \lVert (-\Delta +V)^{1/2}u  \rVert_2^2\\
      +\lVert u  \rVert_2\cdot \lVert \alpha   \rVert_{\mathcal{F}H^{1/2}}^2 +\lVert u  \rVert_2\cdot \lVert (-\Delta +V)^{1/2}u  \rVert_2^{}\cdot  \lVert \alpha   \rVert_{\mathcal{F}H^{1/2}}\Bigr)\; .
    \end{split}
  \end{equation*}
  It is easy to prove an equivalent bound for the imaginary part, and
  hence obtain for any $\xi \in L^2\oplus L^2$:
\begin{equation*}
    \begin{split}
          \Bigl\lvert \langle \xi   , \mathbf{V}(t)(z) \rangle_{}  \Bigr\rvert_{}^{}\leq C(\sigma _0)\lVert \xi   \rVert_{L^2\oplus L^2}^{}\Bigl(\lVert u  \rVert_2^{2}+ \lVert (-\Delta +V)^{1/2}u  \rVert_2^2+\lVert \alpha   \rVert_{\mathcal{F}H^{1/2}}^{2}+\lVert u  \rVert_2\cdot \lVert (-\Delta +V)^{1/2}u  \rVert_2^2\\
      +\lVert u  \rVert_2\cdot \lVert \alpha   \rVert_{\mathcal{F}H^{1/2}}^2 +\lVert u  \rVert_2\cdot \lVert (-\Delta +V)^{1/2}u  \rVert_2^{}\cdot  \lVert \alpha   \rVert_{\mathcal{F}H^{1/2}}\Bigr)\; .
    \end{split}
  \end{equation*}
  Therefore it follows immediately that
  \begin{equation*}
    \begin{split}
          \lVert \mathbf{V}(t)(z) \rVert_{L^2\oplus L^2}^{}\leq C(\sigma _0)\Bigl(\lVert u  \rVert_2^{2}+ \lVert (-\Delta +V)^{1/2}u  \rVert_2^2+\lVert \alpha   \rVert_{\mathcal{F}H^{1/2}}^{2}+\lVert u  \rVert_2\cdot \lVert (-\Delta +V)^{1/2}u  \rVert_2^2\\+\lVert u  \rVert_2\cdot \lVert \alpha   \rVert_{\mathcal{F}H^{1/2}}^2 +\lVert u  \rVert_2\cdot \lVert (-\Delta +V)^{1/2}u  \rVert_2^{}\cdot  \lVert \alpha   \rVert_{\mathcal{F}H^{1/2}}\Bigr)\; .
    \end{split}
  \end{equation*}
  The right hand side of the above equation is in
  $L^{\infty }_{t}\Bigl(\mathds{R}, L^1_{z}\bigl[L^2\oplus
  L^2,d\mu_t(z)\bigr]\Bigr)$,
  as shown at the beginning of the proof of Proposition~\ref{prop:14}.
\end{proof}
As showed in
\citep[Theorem 3.1.1]{liard:2015
}, Lemma~\ref{lemma:19}, Lemma~\ref{lemma:20} and Lemma~\ref{lemma:14} $(ii)$-$(iii)$  are sufficient to
prove uniqueness of the solution to the transport equation~\eqref{eq:123}, as
precisely formulated in the following proposition.
\begin{proposition}\label{prop:16}
  Let
  $(\varrho _{\varepsilon })_{\varepsilon \in (0,\bar{\varepsilon}
    )}\subset \mathcal{L}^1(\mathcal{H})$
  be a family of normal states that satisfies
  Assumptions~\eqref{eq:84},\eqref{eq:arho}  and~\eqref{eq:130}. In addition, let
  $\tilde{\mu}_t:\mathds{R}\to \mathfrak{P}\bigl(L^2\oplus L^2 \bigr)$ such that
  for any $t\in \mathds{R}$ ,
  $\tilde{\mu} _t\in \mathcal{M}\bigl(\tilde{\varrho}_{\varepsilon_k
  }(t),k \in \mathds{N} \bigr)$ for some sequence $(\varepsilon_k)_{k\in\mathds{N}}$ with
  $\varepsilon_k\to 0$\;   and $\tilde{\mu}_t$ satisfies the integral
  equation~\eqref{eq:114}. Then $\tilde{\mu}_t= \mathbf{E}_0(-t)_{\#} \hat{\mathbf{E}}(t)_{\#}\mu
  _0$ .
\end{proposition}

\subsection{The classical limit of the dressing transformation.}
\label{sec:class-limit-dress}

Let's consider now the dressing transformation
$U_{\infty }(\theta )=e^{-i\frac{\theta }{\varepsilon} T_{\infty }}$ on $\mathcal{H}$,
with self-adjoint generator:
\begin{gather*}
  T_{\infty }=\bigl(\mathscr{D}_{g_{\infty }}\bigr)^{Wick}=\int_{\mathds{R}^3}^{}\psi^{*}(x)  \Bigl(a^{*}(g_{\infty }e^{-ik\cdot x})+a(g_{\infty }e^{-ik\cdot x})\Bigr)\psi(x)  dx\; ;\\
  g_{\infty }(k)=-\frac{i}{(2\pi)^{3/2}}\frac{1}{\sqrt{2\omega(k)}}\frac{1-\chi_{\sigma_0}(k)}{\frac{k^2}{2M}+\omega(k)}\in L^2 (\mathds{R}^3  )\; .
\end{gather*}
The family
$\bigl(e^{-i\frac{\theta }{\varepsilon} T_{\infty }}\bigr)_{\theta \in
  \mathds{R}}\subset \mathcal{L}(\mathcal{H})$
is a strongly continuous unitary group, and therefore can be seen as a
dynamical system acting on quantum states. Therefore, given a family
$(\varrho_{\varepsilon } )_{\varepsilon \in (0,\bar{\varepsilon} )}$
of normal quantum states on $\mathcal{H}$, we could determine the
Wigner measures of
\begin{equation}
  \label{eq:115}
  \hat{\varrho} _{\varepsilon }(\theta )=e^{-i\frac{\theta }{\varepsilon} T_{\infty }}\,\varrho _{\varepsilon }\,e^{i\frac{\theta }{\varepsilon} T_{\infty }}\; .
\end{equation}
Since $T_{\infty }=\bigl(\mathscr{D}_{g_{\infty }}\bigr)^{Wick}$,
where $\mathscr{D}_{g_{\infty }}$ is the classical dressing generator
defined in Section~\ref{sec:classical-dressing}, we expect that under
suitable assumptions,
$\Bigl(\varrho _{\varepsilon _k}\rightarrow \mu \Rightarrow
\hat{\varrho }_{\varepsilon _k}(\theta )\rightarrow
\mathbf{D}_{g_{\infty }}(\theta )_{\#}\mu \Bigr)$,
where $\mathbf{D}_{g_{\infty }}(\theta )$ is the classical dressing
transformation.  The last assertion is indeed true, as explained in
the following. Observe that the dressing generator $T_{\infty }$ is
equal to the interaction part $H_I(\sigma )$ of the Nelson model with
cutoff, where $\tfrac{\chi _{\sigma } }{\sqrt{2\omega }}$ is replaced
by $g_{\infty }$, i.e.
$T_{\infty }=H_I(\sigma )\bigr\rvert_{\frac{\chi _{\sigma
    }}{\sqrt{2\omega }}=g_{\infty }}$.
The classical limit of the Nelson model with cutoff has been treated
by the authors in
\citep{Ammari:2014aa
}, thus the results below can be immediately deduced by the results in
\citep[][$d=3$, $H_0=0$ and
$\tfrac{\chi }{\sqrt{\omega }}=g_{\infty
}$]{Ammari:2014aa
}. We recall also that $g_{\infty }$, and therefore also $T_{\infty }$
and $\mathscr{D}_{g_{\infty }}$, depends on
$\sigma _0\in \mathds{R}_+$.
\begin{lemma}\label{lemma:18}
  Let
  $(\varrho_{\varepsilon } )_{\varepsilon \in (0,\bar{\varepsilon} )}$
  be a family of normal (quantum) states on $\mathcal{H}$ that satisfies
  Assumptions~\eqref{eq:84} and~\eqref{eq:86}. Then for any
  $\sigma _0\in \mathds{R}_+$ ,
  $(\hat{\varrho} _{\varepsilon }(-1 ))_{\varepsilon \in
    (0,\bar{\varepsilon} )}$
  satisfies Assumptions~\eqref{eq:84} and~\eqref{eq:arho}.
\end{lemma}
\begin{proposition}\label{prop:15}
  Let
  $\mathbf{D}_{g_{\infty }}:\mathds{R}\times Q(-\Delta +V)\oplus
  \mathcal{F}H^{1/2}\to Q(-\Delta +V)\oplus \mathcal{F}H^{1/2}$
  be the classical dressing transformation. Let
  $(\varrho_{\varepsilon } )_{\varepsilon \in (0,\bar{\varepsilon} )}$
  be a family of normal quantum states on $\mathcal{H}$ that satisfies
  Assumption~\eqref{eq:84} and Assumption~\eqref{eq:86}
  or~\eqref{eq:arho}. Then
  $\mathcal{M}\bigl(\varrho _{\varepsilon }, \varepsilon \in
  (0,\bar{\varepsilon} )\bigr)\neq \emptyset $;
  and for any $\sigma _0\in \mathds{R}_+$ and $\theta \in \mathds{R}$
  ,
  \begin{equation}
    \label{eq:116}
    \mathcal{M}\Bigl(\hat{\varrho} _{\varepsilon}(\theta ), \varepsilon \in(0,\bar{\varepsilon} )\Bigr)=\Bigl\{\mathbf{D}_{g_{\infty }}(\theta )_{\#}\mu \, ,\, \mu \in \mathcal{M}\bigl(\varrho _{\varepsilon },\varepsilon \in (0,\bar{\varepsilon} )\bigr)\Bigr\}\; .
  \end{equation}
  Furthermore, let
  $(\varepsilon _k)_{k\in \mathds{N}}\subset (0,\bar{\varepsilon} )$
  be a sequence such that
  $\lim_{k\to \infty }{\varepsilon _k}=0$. Then the following
  statement is true:
  \begin{equation}
    \label{eq:117}
    \varrho _{\varepsilon _k}\rightarrow \mu \Leftrightarrow\biggl(\forall \theta \in \mathds{R}\; ,\; \forall \sigma _0\in \mathds{R}_+\; ,\; \hat{\varrho }_{\varepsilon _k}(\theta )\rightarrow\mathbf{D}_{g_{\infty }}(\theta )_{\#}\mu \biggr)\; .
  \end{equation}
\end{proposition}

\subsection{Overview of the results: linking the dressed and undressed
  systems.}
\label{sec:link-dress-dynam}

Since as discussed in the previous subsection we can treat the
dressing as a dynamical transformation with its own ``time'' parameter
$\theta $; we are able to link the classical limit of the dressed and
undressed quantum dynamics via the classical dressing. In this way we
are able to recover the expected classical S-KG dynamics for the
undressed dynamics, and finally prove Theorem~\ref{thm:2}.

First of all, we  put together the results proved from
Section~\ref{sec:integral-formula} to
Section~\ref{sec:uniq-solut-transp} on the renormalized dressed
dynamics and remove the Assumption \eqref{eq:130} with the help of an
approximation argument worked out in \cite{MR2802894}.  This is done in the following theorem.
\begin{thm}\label{thm:3}
  Let
  $\hat{\mathbf{E}}:\mathds{R}\times Q(-\Delta +V)\oplus
  \mathcal{F}H^{1/2}\to Q(-\Delta +V)\oplus \mathcal{F}H^{1/2}$
  be the dressed S-KG flow associated to $\hat{\mathscr{E}}$. Let
  $(\varrho _{\varepsilon })_{\varepsilon \in (0,\bar{\varepsilon} )}$
  be a family of normal states in $\mathcal{H}$ that satisfies
  Assumptions~\eqref{eq:84} and~\eqref{eq:arho}. Then for any 
  $\sigma _0\geq2K(\mathfrak{C}+1+\bar\varepsilon)$ the dynamics
  $e^{- i\frac{t}{\varepsilon }\hat{H}_{ren}(\sigma _0)}$ is
  non-trivial  on every relevant sector with fixed nucleons of the state $\varrho _{\varepsilon }$;   $\mathcal{M}\bigl(\varrho _{\varepsilon }, \varepsilon \in
  (0,\bar{\varepsilon} )\bigr)\neq \emptyset $
  ; and for any $t\in \mathds{R}$
  \begin{equation}
    \label{eq:118}
    \mathcal{M}\Bigl(e^{- i\frac{t}{\varepsilon } \hat{H}_{ren}(\sigma _0)}\varrho _{\varepsilon}e^{i\frac{t}{\varepsilon }\hat{H}_{ren}(\sigma _0)}, \varepsilon \in(0,\bar{\varepsilon} )\Bigr)=\Bigl\{\hat{\mathbf{E}}(t)_{\#}\mu \, ,\, \mu \in \mathcal{M}\bigl(\varrho _{\varepsilon },\varepsilon \in (0,\bar{\varepsilon} )\bigr)\Bigr\}\; .
  \end{equation}
  Furthermore, let
  $(\varepsilon _k)_{k\in \mathds{N}}\subset (0,\bar{\varepsilon} )$
  be a sequence such that $\lim_{k\to \infty }\varepsilon _k= 0$.
  Then the following statement is true:
  \begin{equation}
\label{eq:119}
\varrho_{\varepsilon _k}\rightarrow \mu \Leftrightarrow\biggl(\forall t\in \mathds{R}\;,\; e^{- i\frac{t}{\varepsilon_k} \hat{H}_{ren}(\sigma _0)}\varrho_{\varepsilon_k}e^{i\frac{t}{\varepsilon_k } \hat{H}_{ren}(\sigma_0)}\rightarrow \hat{\mathbf{E}}(t)_{\#}\mu \biggr)\; .
  \end{equation}
\end{thm}
\begin{proof}
Thanks to the argument briefly sketched below, we no longer need Assumption \eqref{eq:130}.
 Let $\chi\in C^{\infty}_{0}(\mathds{R})$ such that $0\leq
\chi\leq 1$\,, $\chi\equiv 1$ in a neighbourhood of $0$\, and
$\chi_R(x)=\chi(\frac x R)$. The approximation
$$
\varrho_{\varepsilon,R}=\frac{\chi_{R}(\hat H_{ren}(\sigma_{0}))\;
\varrho_{\varepsilon} \;\chi_{R}(\hat H_{ren}(\sigma_{0}))}{\Tr\left[\chi_{R}(\hat H_{ren}(\sigma_{0}))\;
\varrho_{\varepsilon}\; \chi_{R}(\hat H_{ren}(\sigma_{0}))
\right]}\;,
$$
satisfies the Assumptions~\eqref{eq:84}, \eqref{eq:arho}, \eqref{eq:130} and the property 
\begin{equation*}
||e^{-i\frac{t}{\varepsilon} \hat H_{ren}(\sigma_{0})} \left(
\varrho_{\varepsilon}-\varrho_{\varepsilon,R} \right) e^{i\frac{t}{\varepsilon}\hat H_{ren}(\sigma_{0})}
||_{\mathcal{L}^1(\mathscr{H})}^{}=||\varrho_{\varepsilon}-
\varrho_{\varepsilon,R}||_{\mathcal{L}^{1}(\mathcal{H})} \leq \nu(R)\,,
\end{equation*}
where  $\nu(R)$ is independent of $\varepsilon$ and $\lim_{R\to \infty}\nu(R)=0$\,. The last claim follows by Assumption \eqref{eq:arho}, Theorem \ref{thm:1} and
Definition \ref{def:7}.  Up to
extracting a sequence which a priori depends on $R$ and $t$, we can suppose
that $\mathscr{M}(\varrho_{\varepsilon_{n},R},
n\in\mathds{N})=\left\{\mu_{0,R}\right\}\,$, 
$\mathscr{M}(\varrho_{\varepsilon_{n}},
n\in\mathds{N})=\left\{\mu_{0}\right\}\,$ and $\mathcal{M}(\varrho_{\varepsilon_{n}}(t), n\in\mathds{N})=\left\{\mu_t\right\}$.  In particular, applying Proposition \ref{prop:16} we obtain
$$
\mathcal{M}(e^{-i\frac{t}{\varepsilon_{n}} \hat H_{ren}(\sigma_{0})}\varrho_{\varepsilon_{n},R} 
e^{i\frac{t}{\varepsilon_{n}} \hat H_{ren}(\sigma_{0})}, n\in\mathds{N})=\left\{\hat{\mathbf{E}}(t)_{\#}\mu_{0,R}\right\}\,.
$$
A general estimate proved in \cite[Proposition~2.10 ]{MR2802894} compares 
the total variation distance of Wigner (probability) measures with the trace distance of their
associated quantum states. In our case, it implies  
\begin{eqnarray*}
  &&\int_{L^2\oplus L^2}|\mu_{t}-\hat{\mathbf{E}}(t)_{\#}\mu_{0,R}|\leq
\liminf_{n\to\infty}
  ||e^{-i\frac{t}{\varepsilon_{n}} \hat H_{ren}(\sigma_{0})}
  \left(\varrho_{\varepsilon_n}-\varrho_{\varepsilon_n,R}\right)e^{i\frac{t}{\varepsilon_{n}} \hat H_{ren}(\sigma_{0})}
  ||_{\mathcal{L}^{1}(\mathcal{H})}\leq\nu(R)\,, \\
  &&\int_{L^2\oplus L^2}|\mu_{0}-\mu_{0,R}|\leq \liminf_{n\to\infty}
  ||\varrho_{\varepsilon_n}-\varrho_{\varepsilon_n,R}
  ||_{\mathcal{L}^{1}(\mathcal{H})}\leq\nu(R)\,,
\end{eqnarray*}
where the left hand side denotes the total variation of the signed measures
$\mu_t-\hat{\mathbf{E}}(t)_{\#}\mu_{0,R}$ and $\mu_0-\mu_{0,R}$ respectively. Hence by the triangle inequality, we obtain
$$
\int_{L^2\oplus L^2}|\mu_t-\hat{\mathbf{E}}(t)_{\#}\mu_{0}|\leq \int_{L^2\oplus L^2}|\mu_t-\hat{\mathbf{E}}(t)_{\#}\mu_{0,R}|+\int_{L^2\oplus L^2} |\mu_{0,R}-\mu_{0}| \leq 2\nu(R)\,.
$$
This proves that 
$$\left\{\hat{\mathbf{E}}(t)_{\#}\mu_{0}\right\} \subset \mathcal{M}(e^{-i\frac{t}{\varepsilon_{n}} \hat H_{ren}(\sigma_{0})}\varrho_{\varepsilon_{n}}
e^{i\frac{t}{\varepsilon_{n}} \hat H_{ren}(\sigma_{0})}, n\in\mathds{N})\,,$$\, 
By reversing time and utilizing the analogue inclusion above, we prove  \eqref{eq:119}.
\end{proof}

\bigskip
{\bf Proof of Theorem \ref{thm:2}}: Observe that using the definition of the renormalized dressed
evolution $\varrho_{\varepsilon } (t)$ (Definition~\ref{def:10}) and
the definition of the ``dressing dynamics''
$\hat{\varrho}_{\varepsilon } (\theta )$ (Equation~\eqref{eq:115}), we
obtain:
\begin{equation*}
  \label{eq:120}
  e^{-i \frac{t}{\varepsilon }H_{ren}(\sigma _0)}\,\varrho _{\varepsilon }\,e^{i \frac{t}{\varepsilon }H_{ren}(\sigma _0)}=e^{-\frac{i}{\varepsilon }T_{\infty }}e^{-i \frac{t}{\varepsilon }\hat{H}_{ren}(\sigma _0)}e^{\frac{i}{\varepsilon }T_{\infty }}\,\varrho _{\varepsilon }\,e^{-\frac{i}{\varepsilon }T_{\infty }}e^{i \frac{t}{\varepsilon }\hat{H}_{ren}(\sigma _0)}e^{\frac{i}{\varepsilon }T_{\infty }}=\Bigl(\bigl(\hat{\varrho}_{\varepsilon }(-1) \bigr)(t)\Bigr)\hat{\phantom{\bigr)}} (1)\; .
\end{equation*}
Let
$(\varrho _{\varepsilon })_{\varepsilon \in (0,\bar{\varepsilon} )}$
be a family of normal states in $\mathcal{H}$ that satisfies
Assumptions~\eqref{eq:84} and~\eqref{eq:86}. In addition, as usual, let
$(\varepsilon _k)_{k\in \mathds{N}}\subset (0,\bar{\varepsilon} )$ be
a sequence such that $\lim_{k\to \infty }\varepsilon _k= 0$. Then we
can use Lemma~\ref{lemma:18}, Proposition~\ref{prop:15} and
Theorem~\ref{thm:3} to prove the following statement:
\begin{equation*}
  \label{eq:121}
  \begin{split}
  \varrho_{\varepsilon _k}\rightarrow \mu \Leftrightarrow\biggl(\forall t\in \mathds{R}\;,\; e^{- i\frac{t}{\varepsilon_k} H_{ren}(\sigma _0)}\varrho_{\varepsilon_k}e^{i\frac{t}{\varepsilon_k } H_{ren}(\sigma_0)}\rightarrow \mathbf{D}_{g_{\infty }}(1)_{\#}\hat{\mathbf{E}}(t)_{\#}\mathbf{D}_{g_{\infty }}(-1)_{\#}\mu\\=[\mathbf{D}_{g_{\infty }}(1)\circ\hat{\mathbf{E}}(t)\circ\mathbf{D}_{g_{\infty }}(-1)]_{\#}\mu \biggr)\; .
  \end{split}
\end{equation*}
Therefore Theorem~\ref{thm:2} is proved, since by
Equation~\eqref{eq:69} of Theorem~\ref{prop:9},
$\mathbf{D}_{g_{\infty
  }}(1)\circ\hat{\mathbf{E}}(t)\circ\mathbf{D}_{g_{\infty
  }}(-1)=\mathbf{E}(t)$.
To be more precise, we use the following chain of inferences:
\begin{equation*}
  \begin{split}
    \Bigl(\varrho_{\varepsilon _k}\rightarrow \mu\Bigr)\overset{\substack{\text{Lem.~\ref{lemma:18}}\\\text{Prop.~\ref{prop:15}}}}{\Longrightarrow}\biggl(\forall \sigma _0\in \mathds{R}_+\;,\; \hat{\varrho}_{\varepsilon _k}(-1)\rightarrow \mathbf{D}_{g_{\infty }}(-1)_{\#}\mu \; \text{and $\bigl(\hat{\varrho}_{\varepsilon_k }(-1)\bigr)_{k\in \mathds{N}} $}\\\text{satisfies Ass.~\eqref{eq:84},~\eqref{eq:arho}}\biggr)\\\overset{\substack{\text{Thm.~\ref{thm:3}}\\\text{Lem.~\ref{lemma:13}}}}{\Longrightarrow}\biggl(\exists  \sigma _0\in \mathds{R}_+\;,\; \forall t\in \mathds{R}\; ,\; \bigl(\hat{\varrho}_{\varepsilon _k}(-1)\bigr)(t)\rightarrow \hat{\mathbf{E}}(t)_{\#}\mathbf{D}_{g_{\infty }}(-1)_{\#}\mu \; \\\text{and $\Bigl(\bigl(\hat{\varrho}_{\varepsilon_k }(-1)\bigr)(t)\Bigr)_{k\in \mathds{N}} $ satisfies Ass.~\eqref{eq:84},~\eqref{eq:arho}}\biggr)\\\overset{\text{Prop.~\ref{prop:15}}}{\Longrightarrow}\biggl(\forall t\in \mathds{R}\;,\; \Bigl(\bigl(\hat{\varrho}_{\varepsilon _k}(-1)\bigr)(t)\Bigr) \hat{\phantom{\bigr)}} (1)\rightarrow \mathbf{D}_{g_{\infty }}(1)_{\#}\hat{\mathbf{E}}(t)_{\#}\mathbf{D}_{g_{\infty }}(-1)_{\#}\mu \biggr)\\
\overset{\text{Thm.~\ref{prop:9}}}{\Longrightarrow}\biggl(\forall t\in \mathds{R}\;,\; e^{- i\frac{t}{\varepsilon} H_{ren}(\sigma _0)}\varrho_{\varepsilon_k}e^{i\frac{t}{\varepsilon } H_{ren}(\sigma_0)}\rightarrow \mathbf{E}(t)_{\#}\mu\biggr)\; .
  \end{split}
\end{equation*}
The inference in the opposite sense is trivial.

As it has become evident with the above discussion, we do not prove
Theorem~\ref{thm:2} directly; and it would be very difficult to do so,
due to the fact that we do not know the explicit form of the generator
$H_{ren}(\sigma _0)$ of the undressed dynamics. We know instead how
the dressed generator $\hat{H}_{ren}(\sigma _0)$ acts as a quadratic
form, and that is sufficient to characterize its dynamics in the
classical limit, and obtain the results of Theorem~\ref{thm:3}. The
properties of the dressing transformation and of its classical
counterpart are then crucial to translate the results on the dressed
dynamics to the corresponding results on the undressed one.

\appendix
\section{Uniform higher-order estimate}
\label{sec:higher order}
We prove in this section a higher-order estimate that bounds the meson
number operator $N_2$ by the dressed Hamiltonian $\hat H_\sigma^{(n)}$ uniformly
with respect to the effective (semiclassical) parameter $\varepsilon$ and the cut-off parameter
$\sigma$. Such type
of estimates rely on the pull-through formula and they are known for the $P(\varphi)_2$ model \cite{Ro2} and for the Nelson model \cite{MR1809881}. However, since the
dependence of the dressed Hamiltonian $\hat H_\sigma^{(n)}$ on  $\varepsilon$ is
somewhat nontrivial, we briefly indicate in this appendix how to obtain an uniform estimate.

\begin{lemma}\label{appA.1}
For any $\varepsilon\in(0,\bar\varepsilon)$ and any $\psi\in D(N_2)\subset\mathcal{H}$,
\begin{equation*}
\bigl\lVert N_2 \psi \bigr\rVert^2=\int_{\mathds{R}^3} \bigl\lVert(N_2+\varepsilon)^{\frac{1}{2}} a(k) \psi \bigr\rVert^2 \;
dk\,.
\end{equation*}
\end{lemma}
\begin{proof}
 Recall that $N_2$ and $a(k)$ depends in the parameter $\varepsilon$ according to the
 notations of Subsection \ref{sec:notat-defin}. Taking care of domain issues as in
 \cite[Lemma 2.1]{MR1809881} one proves
 \begin{eqnarray*}
 \bigl\lVert N_2 \psi\bigr\rVert^2&=& \bigl\langle N_2^{\frac{1}{2}} \psi, \int_{\mathds{R}^3}
  a^*(k) a(k) \,dk\, N_2^{\frac{1}{2}} \psi\bigr\rangle
 = \int_{\mathds{R}^3} \bigl\lVert  a(k) N_2^{\frac{1}{2}} \psi \bigr\rVert^2\,dk
 = \int_{\mathds{R}^3} \bigl\lVert  (N_2+\varepsilon)^{\frac{1}{2}} a(k)  \psi \bigr\rVert^2\,dk\,.
\end{eqnarray*}
\end{proof}
Recall that the interaction term $\hat H_I(\sigma)^{(n)}$ is given by
\eqref{eq:9}. A simple computation yields
\begin{eqnarray*}
  [a(k), \hat H_I(\sigma)^{(n)}]
 &=& \varepsilon^2 \Biggl [ \sum_{j=1}^n  \frac{1}{2\sqrt{(2\pi)^3}} \frac{\chi_\sigma(k)}{\sqrt{\omega(k)}} e^{-i k.x_j}
+\frac{1}{M}  \sum_{j=1}^n  r_\sigma(k) e^{-ik\cdot x_j} a^*(r_\sigma
e^{-ik\cdot x_j}) \\&&\hspace{.5in}+ r_\sigma(k) e^{-ik\cdot x_j} a(r_\sigma
e^{-ik\cdot x_j})- r_\sigma(k) e^{-ik\cdot x_j} D_{x_j}\Biggr ]\,.
\end{eqnarray*}
\begin{lemma} \label{appA2}
For any $\mathfrak{C}>0$ and $\sigma_0\geq 2K(\mathfrak{C}+1+\bar\varepsilon)$ there exist $c,b>0$ such that for any $\varepsilon\in(0,\bar \varepsilon)$, $\sigma_0<\sigma\leq +\infty$ and
$n\in\mathds{N}$ such that $n\varepsilon\leq \mathfrak{C}$, we have
\begin{equation*}
\bigl\lVert (b+\varepsilon \omega(k)+\hat H_\sigma^{(n)})^{-\frac{1}{2}}
[a(k),\hat H_I(\sigma)^{(n)}] (b+ \hat H_\sigma^{(n)})^{-\frac{1}{2}} \bigr\rVert
\leq c \left( |\tfrac{\chi_\sigma(k)}{\sqrt{\omega(k)}}|+ |r_\sigma(k)| \omega(k)^{-1/4}\right)\,.
\end{equation*}
\end{lemma}
\begin{proof}
According to Proposition \ref{prop:2} and Theorem \ref{thm:1},  $\hat H_I(\sigma)^{(n)}$ is $H_0^{(n)}$-form bounded with small bound that is uniform with respect to $\varepsilon\in(0,\bar \varepsilon)$, $\sigma_0<\sigma\leq +\infty$ and $n\in\mathds{N}$ such that $n\varepsilon\leq \mathfrak{C}$. Hence  $(H_0^{(n)})^{\frac{1}{2}} (b+\hat H_\sigma^{(n)})^{-\frac{1}{2}}$ is uniformly bounded
 for some $b>0$. So it is enough to prove the claimed  bound  with
 $H_0^{(n)}$ instead of $\hat H_\sigma^{(n)}$. Now using similar estimates as in Lemma \ref{lemma:3} and the fact  that $\sqrt{\varepsilon \omega(k)} (b+\varepsilon \omega(k)+\hat H_\sigma^{(n)})^{-\frac{1}{2}} $ is uniformly bounded  one correctly bounds all the terms of the commutator except the one with  $a^*$. Remark that the commutator contains the power  $\varepsilon^2$ that controls the sum over $1\leq j\leq n$ and the factor $1/\sqrt{\varepsilon \omega(k)}$.  In order to bound the term with $a^*$, one uses the type of estimate in
 \cite[Lemma 3.3 (ii)]{MR1809881} with $s=1/2$. Remark that one gets an $\varepsilon$-dependent estimate from \cite[Lemma 3.3 (ii)]{MR1809881} by noticing that $\varepsilon^{1/4}(H_0^{(n)}+1)^{-1/4}  (d\Gamma_1(\omega)+1)^{1/4}$ and
 $\varepsilon^{1/4}(N_2+1)^{-1/4}  (d\Gamma_1(1)+1)^{1/4}$  are uniformly bounded\footnote{$d\Gamma_1(\cdot)$ is the $\varepsilon$-independent
  second quantization operator in \cite{MR1809881}} and that $a^*$ contains $\sqrt{\varepsilon}$ that cancels the latter $\varepsilon^{-1/4} \cdot \varepsilon^{-1/4}$.
\end{proof}

Let $\mathfrak{C}>0$ and $\sigma_0 \geq 2K(\mathfrak{C}+1+\bar\varepsilon)$ as in the above lemma. In particular
$\hat H_\sigma^{(n)}$ is a self-adjoint operator for any $\varepsilon\in(0,\bar \varepsilon)$, $\sigma_0<\sigma\leq +\infty$ and $n\in\mathds{N}$ such that $n\varepsilon\leq \mathfrak{C}$.

\begin{lemma}[The pull-through formula]
\label{appA3}
The following identity holds true  for some $b<0$, any $\phi\in D(N_2^{\frac{1}{2}})\cap\mathcal{H}_n$ and $k$ almost everywhere in $\mathds{R}^3$,
\begin{eqnarray*}
a(k) (b-\hat H_\sigma^{(n)})^{-1}\phi&=&(b-\varepsilon\omega(k)-\hat H_\sigma^{(n)})^{-1} a(k)\phi
\\&&+(b-\varepsilon\omega(k)-\hat H_\sigma^{(n)})^{-1} \,[a(k),\hat H_I(\sigma)^{(n)}] \,(b-\hat H_\sigma^{(n)})^{-1}\phi\,.
\end{eqnarray*}
\end{lemma}
\begin{proof}
According to \cite[Lemma 4.4]{MR1809881} there exists
 $\psi \in (H_0^{(n)}+1)^{-1} D(N^{\frac{1}{2}})$ such that $\phi=(b-\hat H_\sigma^{(n)})
 \psi$ for some $b<0$. So the claimed formula is equivalent to
 \begin{equation*}
 (b-\varepsilon\omega(k)-\hat H_\sigma^{(n)}) a(k) \psi=a(k) (b-\hat H_\sigma^{(n)})
 \psi
 + \,[a(k),\hat H_I(\sigma)^{(n)}] \psi\,.
 \end{equation*}
The latter identity follows by a simple computation.
\end{proof}
\begin{proposition}
\label{prop:A1}
For any $\mathfrak{C}>0$ and $\sigma_0\geq 2 K (\mathfrak{C}+1+\bar\varepsilon)$ there exist $c,b>0$ such that the operator $\hat H_\sigma^{(n)}$ is  self-adjoint and the  following bound holds true:
\begin{equation*}
\bigl\lVert N_2 \psi\bigr\rVert\leq c \bigl\lVert (\hat H_\sigma^{(n)}+b) \psi\bigr\rVert \,, \quad \forall \psi\in D(\hat H_\sigma^{(n)})\,,
\end{equation*}
for any $\varepsilon\in(0,\bar\varepsilon),\sigma\in (\sigma_0,+\infty], n\in\mathds{N}$ such that $n\varepsilon\leq\mathfrak{C}$.
\end{proposition}
\begin{proof}
The operator $\hat H_\sigma^{(n)}$ is uniformly bounded from below. So by choosing
$b>0$ large enough one can take $\psi=(-b-\hat H_\sigma^{(n)})^{-1} \phi$. Now it is enough to prove the estimate for $\phi\in (H_0^{(n)}+1)^{-1/2} D(N_2^{\frac{1}{2}})$. Using Lemma \ref{appA.1} and Lemma \ref{appA3},
\begin{eqnarray}
\nonumber \bigl\lVert N_2 \psi \bigr\rVert^2&=&\int_{\mathds{R}^3} \bigl\lVert(N_2+\varepsilon)^{\frac{1}{2}} a(k) (b+\hat H_\sigma^{(n)})^{-1} \phi \bigr\rVert^2 \;
dk\\
&\leq& 2 \int_{\mathds{R}^3} \bigl\lVert (N_2+\varepsilon)^{\frac{1}{2}}
(b+\varepsilon \omega(k)+\hat H_\sigma^{(n)})^{-1}a(k) \phi \bigr\rVert^2 \;
dk \label{appeq1}\\&&+ 2\int_{\mathds{R}^3} \bigl\lVert (N_2+\varepsilon)^{\frac{1}{2}} (b+\varepsilon \omega(k)+\hat H_\sigma^{(n)})^{-1}
[a(k),\hat H_I(\sigma)^{(n)}]\, (b+\hat H_\sigma^{(n)})^{-1}\phi \bigr\rVert^2 \;
dk\,.\label{appeq2}
\end{eqnarray}
Since  $(N_2+\varepsilon)^{\frac{1}{2}} (b+\varepsilon \omega(k)+\hat H_\sigma^{(n)})^{-1/2} $
is uniformly bounded, by Lemma \ref{appA2} one shows
\begin{equation*}
\eqref{appeq2}\leq c \int_{\mathds{R}^3}
|\tfrac{\chi_\sigma(k)}{\sqrt{\omega(k)}}|+ |r_\sigma(k)| \omega(k)^{-1/4} \; dk
\; \cdot \;\bigl\lVert (b+\hat H_\sigma^{(n)})^{-1/2} \phi \bigr\rVert^2
\end{equation*}
For simplicity we denote by $c$ any constant. In the same way, one also shows
\begin{eqnarray*}
\eqref{appeq1}&\leq& c \int_{\mathds{R}^3}
\bigl\lVert (b+\varepsilon \omega(k)+\hat H_\sigma^{(n)})^{-1/2}a(k) \phi \bigr\rVert^2 \;
dk \\
&\leq& c \int_{\mathds{R}^3}
\bigl\lVert (b+\varepsilon \omega(k)+ H_0^{(n)})^{-1/2}a(k) \phi \bigr\rVert^2 \;
dk =c \bigl\lVert N_2^{1/2} (b+ H_0^{(n)})^{-1/2}\phi \bigr\rVert^2\,.
\end{eqnarray*}
 The last equality follows by a similar argument as in the proof of Lemma \ref{appA.1}.  Hence, one obtains
\begin{eqnarray*}
\bigl\lVert N_2 \psi \bigr\rVert^2&\leq& c \left( \bigl\lVert \phi \bigr\rVert^2+ \bigl\lVert (b+H_0^{(n)})^{-1/2} \phi \bigr\rVert^2\right)
=c \left( \bigl\lVert (b+\hat H_\sigma^{(n)})\psi \bigr\rVert^2+ \bigl\lVert (b+H_\sigma^{(n)})^{1/2} \psi \bigr\rVert^2\right)\\
&\leq& c \bigl\lVert (b+\hat H_\sigma^{(n)})\psi \bigr\rVert^2 \,.
\end{eqnarray*}
The last inequality is a consequence of the uniform boundedness of the operator  $(b+H_0^{(n)})^{-1/2} (b+\hat H_\sigma^{(n)})^{-1/2}$ with respect to $\varepsilon, \sigma$ and
 $n\in\mathds{N}$ such that $n\varepsilon\leq\mathfrak{C}$.
\end{proof}

\bibliography{bib}
\end{document}